%% file: main.tex
\documentclass[12pt, reqno]{amsart}
\usepackage{main}
\graphicspath{{figures/}}
\usepackage{titletoc}
\usepackage[style=apa, backend=biber, natbib=true, uniquename=false]{biblatex}
\newgeometry{margin=1.25in}
\hypersetup{
  pdftitle={Reliability-Targeted Simulation of Item Response Data: Solving the Inverse Design Problem},
  pdfauthor={JoonHo Lee},
  pdfsubject={Item response theory; simulation design; information-based reliability; inverse design},
  pdfkeywords={item response theory, reliability-targeted simulation, inverse design problem, test information, empirical quadrature calibration, stochastic approximation calibration}
}

\title{\Large Reliability-Targeted Simulation of Item Response Data: Solving the Inverse Design Problem}
\author{JoonHo Lee}
\date{\footnotesize October 1, 2026. \\[0.5em]
Lee: Department of Educational Studies in Psychology, Research Methodology, and Counseling,
The University of Alabama, jlee296@ua.edu.
\\[0.5em]
Acknowledgments: The author is grateful for the support of the Institute of Education Sciences Grant R305D240078.
}

\numberwithin{equation}{section}

\renewcommand{\theHequation}{\thesection.\arabic{equation}}
\renewcommand{\theHfigure}{Main.\arabic{figure}}
\renewcommand{\theHtable}{Main.\arabic{table}}

\begin{document}

\begin{abstract}
\input{sections/00_abstract}
\end{abstract}

\maketitle


\pagestyle{plain}
\newpage

\input{sections/01_introduction}
\input{sections/02_indices}
\input{sections/03_inverse}
\input{sections/04_algorithms}
\input{sections/05_validation}
\input{sections/06_illustration}
\input{sections/07_discussion}

\section*{Software and replication materials}
\noindent We provide the R package \href{https://joonho112.github.io/IRTsimrel/}{\texttt{IRTsimrel}} (\url{https://joonho112.github.io/IRTsimrel/}) for calibrating IRT simulations to specified information targets. The package website includes installation instructions, examples and function documentation. A separate \href{https://github.com/joonho112/IRTsimrel-paper-replication}{replication repository} (\url{https://github.com/joonho112/IRTsimrel-paper-replication}) contains the study code, input data, saved results and package versions used in this paper. Its step-by-step guide walks readers through a worked example, validation runs and reconstruction of the paper's figures and tables.

\printbibliography[title={References}]

\newpage
\appendix
\setcounter{section}{0}
\renewcommand{\thesection}{\Alph{section}}
\numberwithin{equation}{section}
\numberwithin{figure}{section}
\numberwithin{table}{section}
\numberwithin{proposition}{section}
\numberwithin{theorem}{section}
\numberwithin{lemma}{section}
\numberwithin{corollary}{section}
\numberwithin{algorithm}{section}
\renewcommand{\theHequation}{Supp.\thesection.\arabic{equation}}
\renewcommand{\theHfigure}{Supp.\thesection.\arabic{figure}}
\renewcommand{\theHtable}{Supp.\thesection.\arabic{table}}
\crefalias{section}{appendix}
\renewcommand{\floatnote}[1]{%
  \par\vspace{4pt}{\footnotesize\setstretch{1}\raggedright
  \noindent\textit{Note.}\enspace #1\par}}

\begin{center}
{\Large\bfseries Online Supplemental Materials}\\[1em]
{\large Reliability-Targeted Simulation of Item Response Data: Solving the Inverse Design Problem}\\[1.5em]
{\normalsize JoonHo Lee}
\end{center}

\vspace{1.5em}

\noindent We organize the supplement around the decisions needed to design and assess a reliability-targeted simulation. Appendix A defines the indices and establishes the mathematical results. Appendix B gives the algorithms and their convergence conditions; Appendix C evaluates calibration accuracy. Appendix D provides executable examples.

\medskip
\noindent Readers interested in applications can turn to the item-level treatment-effect simulation in Appendix E, the input distributions and Item Response Warehouse difficulty pool in Appendix F, the item-generation comparison in Appendix G, and the DIF and score studies in Appendix H. Tables, figures and equations carry appendix letters; references to the main text retain its numbering.

\vspace{1.5em}

\begingroup
\small
\singlespacing
\startcontents[appendices]
\printcontents[appendices]{}{1}{\textbf{Contents}\vskip1em\hrule\vskip1em}
\vskip1em\hrule\vskip2em
\endgroup

\clearpage
\input{appendices/A_targets}

\clearpage
\input{appendices/B_algorithms}

\clearpage
\input{appendices/C_validation}
\clearpage
\input{appendices/D_example}
\clearpage
\input{appendices/E_illustration}

\clearpage
\input{appendices/F_populations}
\clearpage
\input{appendices/G_atlas}
\clearpage
\input{appendices/H_dif}

\end{document}

%% file: sections/00_abstract.tex
Monte Carlo studies of item response theory (IRT) methods often leave the information level implicit in the choice of item parameters. Designs with the same sample size and test length can therefore represent quite different measurement conditions. We develop a calibration procedure that makes the information level an explicit design choice. The target is an information-based reliability index derived from population mean test information. Before generating responses, we solve for a common multiplier on item discriminations, keeping difficulties, test length and the reference population fixed. The R package \texttt{IRTsimrel} implements the procedure, allowing researchers to compare methods across chosen information levels or compare designs at a common target. In an item-level treatment-effect simulation, power increases across targets, while excess rejection from ignoring item-effect variation persists. These comparisons show how targeting helps distinguish findings that depend on the chosen information level from those that persist across the studied range.

%% file: sections/01_introduction.tex
\section{Introduction}
\label{sec:introduction}

Monte Carlo simulation allows researchers to evaluate IRT estimation methods, model-selection criteria and scoring rules against known generating parameters. The aims of a study determine its data-generating mechanism and the measures used to assess performance \parencite{morris_using_2019}. IRT studies usually vary the sample size, test length, item difficulty and discrimination distributions, and latent distribution \parencite[e.g.,][]{woods_thissen_2006,paganin_irt_2022,cheng_estimating_2025,guastadisegni_generalized_2025}. These factors describe how the responses are generated, but do not directly state how much information the responses carry about the latent trait.

Two designs with the same numbers of persons and items can carry different amounts of information because discrimination and the location of difficulties relative to abilities both matter. In measurement practice, marginal reliability summarizes the proportion of trait-scale variance due to true differences among persons rather than measurement error \parencite{cheng_comparison_2012,kim2012note}. We evaluated 25 item-generation schemes of the kinds that simulation studies commonly use, each at the discriminations it generates, under four latent distributions and two sources of item difficulties. The information-based marginal reliability (defined in \cref{sec:framework}) ranged from .515 to .930 across the schemes. Even eight 20-item two-parameter schemes that differ only in how their difficulties or discriminations are drawn ranged from .781 to .835 (\cref{app:atlas}).

If a study reports only its item-generation scheme, readers must infer the information level from that scheme. Differences in information can then complicate comparisons across conditions or studies, even when their sample sizes and test lengths agree. More information at a given ability level generally means greater precision in estimating ability \parencite[Chapter~3]{baker_item_2004}. When methods are evaluated under different generating schemes, an apparent advantage may partly reflect how much information each test provides. Reporting information alongside the generating parameters helps readers judge whether the methods faced comparable measurement conditions.

In this paper we develop a procedure for setting the information level before generating responses. We choose a target for an information-based reliability index and solve for a common multiplier on the item discriminations, keeping the difficulties, reference population and test length fixed. For example, a researcher may generate a 20-item form and choose a mean-information index of .85 under a normal population. In the worked example of \cref{subsec:example}, multiplying the discriminations by 1.278 reaches this target. The procedure returns the calibrated form, achieved index and calibration status. We then check the index on fresh draws before generating responses.

\textcite{paganin_irt_2022} compared computational strategies for Bayesian semiparametric IRT across sample sizes, test lengths and multimodal latent distributions. \textcite{cheng_estimating_2025} and \textcite{guastadisegni_generalized_2025} varied latent distributions and structural parameters to study estimation and testing under non-normality. These designs specify which features change across conditions. They leave open whether, for example, normal and bimodal populations provide the same information, or whether a difference attributed to shape partly reflects a change in information.

Simulation studies in multilevel modeling often specify an intraclass correlation (ICC) and set the between- and within-group variances to obtain it. For example, \textcite{maas_sufficient_2005} varied the ICC along with the number and size of groups because the ICC and group size determine the design effect. \textcite{mcneish_effect_2016} examined the ICC as a moderator of small-sample bias, and multisite-trial simulations vary the between-site effect variance relative to the within-site sampling variance \parencite{lee_et_al_msts_2025}. An information-based reliability index can serve a similar design purpose in IRT: researchers choose a target index value and adjust the item parameters to reach it. This analogy concerns how these indices are used, not a mathematical equivalence between reliability and the ICC.

The information levels included in a simulation should reflect its purpose. Reported reliabilities span a wide range. \textcite{frisbie_reliability_1988} described coefficients near .50 for teacher-made tests and .90 for commercially developed standardized tests. More recently, \textcite{conoyer_meta-analysis_2022} reported alternate-form reliabilities from .21 to .89 for curriculum-based measures in science and social studies. Developmental surveillance data may contain sparse responses from only 115 children \parencite{gilholm_bayesian_2021}. In a census of 889 item-response data sets, about 30\% fell below the conventional .80 reliability threshold even under a lenient definition \parencite{lee_irw_reliability_2026}.

Method performance can also depend on the information available. \textcite{soland_how_2024} created lower-reliability conditions by lowering every item slope by the same amount. For short survey scales, the choice of scoring method affected the Type~I and Type~II error rates of group comparisons. A study restricted to high-information conditions cannot assess such behavior. Reporting the information level helps readers judge which measurement settings a study represents; targeting it allows researchers to choose those settings explicitly.

IRT already provides tools for prospective measurement design. Researchers have used item parameters to plan raw-score reliability \parencite{dimitrov_marginal_2003}, assembled tests to meet information requirements \parencite{luecht_applications_2016}, and computed reliability coefficients from fitted models \parencite{andersson_irtreliability_2018}. Our problem starts with a generated form and asks how to adjust it to reach a chosen information level. As far as we are aware, simulation practice has no routine procedure for this step. We retain the generating model, difficulties, relative discriminations and test length, and solve only for the scale of the discriminations.

In \cref{sec:framework} we define the target and distinguish it from other information indices and score reliability. In \cref{sec:inverse} we formulate the calibration problem and give conditions for a solution. \Cref{sec:algorithms} describes the two procedures, their independent evaluation and their implementation in \texttt{IRTsimrel}. We assess calibration accuracy in \cref{sec:validation} and vary the target in an item-level treatment-effect simulation in \cref{sec:illustration}. In \cref{sec:discussion} we discuss when to use targeting, what to report and the limits of a scalar target.

The Online Supplemental Materials provide proofs, algorithm details, full validation results, an executable example and details of the item-level treatment-effect simulation (\crefrange{app:targets}{app:illustration}). They also describe the latent distributions and empirical difficulty pool, compare item-generation schemes, and examine differential item functioning (\crefrange{app:shapes}{app:dif}).

%% file: sections/02_indices.tex
\section{Information Indices for a Fixed Form and Population}
\label{sec:framework}

\subsection{The item response model and test information}
\label{subsec:model}

Consider a test of $I$ dichotomous items administered to persons with latent abilities $\theta$. We use the two-parameter logistic (2PL) model \parencite{birnbaum_latent_1968}, under which person $p$ answers item $i$ correctly with probability
\begin{equation}
\label{eq:2pl}
\Pr(Y_{pi}=1\mid\theta_p)=\logit^{-1}\{\lambda_i(\theta_p-\beta_i)\},
\end{equation}
where $\beta_i$ is the difficulty, $\lambda_i>0$ the discrimination, and the logistic scaling constant is one. The Rasch model \parencite{rasch_probabilistic_1980} has equal baseline discriminations. Abilities follow a distribution $G$ with finite positive variance $v=\Var_G(\theta)$. During calibration we hold one realized item form $\{(\beta_i,\lambda_{i,0})\}_{i=1}^{I}$ and $G$ fixed. A study can calibrate each generated form separately or choose a common multiplier for a population of forms; we discuss the latter in \cref{subsec:sac,app:algorithms}. \Cref{app:targets} gives the three-parameter (3PL) extension with a lower asymptote, which we also include in the validation.

Precision in IRT is conditional on the trait level. Under local independence the test information function is the sum of the item information functions \parencite{baker_item_2004,de_ayala_theory_2022},
\begin{equation}
\label{eq:tif}
\Jinfo(\theta)=\sum_{i=1}^{I}\lambda_i^{2}\,\pi_i(\theta)\{1-\pi_i(\theta)\},
\qquad
\pi_i(\theta)=\logit^{-1}\{\lambda_i(\theta-\beta_i)\},
\end{equation}
and $1/\sqrt{\Jinfo(\theta)}$ is the asymptotic standard error of the maximum likelihood estimate of ability at $\theta$. \Cref{fig:anatomy} shows how the response curves determine item information, test information and the information-based standard error for three small forms. Precision varies across abilities. To summarize it, we average over a specified population $G$ \parencite{cheng_comparison_2012}. The resulting index depends on both the items and the population; changing from a normal to a bimodal population can change the index without changing any item.

\begin{figure}[ht]
\centering
\includegraphics[width=\textwidth]{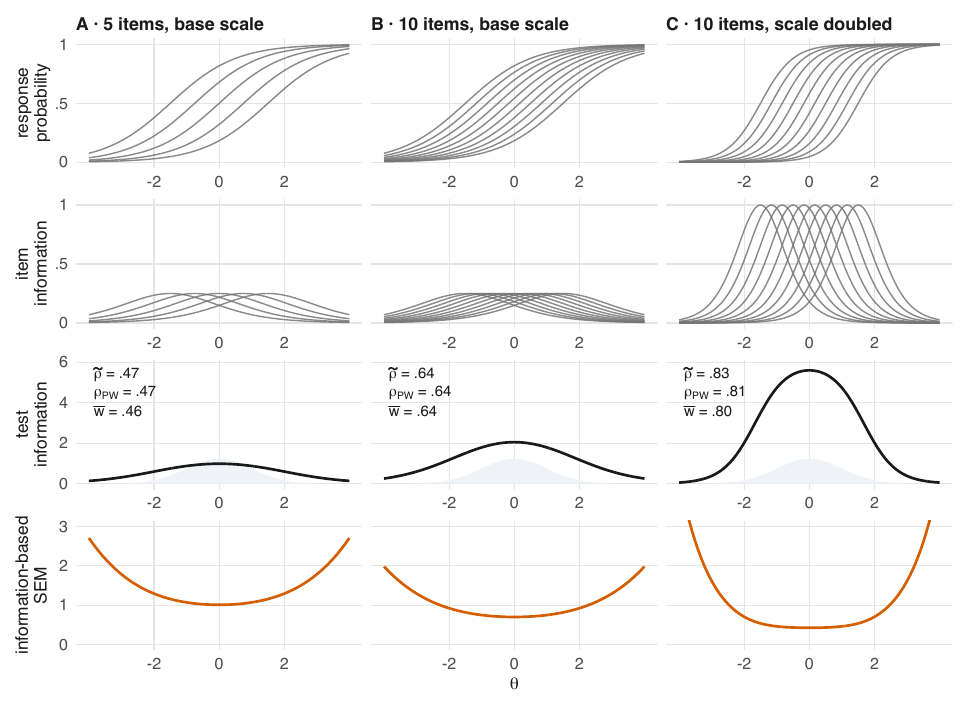}
\caption{From response curves to a scalar index: three small forms}
\label{fig:anatomy}

\smallskip
\begin{minipage}{0.95\textwidth}
\footnotesize
\emph{Note.} Rows show item response probabilities, item information, test information $\Jinfo(\theta)$ with the standard normal reference density shaded, and the information-based standard error $\Jinfo(\theta)^{-1/2}$ (displayed up to three). Columns A and B use five and ten equally spaced difficulties on $[-1.5,1.5]$ with unit slopes; column C keeps the ten difficulties and doubles every slope. The three indices annotated in the third row are defined in \cref{eq:index} and \cref{eq:three} and were evaluated on 200,000 standardized normal nodes with $v=1$. Doubling the slopes raises all three indices, sharpens the information peak and deepens the tails; it does not multiply a fixed curve by a constant.
\end{minipage}
\end{figure}

\subsection{The mean-information index}
\label{subsec:target}

We take as the primary target a transformation of population mean information. Let $A(c)$ denote the mean of the test information function under $G$ when every discrimination is multiplied by $c$ (the multiplier is introduced formally in \cref{sec:inverse}). The index is
\begin{equation}
\label{eq:index}
\rhotilde(c)=\frac{v\,A(c)}{1+v\,A(c)},
\qquad
A(c)=\E_G\big[\Jinfo(\theta;c)\big].
\end{equation}
This transformation puts mean information on a zero-to-one scale. The quantity $vA$ is the signal-to-noise ratio obtained by using $1/A$ as an error variance, and $\rhotilde$ expresses that ratio as a proportion. With unit population variance, a target of .85 requires mean information $.85/.15\approx5.67$. We call $\rhotilde$ the mean-information index. It summarizes a design's information and is not the reliability of any particular score \parencite{doran_information_2005}.

We use this index because its target equation, $vA(c)=\rhotilde^*/(1-\rhotilde^*)$, is linear in mean information. We can therefore solve it on fixed nodes or estimate it with an unbiased score from fresh draws. The index is finite and continuous in $c$ for every population with finite positive variance. We detected one target crossing in each of the 743 successful calibrations (\cref{sec:validation}). Implementations may use either the population variance or the node sample variance for $v$; we address this choice in \cref{subsec:fresh-eval}.

\subsection{Two other information indices, and the reliability of scores}
\label{subsec:other-indices}

Two other information indices appear in the reliability literature. The mean-information index averages information and then transforms; the other two reverse that order or average the reciprocal:
\begin{equation}
\label{eq:three}
\rhopw(c)=\E_G\!\left[\frac{v\Jinfo(\theta;c)}{1+v\Jinfo(\theta;c)}\right],
\qquad
\wbar(c)=\frac{v}{v+\E_G\big[1/\Jinfo(\theta;c)\big]} .
\end{equation}
The pointwise index $\rhopw$ transforms at each ability and then averages. The inverse-information index $\wbar$ is the globalized coefficient of \textcite{cheng_comparison_2012}, derived using an asymptotic maximum likelihood approximation. The parallel-forms interpretation discussed by \textcite{kim2012note} requires conditionally unbiased ability estimates with conditional error variance $1/\Jinfo(\theta;c)$ and finite marginal variance. Because the reciprocal is convex and the transformation is concave, $\wbar\le\rhopw\le\rhotilde$, with equality only when information is constant across the population (\cref{app:targets}). \Cref{tab:indices} summarizes the properties of the three indices.

\input{tables/tab_indices}

The gaps can be substantial. In the fixed-form validation of \cref{sec:validation}, we calibrated each form to a target for $\rhotilde$. The pointwise index was within .034 of that target in 90\% of cells, but up to .40 lower in heavy-tailed and three-parameter cells. The inverse-information index was lower by a median of .008 and by more than .22 in the top decile of cells with finite expected reciprocal information (\cref{app:targets}). Along the multiplier, the indices can also move differently. For the eight forms in \cref{fig:functionals}, the mean-information index increases throughout the displayed range, whereas the other two eventually decrease. Under the heavy-tailed population, expected reciprocal information is infinite and the population inverse-information index is zero at every multiplier.

\begin{figure}[!htb]
\centering
\includegraphics[width=\textwidth]{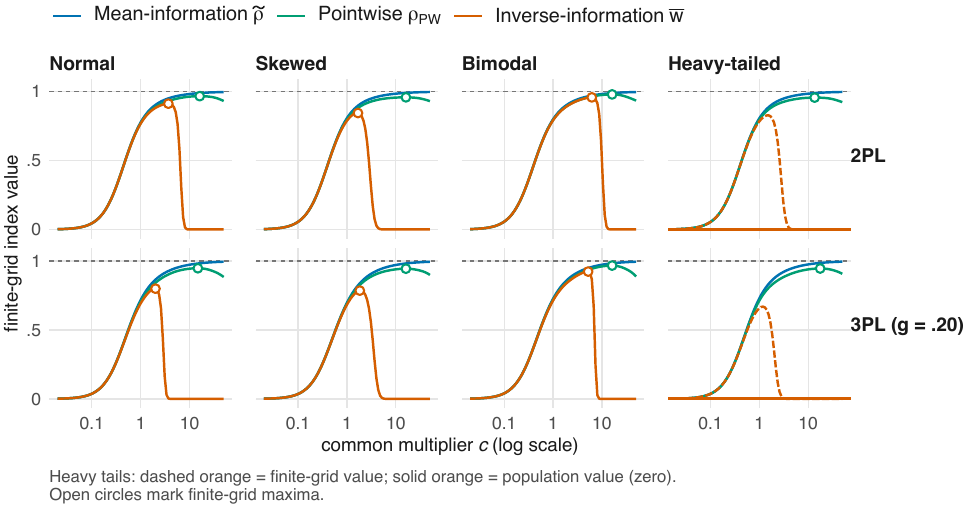}
\caption{The three information indices as the common multiplier moves}
\label{fig:functionals}

\smallskip
\begin{minipage}{0.95\textwidth}
\footnotesize
\emph{Note.} Eight 20-item forms cross the four latent shapes of the validation grid with the 2PL kernel (top row) and the 3PL kernel with guessing .20 (bottom row). Blue is $\rhotilde$, green is $\rhopw$ and orange is $\wbar$, evaluated on a fixed grid of 4,001 density-weighted nodes with represented variance one as the multiplier runs over 91 log-spaced values from .02 to 50. Open circles mark the grid maxima of the pointwise and inverse-information curves. Under the heavy-tailed ($t_5$) population the expected reciprocal information is infinite, so the population inverse-information index is zero (solid orange line); the dashed curve is the finite-grid value only. The retained grid data are described in \cref{app:targets}.
\end{minipage}
\end{figure}

Score reliability also depends on how responses are scored. For expected a posteriori (EAP) and weighted likelihood (WLE) estimates, it depends on the estimator, any scoring prior and the estimated item parameters. None of the three information indices equals score reliability in general. A target for $\rhotilde$ specifies the mean information of the generating design; the reliability of the resulting scores must be assessed separately. \Cref{app:targets} gives the exact posterior identity for EAP scores and the empirical comparisons.

%% file: tables/tab_indices.tex
\begin{table}[!htb]
\centering
\caption{Three information indices and what distinguishes them}
\label{tab:indices}
\footnotesize
\setlength{\tabcolsep}{4pt}
\begin{tabular}{@{}>{\raggedright\arraybackslash}p{4.3cm}>{\raggedright\arraybackslash}p{2.9cm}>{\raggedright\arraybackslash}p{2.9cm}>{\raggedright\arraybackslash}p{3.1cm}@{}}
\toprule
 & Mean-information $\rhotilde$ & Pointwise $\rhopw$ & Inverse-information $\wbar$ \\
\midrule
Construction & transform of mean information & mean of pointwise transforms & transform of mean reciprocal information \\
\addlinespace
Positive (finite inverse-information moment) for a population with a polynomial tail & yes & yes & no: extended value zero \\
\addlinespace
Turning points on the 32 finite grids of \cref{fig:functionals} ($c\in[.02,50]$) & none & $c$ from 9.6 to 27.2 & $c$ from 1.0 to 10.5 \\
\addlinespace
Behavior as $c\to\infty$ (density positive at the difficulties) & $1-\rhotilde=O(c^{-1})$ & $O(c^{-1}\log c)\to0$ & $\to0$ \\
\addlinespace
Targets unreachable within $[0.1,10]$ in the 768-cell census & 20 of 768 & not targeted & 167 of 576 eligible cells \\
\addlinespace
Gap below $\rhotilde$ at the calibrated root (743 cells; median, 90th percentile) & --- & .005, .034 & .008, .22 (559 cells) \\
\addlinespace
Distance to a fitted EAP coefficient in one 30-item Rasch example & .00508 & .00023 & .00141 \\
Distance to the fitted WLE coefficient in the same example & .01716 & .01231 & .01068 \\
\bottomrule
\end{tabular}

\smallskip
\begin{minipage}{0.95\textwidth}
\footnotesize
\emph{Note.} Turning points and limits are properties of the finite grids and the propositions of \cref{app:targets}, not of every population. The census counts targets outside the strict interior of the range reached on 20,000 nodes within the declared interval; the inverse-information census excludes 192 heavy-tailed cells whose population index is zero. The EAP and WLE distances come from one retained example with 6,000 persons under a standard normal population and are discrepancies from fitted coefficients, not errors against a population truth; the 48-cell correspondence study is summarized in \cref{app:targets}.
\end{minipage}
\end{table}

%% file: sections/03_inverse.tex
\section{Calibration as a Root-Finding Problem}
\label{sec:inverse}

\subsection{A common multiplier on the discriminations}
\label{subsec:inverse}

The design variable is a common multiplier. For a realized form with baseline discriminations $\lambda_{i,0}$ we set
\begin{equation}
\label{eq:scaling}
\lambda_i(c)=c\,\lambda_{i,0},\qquad c>0,\; i=1,\dots,I.
\end{equation}
The multiplier preserves the difficulties, relative discriminations, population and test length. It changes the absolute discriminations and the shape of the information function. In column C of \cref{fig:anatomy}, doubling the slopes raises information in the middle of the population but lowers it at the edges. Thus, calibration can change where the form measures well. Test assembly methods select items to meet information requirements \parencite{van_der_linden_linear_2005}. We use a common multiplier because it retains the generated form's relative item structure and requires only one unknown. A Rasch baseline becomes a one-parameter logistic form with common discrimination $c$ on the fixed latent scale.

\subsection{Existence, uniqueness and infeasible targets}
\label{subsec:existence}

Given a target $\rhotilde^*\in(0,1)$ and an admissible interval $[c_L,c_U]$, we seek a multiplier $c^*$ such that $\rhotilde(c^*)=\rhotilde^*$. Continuity guarantees a root if the endpoint values lie on opposite sides of the target. Strict increase guarantees uniqueness on a sub-interval, but need not hold globally. \Cref{app:targets} gives a sufficient condition based on the distance between persons and items on the calibrated logit scale, and a smooth population with finite variance for which the equation has several solutions.

The other indices often have multiple roots. On the finite grids in \cref{fig:functionals}, the pointwise index reached its maximum at multipliers between 9.6 and 27.2 and the inverse-information index between 1.0 and 10.5, giving two roots for targets below the maximum. The mean-information index did not turn over on any of the 32 grids examined (\cref{app:targets}).

Although every successful calibration in our numerical study had one detected crossing, the procedure allows for several. It scans the interval, records detected crossings and applies a branch policy, by default selecting the lowest increasing crossing. We report both the interval and the policy.

A target is infeasible on the chosen interval if it lies outside the range of $\rhotilde$. Because the function need not be monotone, the scan searches for interior extrema as well as crossings. For an infeasible request, the R package \texttt{IRTsimrel} returns the multiplier with the best detected value, which may be an endpoint or an interior extremum, together with the achieved index, an infeasible status and a warning. We do not count this return as a calibrated design. A finite scan can miss an extremum and therefore cannot certify a global optimum.

The researcher can widen the multiplier bounds if the resulting discriminations are acceptable, add items, change difficulty coverage or lower the target. Each changes the design. Infeasibility concerns the specified form, population and interval; another design of the same length may reach the target.

%% file: sections/04_algorithms.tex
\section{Calibration Algorithms}
\label{sec:algorithms}

\subsection{The workflow}
\label{subsec:workflow}

\Cref{fig:forward-inverse} compares the conventional forward design with the inverse design used here. We recommend five steps. First, specify the latent population $G$ and generate an item form. Second, choose the target $\rhotilde^*$, index and admissible multiplier interval $[c_L,c_U]$. Third, solve for $c^*$ using fixed nodes drawn from $G$. Fourth, evaluate the calibrated form on fresh draws and record both the solver residual and the fresh-draw discrepancy. Fifth, generate responses for the required persons and replications. The calibration nodes approximate an expectation over $G$; the response simulation uses new persons and the calibrated item parameters.

\begin{figure}[ht]
\centering
\input{figures/fig1_forward_inverse}
\caption{Forward and inverse design of an item response simulation}
\label{fig:forward-inverse}

\smallskip
\begin{minipage}{0.95\textwidth}
\footnotesize
\emph{Note.} Both sequences start from the same generated structure. The forward sequence reads the information level off the generated parameters. The inverse sequence declares the target, solves for one multiplier on fixed nodes, evaluates the calibrated form on fresh draws, and only then generates responses. The values in the lower row are those of the worked example in \cref{subsec:example}.
\end{minipage}
\end{figure}

\subsection{Empirical quadrature calibration}
\label{subsec:eqc}

Empirical quadrature calibration (EQC) evaluates the target equation on a fixed sample. It draws $M$ nodes $\theta_1,\dots,\theta_M$ from $G$ once, keeps the realized form fixed, and for any candidate $c$ computes
\begin{equation}
\label{eq:eqc}
\hat A_M(c)=\frac1M\sum_{m=1}^{M}\Jinfo(\theta_m;c),
\qquad
\hat{\rhotilde}_M(c)=\frac{\hat v\,\hat A_M(c)}{1+\hat v\,\hat A_M(c)},
\end{equation}
with $\hat v$ the sample variance of the nodes. Reusing the nodes makes $\hat{\rhotilde}_M(\cdot)$ a deterministic function, so we can inspect its shape, scan for crossings and bracket a root. Fixed nodes do not guarantee monotonicity or uniqueness. Quadrature also provides numerical approximations for reliability coefficients computed from fitted models \parencite{andersson2018large}.

EQC scans a grid in $\log c$ over the specified interval. It detects sign changes of $\hat{\rhotilde}_M(c)-\rhotilde^*$, classifies crossings as increasing or decreasing, and records detected extrema. A bracketed root finder \parencite{brent_algorithm_1971} then refines the selected crossing to the requested tolerance. The output includes $c^*$, the calibrated discriminations $c^*\lambda_{i,0}$, achieved index, residual, crossing count, interval and status. If the target is outside the detected range, EQC returns the best detected multiplier, an endpoint or an interior extremum, with an infeasible status. We report such requests as unsuccessful. \Cref{app:algorithms} gives the algorithm and branch rules.

\subsection{Evaluation on fresh draws}
\label{subsec:fresh-eval}

A small solver residual establishes that the numerical index met its target. To check the population index, we draw a large independent sample from $G$ and recompute the index at $c^*$. We report this fresh-draw discrepancy alongside the solver residual. It includes error from approximating mean information, the variance convention and the numerical root; the solver residual measures only the last component.

The evaluation must state which variance it uses. Calibration uses the node sample variance, whereas a fresh evaluation can use either its own sample variance or the known population variance. This choice can shift the index by as much as the fresh-draw discrepancy itself. The evaluation also concerns a fixed form before responses are generated. Variation from sampling persons and scoring their responses occurs later and is outside the calibration (\cref{app:validation}).

\subsection{A worked example}
\label{subsec:example}

We calibrate a 20-item 2PL form to .85 under a standard normal population using \texttt{IRTsimrel} 0.3.1 \parencite{Lee2026IRTsimrel031} in R 4.6.0 \parencite{RCoreTeam2026R460}. The solver finds one increasing crossing at $c^*=1.278$ and achieves .8500 on 20,000 nodes with sample variance 1.0056. Evaluation on 200,000 fresh draws standardized to variance one gives .8493. The difference of $-.0007$ mainly reflects the variance convention: using 1.0056 in the population calculation gives .8500, leaving a finite-draw difference in the fourth decimal. The calibrated discriminations have 10th, 50th and 90th percentiles of .96, 1.25 and 1.92. The box and \cref{fig:example} summarize the calibration; \cref{app:example} provides the script, diagnostics and an infeasible request. We generate no responses or scores in this example.

\begin{mdframed}[linewidth=0.5pt, roundcorner=5pt]
\small
\textbf{Worked example: one 20-item form calibrated to $\rhotilde^*=.85$}

\smallskip
\begin{tabular}{@{}p{0.30\textwidth}p{0.62\textwidth}@{}}
\textbf{Structure} & 2PL; $\beta_i\sim N(0,1)$; log-normal baseline discriminations; $G=N(0,1)$; $I=20$ \\
\textbf{Declaration} & target $\rhotilde^*=.85$; mean-information index; interval $[0.1,\,10]$; lowest increasing crossing \\
\textbf{Solver} & $M=20{,}000$ nodes, node sample variance $\hat v=1.0056$; one crossing detected; status \texttt{ok} \\
\textbf{Output} & $c^*=1.2779$; achieved $\hat{\rhotilde}_M(c^*)=.8500$ (residual $2\times10^{-9}$); calibrated discriminations $p_{10}/p_{50}/p_{90}=.96/1.25/1.92$ \\
\textbf{Fresh evaluation} & 200,000 draws from $N(0,1)$, variance one: $\rhotilde=.8493$; pointwise index $.8400$; inverse-information index $.8354$ \\
\textbf{Not done} & no responses generated, no scoring model fitted
\end{tabular}
\end{mdframed}

\begin{figure}[ht]
\centering
\includegraphics[width=\textwidth]{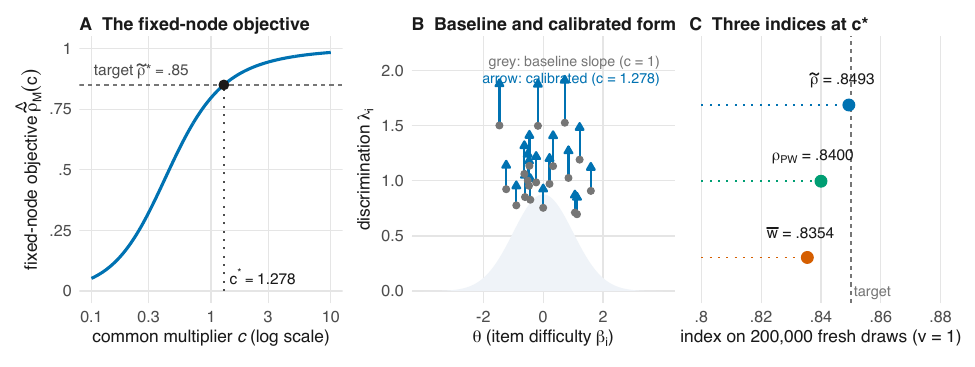}
\caption{The worked example: the objective, the calibrated form and the three indices}
\label{fig:example}

\smallskip
\begin{minipage}{0.95\textwidth}
\footnotesize
\emph{Note.} Panel A shows the fixed-node objective $\hat{\rhotilde}_M(c)$ of the worked example on the declared interval $[0.1,10]$ (log scale), the target .85 and the single increasing crossing at $c^*=1.278$. Panel B shows the 20 difficulties with their baseline (grey) and calibrated (blue) discriminations against the standard normal reference density. Panel C shows the three indices of the calibrated form on 200,000 fresh standardized draws with $v=1$; only the mean-information index was targeted. The objective is recomputed from the retained item form on 20,000 standard-normal nodes with their sample variance; the retained solver output is in \cref{app:example}.
\end{minipage}
\end{figure}

\subsection{Stochastic approximation calibration}
\label{subsec:sac}

Stochastic approximation calibration (SAC) is useful when fresh draws are convenient or when we want to calibrate mean information over a population of forms. SAC updates the multiplier using a fresh batch at each iteration \parencite{robbins_monro_1951}. With frozen variance estimate $\hat v$, $q=\rhotilde^*/(1-\rhotilde^*)$ and batch mean information $\hat A_n$, its score is $U_n(c)=(1-\rhotilde^*)^2\{\hat v\hat A_n(c)-q\}$. The update is
\begin{equation}
\label{eq:sac}
c_{n+1}=\Pi_{[L,U]}\big\{c_n-a_n\,U_n(c_n)\big\},
\qquad
a_n=\frac{a}{(n+A)^{\gamma}},
\end{equation}
where $\Pi$ projects onto an effective interval selected by a preliminary scan. We report the average of the iterates after burn-in \parencite{polyak_acceleration_1992}. The score is linear in batch mean information and hence unbiased for the population equation conditional on the frozen variance and branch. With redrawn forms, SAC averages information across forms before transforming it; this differs from averaging the forms' own indices. Convergence requires a stable branch and standard step-size and moment conditions. Finite-run diagnostics remain necessary. \Cref{app:algorithms} gives the conditions, sensitivity results and the distinction from an earlier implementation based on transformed batch reliabilities.

\subsection{The choice between the two procedures}
\label{subsec:choice}

For routine use we recommend EQC followed by evaluation on fresh draws. Its objective is reproducible from a seed and can be inspected, and calibration requires one scan and a bracketed solve. SAC is useful when the equation is naturally stochastic, forms are redrawn, or a comparison using fresh batches is needed. Starting SAC at the EQC solution reduces burn-in, but does not independently check whether SAC can find the root. We therefore report both warm and cold starts in \cref{app:algorithms}. For either procedure, report the target index, population and variance rule, search interval and branch, multiplier, calibrated discrimination summary and fresh-draw evaluation.

%% file: figures/fig1_forward_inverse.tex
\begin{tikzpicture}[
  font=\small,
  box/.style={draw=black!60, rounded corners=2pt, align=center, minimum height=1.15cm, text width=2.6cm, inner sep=4pt, fill=black!4},
  new/.style={box, fill=blue!9, draw=blue!55!black},
  lab/.style={font=\small\bfseries, anchor=east},
  val/.style={font=\footnotesize, text width=2.6cm, align=center, text=black!70},
  arr/.style={-{Latex[length=2mm]}, thick, draw=black!65},
  node distance=0.36cm
]
\node[lab] (fl) at (-0.25,0) {Forward};
\node[box, right=0.25cm of fl] (f1) {choose $G$ and\\ an item form};
\node[box, right=of f1] (f2) {compute information\\ or fit a model};
\node[box, right=of f2] (f3) {observe the\\ information level};
\node[box, right=of f3] (f4) {generate responses\\ at whatever level\\ resulted};
\draw[arr] (f1) -- (f2); \draw[arr] (f2) -- (f3); \draw[arr] (f3) -- (f4);

\node[lab] (il) at (-0.25,-2.6) {Inverse};
\node[box, right=0.25cm of il] (i1) {choose $G$ and\\ a baseline form};
\node[new, right=of i1] (i2) {declare $\rhotilde^*$ and\\ an interval $[c_L,c_U]$};
\node[new, right=of i2] (i3) {solve for $c^*$\\ on fixed nodes};
\node[new, right=of i3] (i4) {evaluate on fresh\\ draws; then\\ generate responses};
\draw[arr] (i1) -- (i2); \draw[arr] (i2) -- (i3); \draw[arr] (i3) -- (i4);

\node[val, below=0.12cm of i1] {20 items, 2PL\\ $G=N(0,1)$};
\node[val, below=0.12cm of i2] {$\rhotilde^*=.85$\\ $[0.1,\,10]$};
\node[val, below=0.12cm of i3] {$c^*=1.278$\\ achieved $.8500$};
\node[val, below=0.12cm of i4] {fresh $\rhotilde=.8493$\\ (200,000 draws)};
\end{tikzpicture}

%% file: sections/05_validation.tex
\section{Numerical Validation}
\label{sec:validation}

\subsection{Design and evaluation}
\label{subsec:questions}

We assess whether EQC reaches its numerical target and how closely the calibrated form meets the target on fresh draws. We also examine how the fresh-draw discrepancy changes with the number of calibration nodes. \Cref{tab:evidence} summarizes the validation studies and the item-level treatment-effect simulation in \cref{sec:illustration}, including their designs, replication counts and evaluation rules. \Cref{app:algorithms,app:validation} report the SAC comparisons, person-sampling variability and the feasibility census that preceded the fixed-form grid.

\input{tables/tab_main_evidence}

\subsection{Accuracy of the solver on its own nodes}
\label{subsec:numerical}

We crossed four item models (Rasch, 2PL, and 3PL with guessing fixed at .20 or drawn from a $\mathrm{Beta}(5,17)$ distribution) with four latent distributions: standard normal, symmetric bimodal, positively skewed gamma and heavy-tailed Student $t$ with five degrees of freedom. We standardized each population to mean zero and variance one. Difficulties were drawn from $N(0,1)$ or sampled from the Item Response Warehouse \parencite{domingue_introduction_2025,zhang_realistic_2025}. Tests contained 10, 20, 40 or 60 items, and targets were .50, .60, .70, .80, .90 or .95. \Cref{app:validation,app:shapes} give the generators and full design.

A feasibility screen reduced the nominal 768 combinations to 748 cells. We generated one form per cell and calibrated it using 20,000 fixed draws from its reference population, searching for a multiplier within $[0.1,10]$. The root-finding tolerance was $10^{-6}$ on $\log c$.

EQC succeeded in 743 of the 748 cells. Among these, the mean absolute solver residual was $1.5\times10^{-9}$ and the maximum was $3.8\times10^{-8}$. These residuals are differences between the achieved index and the target on the same nodes used for calibration. They were similarly small across lengths, shapes and targets.

Even the largest residual was less than one millionth of an index point. Thus, the solver matched the requested target very closely on its calibration nodes. Accuracy for the reference population also depends on how well those nodes represent it. We examine this source of error using fresh draws in \cref{subsec:fresh}.

\subsection{Discrepancy on fresh draws and the number of nodes}
\label{subsec:fresh}

We examined whether using more calibration nodes brings the achieved index closer to the target on fresh draws. We crossed three item models (Rasch, 2PL and 3PL with guessing fixed at .20) with two latent distributions (standard normal and a positively skewed gamma standardized to mean zero and variance one) and two test lengths (10 and 40 items). All 12 cells used difficulties drawn from $N(0,1)$ and a target of .70.

We calibrated each form using $M$ ability draws, with $M$ set to 500, 1,000, 2,500, 5,000, 10,000, 20,000 or 50,000. At each node count, we ran 20 repetitions per cell, each with a new form and new calibration nodes, giving 240 calibrations. We then held each calibrated form fixed and evaluated its index on 200,000 fresh ability draws from the same population, using their sample variance. We measured error as the achieved index minus .70. Because the evaluation sample is finite, this error approximates the population discrepancy. No item responses were generated or scored.

In \cref{fig:holdout}, solver residuals stayed near $10^{-9}$ at every node count. Mean absolute fresh-draw error was .0117 at $M=500$, .0018 at $M=20{,}000$ and .0014 at $M=50{,}000$. Thus, a form could meet the target almost exactly on its calibration nodes and still differ from it on new draws. Increasing $M$ reduced this discrepancy, approximately at the square-root rate described in \cref{app:algorithms}.

\begin{figure}[!htb]
\centering
\includegraphics[width=0.8\textwidth]{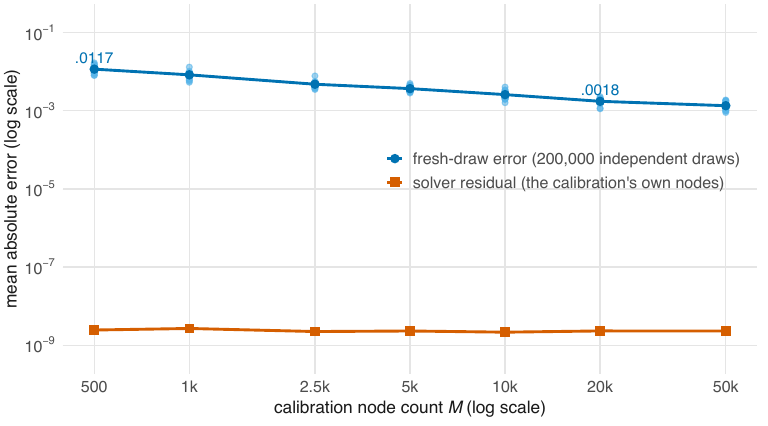}
\caption{Solver residual and fresh-draw error as functions of the node count}
\label{fig:holdout}

\smallskip
\begin{minipage}{0.95\textwidth}
\footnotesize
\emph{Note.} Target .70; 12 design cells (Rasch, 2PL, 3PL with guessing .20 $\times$ normal, skewed gamma $\times$ 10, 40 items; parametric difficulties) $\times$ 20 seeds at each of seven node counts. Both lines are mean absolute errors over the 240 calibrations at each $M$; light points are the 12 cell means. The fresh-draw error is computed on 200,000 independent draws from the same population, using the sample variance of those draws. The standard deviation of the signed fresh-draw error was .0147 at $M=500$ and .0023 at $M=20{,}000$ (\cref{app:validation}).
\end{minipage}
\end{figure}

The node count controls how well the calibration represents the reference population. For these conditions, the results support starting with a few thousand nodes when discrepancies of about .01 are acceptable, and using tens of thousands when greater precision is needed. These are guides based on average errors, not guarantees for individual forms. Errors were larger for the skewed population, so we recommend checking the achieved index on fresh draws for the chosen form and population (\cref{app:validation}).

\Needspace{4\baselineskip}
\subsection{The multipliers that the targets require}
\label{subsec:settings}

The distributions in \cref{subsec:numerical} are stress cases rather than representative operational populations. The Item Response Warehouse provides an empirical spread of difficulties, but not a joint distribution of all item parameters (\cref{app:shapes}).

The same target required different multipliers across shapes and difficulty sources. At target .70 with 20 items, 2PL forms with IRW difficulties required multipliers from .60 under the heavy-tailed population to 1.59 under the skewed population (\cref{app:validation}). Because each cell used a separately generated form, these differences reflect both the population shape and the realized items.

The multiplier also depends strongly on length and target, as \cref{fig:cstar} shows. At target .95 the median multiplier for 10-item Rasch forms was 7.64 and for 60-item forms 1.48, and over the whole grid the returned multipliers ranged from .23 to 9.52, with 16\% above 2.5.

Heavy tails and guessing can make high targets difficult to reach within the multiplier bounds. In cells matched on target, shape and length, 3PL forms required median multipliers 1.31 (fixed guessing) and 1.41 (beta-distributed guessing) times those of 2PL forms (\cref{app:validation}). These comparisons describe the generated forms; they do not isolate the effect of guessing or shape.

\begin{figure}[ht]
\centering
\includegraphics[width=\textwidth]{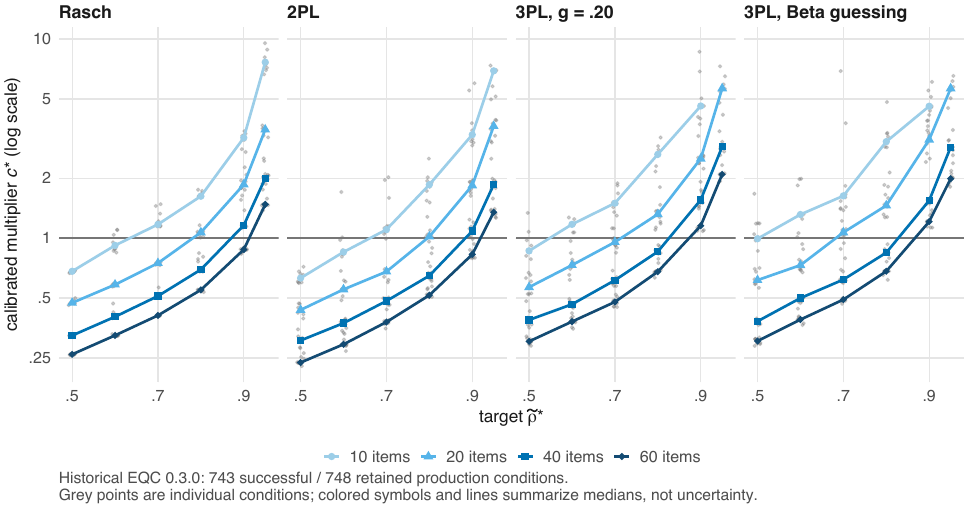}
\caption{The multiplier that the target requires, by item model, test length and target}
\label{fig:cstar}

\smallskip
\begin{minipage}{0.95\textwidth}
\footnotesize
\emph{Note.} Grey points are the 743 successful EQC calibrations of the fixed-form grid (\texttt{IRTsimrel} 0.3.0; 20,000 nodes; interval $[0.1,10]$); the five infeasible cells have no point. Colored symbols and lines join the medians over the four latent shapes and two difficulty pools at each length and target; they are visual guides, not uncertainty intervals or the path of one fixed form. The multiplier rescales each form's own baseline discriminations, so the horizontal line at one marks no rescaling (it is not a common discrimination value). Because forms were generated separately in every cell, the patterns across models and lengths are descriptive comparisons.
\end{minipage}
\end{figure}

\subsection{Infeasible cells}
\label{subsec:infeasible}

Five of the 748 cells did not reach their target within $[0.1,10]$. All were 3PL cells with 10 or 20 items, four at target .95 and one at .90. In each, the largest detected index fell below the target (\cref{app:validation}). We retain these cells in the denominator and report them as infeasible. A researcher facing such a result can lengthen the form, allow larger discriminations or lower the target.

The preliminary census used the full 768-cell design. Within $[0.1,10]$, 20 mean-information targets were outside the detected range, compared with 167 of the 576 cells eligible for the inverse-information index. This greater frequency of infeasible requests is one reason we target mean information. The production scan detected exactly one crossing in all 743 successful cells. \Cref{app:example} gives a complete infeasible return.

%% file: tables/tab_main_evidence.tex
\begin{table}[ht]
\centering
\caption{Study designs and evaluation criteria}
\label{tab:evidence}
\footnotesize
\setlength{\tabcolsep}{4pt}
\begin{tabular}{@{}>{\raggedright\arraybackslash}p{2.4cm}>{\raggedright\arraybackslash}p{3.6cm}>{\raggedright\arraybackslash}p{3.6cm}>{\raggedright\arraybackslash}p{3.6cm}@{}}
\toprule
 & Fixed-form calibration grid & Fresh-node holdout study & Fixed-length item-level treatment-effect simulation \\
\midrule
Question & Does the solver meet its numerical objective? & How far is the calibrated form from the population index, and how does that depend on $M$? & How do rejection rates change across information targets? \\
\addlinespace
Design and denominator & Rasch, 2PL and two 3PL variants $\times$ 4 latent shapes $\times$ 2 difficulty pools $\times$ 4 lengths $\times$ 6 targets; 748 production cells after a feasibility screen, one form per cell & Rasch, 2PL, 3PL ($g=.20$) $\times$ normal, skewed $\times$ 10, 40 items; target .70; 20 seeds at each of 7 node counts $=$ 1,680 calibrations & 6 targets $\times$ 6 effect conditions $\times$ 1,000 independent replications; 20 items, 500 persons; a new form per replication \\
\addlinespace
Calibration engine & \texttt{IRTsimrel} 0.3.0 EQC through the study wrapper; $M=20{,}000$; interval $[0.1,10]$; lowest increasing crossing & the same EQC path with the node count varied & study-code one-dimensional root solve for a common slope; 20,000 shared reference draws; $v=1$ \\
\addlinespace
Evaluation rule & achieved numerical index on the calibration nodes against the target (own-rule residual) & index on 200,000 fresh draws, with their sample variance, against the target & rejection rates of constant- and varying-effect mixed models; exact binomial intervals \\
\addlinespace
Result & 743 successes; MAE $1.5\times10^{-9}$; 5 infeasible cells, all 3PL at targets .90 or .95 & fresh-draw MAE .0117 at $M=500$ and .0014 at $M=50{,}000$; solver residual near $10^{-9}$ throughout & power .466 to .723 (reference mean .4, SD 0); null rejection .063 to .091 (mean 0, SD .4); no per-replication fresh evaluation \\
\bottomrule
\end{tabular}

\smallskip
\begin{minipage}{0.95\textwidth}
\footnotesize
\emph{Note.} MAE $=$ mean absolute error. The first two studies calibrate the mean-information index with the package solver; the third solves the same equation for a common slope in study code. \Cref{app:validation,app:illustration} give the designs, the full tables and the version provenance; \cref{app:algorithms} reports the stochastic calibrations of the same grid.
\end{minipage}
\end{table}

%% file: sections/06_illustration.tex
\section{Varying the Information Level in an Item-Level Treatment-Effect Simulation}
\label{sec:illustration}
\label{subsec:illustration}

We extend the item-level treatment-effect simulation of \textcite{gilbert_modeling_2023} to compare method performance at specified information levels. Each replication contains 500 persons answering 20 dichotomous items, with 250 persons assigned to treatment and 250 to control. Abilities and item difficulties are drawn from $N(0,1)$, and all items share a common discrimination.

The benchmark sets the common discrimination to 1.7 and draws item effects from normal distributions with means of 0 or .4 and standard deviations of 0, .2 or .4. A standard deviation of zero gives the same effect on every item; positive values allow effects to differ across items. We retain these six effect conditions and hold their distributions fixed on the ability scale as we vary information. Specifically, we rescale each reference logit effect $\tau_{i,0}$ as $\tau_i=(d/1.7)\,\tau_{i,0}$, where $d$ is the calibrated common discrimination. A reference mean of .4 therefore represents an average effect of about .24 standard deviations of ability at every target.

We cross the six effect conditions with six mean-information targets: .50, .60, .70, .80, .887 and .90. The .887 target approximates the benchmark's information level. Each of the 36 conditions has 1,000 replications, with a new form generated and calibrated in each replication. The study code solves directly for the common discrimination in the untreated response model, using 20,000 fixed standard-normal reference draws and population variance one. We did not record a fresh-draw evaluation for each replication. Runs at different targets are independent, with no pairing of forms across targets.

We fit two logistic mixed models to each data set. Both include person and item random intercepts. The constant-effect model assumes the same treatment effect on every item; the varying-effect model adds an item random slope for treatment, allowing effects to differ across items. We test the treatment coefficient using a two-sided Wald test at the 5\% level. \Cref{app:illustration} gives the generating model and results for all 36 conditions.

\begin{figure}[!htb]
\centering
\includegraphics[width=0.8\textwidth]{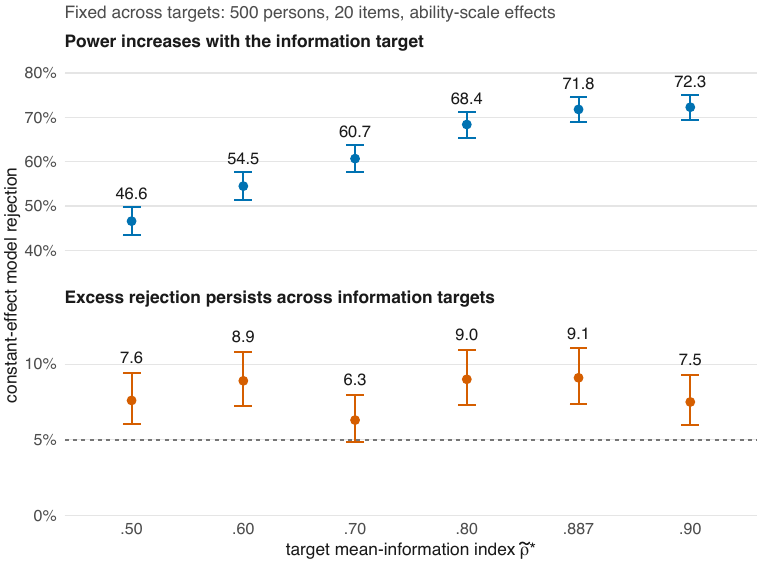}
\caption{Varying the information target in a fixed-length item-level treatment-effect simulation}
\label{fig:arm-a}

\smallskip
\begin{minipage}{0.95\textwidth}
\footnotesize
\emph{Note.} Constant-effect logistic mixed model; 20 items, 500 persons, a new form calibrated at each of 1,000 independent replications per target; exact two-sided 95\% binomial intervals. Within each panel, the distribution of item effects is held fixed on the ability scale across targets (reference-logit coefficients rescaled by $d/1.7$); under the fixed-logit convention the top panel reverses direction (\cref{app:illustration}). Top: reference effect mean .4, item-effect SD 0. Bottom: no average treatment effect in the item population, with effects varying across items (reference mean 0, item-effect SD .4); the dashed line is the nominal 5\% level. The two panels are different conditions and use different vertical scales.
\end{minipage}
\end{figure}

When the reference effect is .4 and the item-effect standard deviation is zero, the constant-effect model is correctly specified. It detected the treatment effect in 46.6\% of replications at target .50 and 72.3\% at target .90 (\cref{fig:arm-a}, top). The exact 95\% binomial intervals at these endpoints do not overlap. Thus, studies with the same number of persons and items, and the same ability-scale treatment effect, can have substantially different power across information levels.

This comparison depends on how we hold the treatment effect fixed. Our ability-scale convention makes the logit coefficient increase with the calibrated discrimination. Fixing the logit coefficient instead makes the ability-scale effect shrink as the target rises. Under that convention, rejection fell from 100\% to 56\% between targets .50 and .90 in a sensitivity analysis with 100 replications per condition (\cref{app:illustration}). Researchers therefore need to specify the effect scale alongside the information target. The two conventions answer different questions; neither isolates a pure effect of reliability.

We next consider a treatment whose effects vary across items but average to zero in the item population (reference mean 0, item-effect standard deviation .4). This null hypothesis allows positive effects on some items and negative effects on others. The constant-effect model rejected it in 6.3\% to 9.1\% of replications (\cref{fig:arm-a}, bottom). Rejection exceeded the nominal 5\% at every target, with exact lower interval limits above .05 at five of the six targets.

A randomly generated 20-item form will generally have a nonzero mean effect even when the item-population mean is zero. The constant-effect model estimates the effect for the administered items and omits the uncertainty from sampling those items \parencite{gilbert_modeling_2023}. Increasing information improves precision for that form's effect, while the variation across possible forms remains. The varying-effect model accounts for item-effect variation and gave rejection rates of 4.1\% to 6.1\%, consistent with the nominal level (\cref{app:illustration}). For an item-population conclusion, the analysis needs to account for differences in effects across items even when the test is highly informative.

Within each effect condition, targeting lets us compare method performance across specified information levels while keeping sample size, test length and the ability-scale effect distribution fixed. Power increased across targets, whereas excess rejection from ignoring item-effect variation persisted throughout the range studied.

%% file: sections/07_discussion.tex
\section{Discussion}
\label{sec:discussion}

Our procedure makes the information level an explicit design choice in an IRT simulation. We set a target for the mean-information index and solve for a common multiplier on the item discriminations before generating responses. Within each calibration, the difficulties, relative discriminations, reference population and test length remain fixed. Researchers can use this procedure to compare methods across chosen information levels or to compare different designs at a common target.

The choice should follow the simulation's aims \parencite{morris_using_2019}. If the aim is to study particular item parameters, researchers can retain those parameters as design factors and report the resulting information index. If the aim is to evaluate methods under specified measurement conditions, targeting lets researchers construct those conditions directly, including tests with low information.

The item-level treatment-effect simulation let us compare methods across targets with sample size, test length and the ability-scale effect distribution fixed within each effect condition. Power increased across targets, while excess rejection from ignoring item-effect variation persisted (\cref{fig:arm-a}). These comparisons show which findings depend on the chosen information level and which persist across the studied range.

Comparisons also depend on how simulated effects are defined. In the differential item functioning study, we compared Mantel--Haenszel detection rates for 40-item and 20-item forms (\cref{app:dif}). Relative to the uncalibrated design, setting the information target to .75 changed this length contrast by $-20.3$ percentage points when the group difference in item difficulty was held fixed, but by $+8.6$ points when the group difference in the response model's linear predictor was held fixed. Changing discriminations changes the relationship between these two effects. Researchers therefore need to specify the effect scale as well as the information target.

A small solver residual shows that the numerical target was met. Evaluation on fresh ability draws checks how closely the calibrated form meets the target beyond those nodes. We recommend reporting both. \Cref{tab:checklist} gives the details needed to reproduce the calibration, assess the resulting items and interpret the performance summaries, extending standard simulation reporting \parencite{morris_using_2019,siepe_simulation_2024}.

\input{tables/tab_checklist}

Reaching a target may require items that are unsuitable for the intended application. In our comparison of item-generation schemes, the median calibrated discrimination at target .85 was 3.76 for five-item forms and .57 for 60-item forms (\cref{app:atlas}). The shorter tests needed much more discriminating items to reach the same target. Researchers should inspect the calibrated parameters and choose multiplier bounds that reflect plausible items. If a target cannot be reached within those bounds, the target, test length or item-generation scheme needs reconsideration.

Matching the mean-information index does not make two tests equally precise at every ability level. Information curves may still differ near a cut score or in the tails, so researchers should inspect the regions that matter for their question. The target also does not determine score reliability, which depends on how responses are scored \parencite{kim2012note}. In our comparisons, EAP reliabilities were generally closer to the pointwise index than to the mean-information index, while WLE reliabilities were closer to the inverse-information index than to the pointwise index (\cref{app:targets}). Studies concerned with score reliability should assess it using their chosen scoring procedure.

We compare complete calibrated designs. Changing the target changes item discriminations, and the extension that varied test length and target used newly generated, recalibrated forms (\cref{app:illustration}). These comparisons do not isolate causal effects of reliability or item count. An analysis intended to generalize to an item population must also account for variation in effects across items.

We evaluated dichotomous Rasch, 2PL and 3PL forms under four latent distributions and two difficulty sources. Guessing parameters were held fixed during 3PL calibration, and all infeasible cells in the fixed-form validation grid were 3PL cells. These four shapes do not represent the full range of populations encountered in practice. The Item Response Warehouse provided empirical difficulties; we imposed the distributions of the other item parameters and their relationships (\cref{app:shapes}). Polytomous and multidimensional models, and targets for information in specific ability regions, require different information functions and further validation.

We recommend reporting the achieved information index and its reference population for every simulation condition, and the target when calibration is used. This makes the measurement conditions visible to readers and helps them judge where a study's findings apply.

%% file: tables/tab_checklist.tex
\begin{table}[!htb]
\centering
\caption{What to report for a reliability-targeted simulation design}
\label{tab:checklist}
\footnotesize
\setlength{\tabcolsep}{6pt}
\begin{tabularx}{\textwidth}{@{}>{\raggedright\arraybackslash}p{0.50\textwidth}>{\raggedright\arraybackslash}X@{}}
\toprule
\textbf{What to report} & \textbf{Why it matters} \\
\midrule
\textbf{Index and target.} Name the information index, give its formula and state the target value. & This defines the measurement condition. Different information indices are not interchangeable, and none is a score reliability. \\
\addlinespace[4pt]
\textbf{Reference population and variance.} Specify the ability distribution and whether each calculation uses the population variance or a sample variance. & The same items can provide different information for different populations. The variance rule can change the index even when the solver residual is negligible. \\
\addlinespace[4pt]
\textbf{Fixed and adjusted quantities.} State which item parameters, population features and design factors are held fixed, and which change. & This makes clear what differs between calibrated conditions and helps readers interpret the comparison. \\
\addlinespace[4pt]
\textbf{Effect scale, if applicable.} State the scale on which simulated effects are held fixed and how they are rescaled during calibration. & A fixed logit coefficient can imply different ability-scale effects as discriminations change. \\
\addlinespace[4pt]
\textbf{Calibration across forms.} State whether each form is calibrated separately or one multiplier is used for a population of forms. & Calibrating mean information over forms does not make every form meet the target. \\
\addlinespace[4pt]
\textbf{Numerical settings.} Give the algorithm, integration rule, node count, search interval, tolerance and rule for choosing among multiple solutions. & These settings make the calculation reproducible. A failed search within the stated bounds does not show that the target is impossible. \\
\addlinespace[4pt]
\textbf{Accuracy on calibration nodes and fresh draws.} Report the solver residual and the achieved index on an independent ability sample, with its size. & The solver can match the target on its own nodes while missing it on new draws. \\
\addlinespace[4pt]
\textbf{Calibrated item parameters.} Report the multiplier and summarize the resulting discriminations. & Reaching a target may require item parameters that are implausible for the intended application. \\
\addlinespace[4pt]
\textbf{Information beyond the target.} Report relevant companion indices and information curves. & Equal overall indices can conceal different precision near a cut score or in the tails. \\
\addlinespace[4pt]
\textbf{Calibration outcomes.} Give attempted, successful, infeasible and failed counts, and explain how unsuccessful runs were handled. & Reporting only successful runs can hide unattainable conditions and make performance look better than it is. \\
\bottomrule
\end{tabularx}

\smallskip
\begin{minipage}{\textwidth}
\footnotesize
\emph{Note.} For a simulation without a target, report the achieved index for each condition under the stated reference population and variance rule.
\end{minipage}
\end{table}

%% file: appendices/A_targets.tex
\section{Targets and Mathematical Conditions}
\label{app:targets}
\label{app:theory}
\label{app:definitions}

We use the mean-information index to set the information level of a simulated test. This appendix explains how it differs from the pointwise and inverse-information indices, when a target can be reached, and why none of these indices determines score reliability. We state the population results first and then examine their numerical implications.

\subsection{The three information indices}
\label{app:sa01}

We hold the items and reference population fixed and vary one common discrimination multiplier. Write the baseline discriminations as $a_i=\lambda_{i,0}>0$, difficulties as $b_i=\beta_i$, and lower asymptotes as $g_i\in[0,1)$. Setting $g_i=0$ gives the Rasch and 2PL cases. For a multiplier $c>0$, let $x_i=c\,a_i(\theta-b_i)$ and $s(x)=(1+e^{-x})^{-1}$. The response probability is $p_{i,c}(\theta)=g_i+(1-g_i)s(x_i)$, and test information is
\begin{equation}
\label{eq:app-tif}
\Jinfo_c(\theta)=\sum_{i=1}^{I}\frac{\{\partial_\theta p_{i,c}(\theta)\}^2}{p_{i,c}(\theta)\{1-p_{i,c}(\theta)\}}
=\sum_{i=1}^{I} c^2a_i^2\,\frac{(1-g_i)\,s(x_i)^2\{1-s(x_i)\}}{g_i+(1-g_i)s(x_i)},
\end{equation}
which reduces to $\sum_i c^2a_i^2 s(x_i)\{1-s(x_i)\}$ when every $g_i=0$.

For a population $G$ with variance $v=\Var_G(\theta)\in(0,\infty)$, define
\begin{equation}
\label{eq:app-three}
\rhotilde(c)=\frac{vA(c)}{1+vA(c)},\qquad
\rhopw(c)=\E_G\!\left[\frac{v\Jinfo_c(\theta)}{1+v\Jinfo_c(\theta)}\right],\qquad
\wbar(c)=\frac{v}{v+M(c)},
\end{equation}
where $A(c)=\E_G[\Jinfo_c(\theta)]$ and $M(c)=\E_G[1/\Jinfo_c(\theta)]\in(0,\infty]$. We set $\wbar(c)=0$ when $M(c)=\infty$. The indices average information, transformed information and reciprocal information, respectively. All use the same items, population and variance. A numerical approximation must also identify its nodes and variance rule; we write $\hat{\rhotilde}_M$ for the fixed-node objective.

\Needspace*{6\baselineskip}
\subsection{Finiteness, bounds and the Jensen ordering}
\label{app:sa02}

\begin{proposition}[Finiteness, bounds and continuity]
\label{prop:finite}
Fix a bank as above and a population $G$ with $v\in(0,\infty)$. For every finite $c>0$, $0<\Jinfo_c(\theta)\le B(c)\equiv c^2\sum_i a_i^2/4$ for all $\theta$, so $0<A(c)\le B(c)<\infty$, and
\[
0\le \wbar(c)\le\rhopw(c)\le\rhotilde(c)\le U(c)\equiv\frac{vB(c)}{1+vB(c)}<1 .
\]
The maps $c\mapsto A(c)$, $\rhotilde(c)$ and $\rhopw(c)$ are continuous on $(0,\infty)$ and tend to zero as $c\downarrow0$. The inverse-information index is positive if and only if $M(c)<\infty$; under the convention of \cref{eq:app-three} it equals zero when $M(c)=\infty$, so it is finite in either case. Since $U$ increases in $c$, $U(c_U)$ bounds all three indices on any interval $(0,c_U]$, and a target above it cannot be reached there.
\end{proposition}

\begin{proof}
Write $q=s(x_i)$. The 3PL information summand equals the 2PL summand $c^2a_i^2q(1-q)$ multiplied by $(1-g_i)q/\{g_i+(1-g_i)q\}\le1$. Thus each summand is positive and at most $c^2a_i^2/4$. Summing gives the bound on $\Jinfo_c$ and $A(c)$. Applying the increasing map $x\mapsto vx/(1+vx)$ gives the bounds on $\rhotilde$ and $\rhopw$; \cref{prop:order} gives the ordering.

Each summand is continuous in $c$ and bounded uniformly on any compact multiplier interval. Dominated convergence therefore gives continuity of $A$ and $\rhopw$, and hence of $\rhotilde$. The zero limits follow from $B(c)\to0$.
\end{proof}

\begin{proposition}[Ordering]
\label{prop:order}
For every $c>0$, $\wbar(c)\le\rhopw(c)\le\rhotilde(c)$. If $M(c)<\infty$, equality holds in either inequality if and only if $\Jinfo_c(\theta)$ is $G$-almost surely constant. If $M(c)=\infty$ the lower inequality is strict, because $\wbar=0<\rhopw$.
\end{proposition}

\begin{proof}
The map $\phi(x)=vx/(1+vx)$ is strictly concave on $(0,\infty)$. Jensen's inequality gives $\rhopw=\E_G[\phi(\Jinfo_c)]\le\phi(\E_G[\Jinfo_c])=\rhotilde$, with equality exactly when information is almost surely constant.

For the other inequality, write $\phi(x)=\psi(1/x)$, where $\psi(y)=(1+y/v)^{-1}$ is strictly convex. If $M(c)<\infty$, Jensen's inequality gives $\rhopw=\E_G[\psi(1/\Jinfo_c)]\ge\psi(\E_G[1/\Jinfo_c])=\wbar$, with the same equality condition. If $M(c)=\infty$, then $\wbar=0$ while $\phi(\Jinfo_c)>0$ almost surely.
\end{proof}

At a mean-information target $\rhotilde^*$, the other two indices can be no larger than that target. The ordering does not say how much smaller they will be. Positive population variance alone does not make the inequalities strict, because the support of $G$ could lie on a level set of $\Jinfo_c$.

Root ordering requires more: on an interval where all three indices are continuous, strictly increasing and reach a common target, $c^*_{\rhotilde}\le c^*_{\rhopw}\le c^*_{\wbar}$. Each inequality is strict exactly when information is not almost surely constant at the larger root. These conditions matter because the indices need not increase throughout the search interval.

In the 743 successful fixed-form calibrations, the ordering held on the calibration nodes. The gap $\rhotilde-\rhopw$ had median .005, 90th percentile .034 and maximum .40. Among the 559 successful cells without heavy tails, $\rhotilde-\wbar$ had median .008 and 90th percentile .22. The largest gaps occurred in heavy-tailed and 3PL cells and the smallest under the bimodal population (\cref{fig:gaps}). We exclude the heavy-tailed cells from the inverse-information comparison because their population moment $M(c)$ is infinite (\cref{prop:tail}).

\begin{figure}[ht]
\centering
\includegraphics[width=\textwidth]{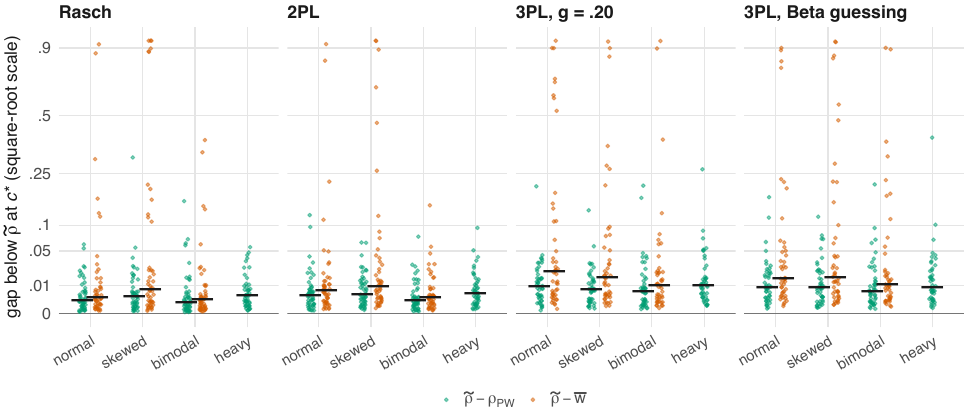}
\caption{How far the pointwise and inverse-information indices fall below the calibrated mean-information index}
\label{fig:gaps}

\smallskip
\begin{minipage}{0.95\textwidth}
\footnotesize
\emph{Note.} Points represent the 743 successful fixed-form cells, evaluated at the EQC root on its calibration nodes. Bars mark model-by-population medians for each gap (28 groups), computed before the square-root display transformation. Green shows $\rhotilde-\rhopw$; orange shows $\rhotilde-\wbar$ for the 559 cells with finite population reciprocal-information moment. The heavy-tailed population is excluded from the orange series by \cref{prop:tail}.
\end{minipage}
\end{figure}

\subsection{When a target has one or more solutions}
\label{app:sa03}

A root finder needs a search interval and a rule for choosing among solutions. Continuity alone does not guarantee a unique solution, even for the mean-information index. We first give a simple counterexample and a sufficient condition for local increase, then distinguish the population limits from what a finite set of nodes can show.

\begin{example}[A two-point population]
\label{ex:twopoint}
Take one 2PL item with $a_1=1$ and $b_1=0$, and let $G$ place mass $1/2$ at each of $\theta=-1$ and $\theta=1$, so that $v=1$. Then $A(c)=c^2 h(c)$ with $h(x)=s(x)\{1-s(x)\}=e^{x}/(1+e^{x})^2$, which tends to zero both as $c\downarrow0$ and as $c\to\infty$ while being positive in between. The mean-information index therefore rises and then falls, and every target below its maximum has two solutions. A finite set of nodes that happens to avoid the item difficulty behaves in the same way even when the generating population is continuous.
\end{example}

\begin{proposition}[Derivative and a sufficient local condition; 2PL]
\label{prop:local}
Let every $g_i=0$ and put $h(x)=s(x)\{1-s(x)\}$ and $x_i=c\,a_i(\theta-b_i)$. Then for each $\theta$
\begin{equation}
\label{eq:app-deriv}
\frac{\partial}{\partial c}\Jinfo_c(\theta)=\sum_{i=1}^{I} c\,a_i^2\,h(x_i)\,\Big\{2-x_i\tanh(x_i/2)\Big\},
\end{equation}
and differentiation under the expectation is permitted on any compact multiplier interval. Let $x_0\approx2.3994$ be the unique positive root of $2-x\tanh(x/2)$. If an interval $K\subset(0,\infty)$ satisfies $|c\,a_i(\theta-b_i)|<x_0$ for every item, every $c\in K$ and $G$-almost every $\theta$, then $A'(c)>0$ and $\rhotilde'(c)>0$ on $K$, so a target has at most one solution in $K$.
\end{proposition}

\begin{proof}
Since $h'(x)=h(x)\{1-2s(x)\}=-h(x)\tanh(x/2)$ and $\partial x_i/\partial c=x_i/c$, differentiating $c^2a_i^2h(x_i)$ gives \cref{eq:app-deriv}. 

On a compact interval with upper end $c_U$, $|\partial_c\Jinfo_{i,c}|\le c_U a_i^2\{2h(x_i)+|x_i|h(x_i)\}\le c_Ua_i^2(1/2+1/e)$, a constant, so dominated convergence justifies $A'(c)=\E_G[\partial_c\Jinfo_c]$. 

The function $\psi(x)=2-x\tanh(x/2)$ is even, equals $2$ at zero, is strictly decreasing on $(0,\infty)$ and tends to $-\infty$, so it has one positive root $x_0$ and is positive exactly on $(-x_0,x_0)$. Under the stated condition every summand of \cref{eq:app-deriv} is positive almost surely, so $A'(c)>0$, and $\rhotilde'(c)=vA'(c)/\{1+vA(c)\}^2>0$.
\end{proof}

\Cref{prop:local} gives a sufficient condition that requires bounded population support relative to the items. It cannot hold for a normal population, whose tails extend beyond $x_0/(c\,a_i)$ from an item. The condition applies to the 2PL; the additional 3PL factor in \cref{eq:app-tif} also depends on $c$. The next example shows that smoothness and finite variance alone do not guarantee monotonicity.

\begin{proposition}[A smooth population with three solutions]
\label{prop:counter}
Take a single 2PL item with $a_1=1$, $b_1=0$, and the mixture $G=.9998\,N(3,.15^2)+.0001\,N(0,.15^2)+.0001\,N(-100,.15^2)$, which has a bounded smooth density and finite variance $v\approx1.084$. Then $A(1/2)>1/50>A(2)$, while $A(c)\to0$ as $c\downarrow0$ and $A(c)\to\infty$ as $c\to\infty$. Consequently $A(c)=1/50$, equivalently $\rhotilde(c)=(v/50)/(1+v/50)\approx.0212$, has at least three solutions in $(0,\infty)$, and $\rhotilde$ is not monotone on $(0,\infty)$.
\end{proposition}

\begin{proof}
Let $Z\sim N(3,.15^2)$ denote the dominant component. By Chebyshev's inequality, $\Pr(2.55\le Z\le3.45)\ge8/9$, and on that event $0<Z/2\le1.725$ and $e^{Z/2}<e^{2}<8$, so $h(Z/2)=e^{Z/2}/(1+e^{Z/2})^2>8/81$ because $t/(1+t)^2$ decreases for $t>1$. Dropping the other nonnegative components, $A(1/2)\ge\frac{4999}{5000}\cdot\frac14\cdot\frac89\cdot\frac{8}{81}=\frac{9998}{455625}>\frac1{50}$. 

For the second bound, $h(2Z)\le e^{-2Z}$ and $\E[e^{-2Z}]=\exp\{-6+2(.15)^2\}=e^{-5.955}<1/225$, while each of the two small components contributes at most $4\cdot\frac14=1$ before its weight $.0001$, so $A(2)\le\frac{4999}{5000}\cdot4e^{-5.955}+\frac{2}{10000}\cdot1<\frac{4}{225}+\frac1{5000}<\frac1{50}$. The bounds $e^2<8$ and $e^{5.955}>225$ follow from finitely many Taylor terms. 

The limit at zero is \cref{prop:finite}. The limit at infinity is \cref{prop:overlap} below, since $f(0)>0$. The intermediate value theorem then yields one solution in each of $(0,1/2)$, $(1/2,2)$ and $(2,\infty)$.
\end{proof}

The three solutions are near $c\approx0.315$, $1.639$ and $76.1$. This example rules out a general monotonicity assumption even for smooth populations. Scanning for crossings helps find solutions, but a finite scan can miss crossings between grid points and cannot locate roots outside its bounds (\cref{app:algorithms}).

\subsection{Large multipliers and population tails}
\label{app:sa03-limits}

Under a continuous population, increasing discrimination concentrates information near the item difficulties. The mean-information index can approach one while the other two indices approach zero. The following results explain this difference.

\begin{proposition}[Local overlap and reachability of every target]
\label{prop:overlap}
Let $G$ have a bounded density $f$ that is continuous at each item difficulty $b_i$, and define $K(g)=1-\{g/(1-g)\}\log(1/g)$ for $0<g<1$ and $K(0)=1$, so that $K(g)\in(0,1]$. Put $S_G=\sum_i a_i f(b_i)K(g_i)$. Then $A(c)/c\to S_G$ as $c\to\infty$. If $S_G>0$, then $\rhotilde(c)\to1$ as $c\to\infty$, $1-\rhotilde(c)=\{vcS_G\}^{-1}\{1+o(1)\}$, and every target in $(0,1)$ has at least one finite positive solution of $\rhotilde(c)=\rhotilde^*$. No monotonicity or uniqueness follows, and the solution need not lie in a prescribed interval.
\end{proposition}

\begin{proof}
Write $h_g(u)=h(u)(1-g)s(u)/\{g+(1-g)s(u)\}$, so that $\Jinfo_{i,c}(\theta)=c^2a_i^2h_{g_i}\{ca_i(\theta-b_i)\}$ by \cref{eq:app-tif}, with $0<h_g\le h$ and $\int h=1$ because $h=s'$. The substitution $u=ca_i(\theta-b_i)$ gives $\E_G[\Jinfo_{i,c}]/c=a_i\int h_{g_i}(u)f(b_i+u/(ca_i))\,du$. The integrand is dominated by $\|f\|_\infty h(u)$ and converges pointwise to $h_{g_i}(u)f(b_i)$, so dominated convergence gives the limit $a_if(b_i)\int h_{g_i}$. 

With $t=s(u)$, $\int h_g=\int_0^1(1-g)t/\{g+(1-g)t\}\,dt=1-\{g/(1-g)\}\log(1/g)=K(g)$, which is positive, and equals one when $g=0$. Summing over items gives $A(c)/c\to S_G$. When $S_G>0$, $A(c)=cS_G\{1+o(1)\}$ and $1-\rhotilde=1/\{1+vA\}$ give the limits; together with $\rhotilde(c)\to0$ as $c\downarrow0$ and continuity, the intermediate value theorem gives a root for every target.
\end{proof}

The population must overlap the item difficulties. For one 2PL item at $b_1=0$ and $G$ uniform on $[1,2]$, $f(0)=0$ and direct integration gives $A(c)=c\{s(2c)-s(c)\}\sim ce^{-c}$. The mean-information index then tends to zero, and sufficiently high targets cannot be reached at any multiplier. Reachability therefore depends on the items and population as well as the allowed interval.

\begin{proposition}[Large multipliers: the other two indices vanish]
\label{prop:vanish}
Suppose only that $G$ puts no mass on any item difficulty. Then $\rhopw(c)\to0$, $M(c)\to\infty$ and $\wbar(c)\to0$ as $c\to\infty$. The pointwise index therefore attains a positive global maximum, strictly below one, at a finite multiplier, and every target strictly between zero and that maximum has at least two solutions, one on each side of a maximizing multiplier. The same conclusion holds for the inverse-information index if it is continuous on $(0,\infty)$ and positive at some multiplier $c_0$, that is, whenever $M(c_0)<\infty$.

If, in addition, all $g_i=0$, the difficulties are distinct and $G$ has a bounded density continuous at every $b_i$, then $\rhopw(c)=(4\log c/c)\sum_i f(b_i)/a_i+o(\log c/c)$.
\end{proposition}

\begin{proof}
Fix $\theta$ different from every $b_i$, a set of full $G$-measure. Since $h(u)\le e^{-|u|}$ and $h_g\le h$, each $\Jinfo_{i,c}(\theta)\le c^2a_i^2e^{-ca_i|\theta-b_i|}\to0$, so $\Jinfo_c(\theta)\to0$; the bounded pointwise transform tends to zero and bounded convergence gives $\rhopw\to0$. Fatou's lemma gives $\liminf M(c_n)\ge\E_G[\liminf 1/\Jinfo_{c_n}]=\infty$ along any $c_n\to\infty$, hence $\wbar\to0$. 

A continuous nonnegative function on $(0,\infty)$ with zero limits at both ends and a positive value somewhere attains a positive maximum on a compact interval, and the intermediate value theorem on either side of a maximizer gives two roots for any smaller positive target; the strict bound below one is \cref{prop:finite}. 

For the rate, put $\phi(z)=z/(1+z)$ and $L_i=vc^2a_i^2$. Because $e^{-u}/4\le h(u)\le e^{-u}$ for $u\ge0$ and each half-line bound integrates $ke^{-u}/(1+ke^{-u})$ to $\log(1+k)$, one has $2\log(1+L/4)\le\int\phi\{Lh(u)\}\,du\le2\log(1+L)$, which is $2\log L+O(1)$. 

Choose $\delta>0$ so that the neighborhoods $N_i=(b_i-\delta,b_i+\delta)$ are disjoint. Outside the union of the $N_i$, and for the other items inside $N_i$, information is at most a constant times $c^2e^{-\kappa c}$ for some $\kappa>0$, and $0<\phi'\le1$ bounds the resulting error in $\phi(v\Jinfo_c)$ by the same order. 

Inside $N_i$, substitute $u=ca_i(\theta-b_i)$ to obtain
\[
\int_{N_i}\phi(v\Jinfo_{i,c})f
=\frac{1}{ca_i}\int_{|u|<ca_i\delta}
\phi\{L_ih(u)\}f\{b_i+u/(ca_i)\}\,du .
\]
By continuity, replacing $f\{b_i+u/(ca_i)\}$ by $f(b_i)$ in the right-hand integral costs $o(\log c)$ before multiplication by $1/(ca_i)$. Extending that integral to the real line costs an exponentially small amount. Thus
\[
\int_{N_i}\phi(v\Jinfo_{i,c})f
=\frac{f(b_i)}{ca_i}\{2\log L_i+O(1)\}+o(\log c/c)
=\frac{4f(b_i)\log c}{ca_i}+o(\log c/c).
\]
Summing over items gives the rate.
\end{proof}

\begin{proposition}[Polynomial tails and the inverse-information index]
\label{prop:tail}
Suppose $G$ has a density $f$ with $f(\theta)\ge k\,\theta^{-(\nu+1)}$ for all $\theta\ge T$, for some $k>0$, $\nu>2$ and $T>\max\{1,\max_i b_i\}$. Then $M(c)=\E_G[1/\Jinfo_c(\theta)]=\infty$ for every $c>0$, so $\wbar(c)=0$ in the extended sense, whereas $\rhotilde(c)$ and $\rhopw(c)$ remain finite and positive.
\end{proposition}

\begin{proof}
For $\theta\ge T$ every $x_i=c\,a_i(\theta-b_i)$ is positive and at least $c\,a_{\min}(\theta-b_{\max})$, and $h(x)\le e^{-x}$ for $x\ge0$; the three-parameter summand is at most the two-parameter one by the bound in the proof of \cref{prop:finite}. Hence $\Jinfo_c(\theta)\le c^2\big(\sum_ia_i^2\big)\exp\{-c\,a_{\min}(\theta-b_{\max})\}$ on $[T,\infty)$, and
\[
M(c)\ge\int_T^\infty\frac{f(\theta)}{\Jinfo_c(\theta)}\,d\theta\ge\frac{k\,e^{-c a_{\min}b_{\max}}}{c^2\sum_ia_i^2}\int_T^\infty e^{c a_{\min}\theta}\,\theta^{-(\nu+1)}\,d\theta=\infty ,
\]
because the integrand exceeds one for all large $\theta$. Finiteness of the other two indices is \cref{prop:finite}.
\end{proof}

The standardized $t_5$ population in our validation grid satisfies this condition. Its population inverse-information index is zero at every multiplier, although any finite set of nodes gives a positive value. As independent nodes sample more of the tails, their mean reciprocal information diverges and the corresponding index tends to zero. A finite calculation cannot establish that the population inverse-information moment is finite.

The lower bound on the tail density is essential. For comparison, a single 2PL item at $b_1=0$ with $a_1=1$ under $G=N(0,1)$ has $M(c)=2\{1+e^{c^2/2}\}/c^2$, which is finite for every finite $c$. Its inverse-information index still tends to zero as $c$ grows, with $\wbar(c)\sim(c^2/2)e^{-c^2/2}$. We therefore use $\wbar$ as a companion index and exclude the $t_5$ population when calibrating it.

\subsection{Score reliability and the information indices}
\label{app:sa04}
\label{app:recovery}

Score reliability also depends on how responses are scored. Let $X$ be a response pattern generated from $\theta\sim G$ and the item model with known parameters. If $T=\E(\theta\mid X)$ is the EAP score under this same joint law, then the tower property gives $\E(\theta\mid T)=T$ and $\Cov(\theta,T)=\Var(T)$. Hence
\begin{equation}
\label{eq:app-eap}
\corr(\theta,T)^2=\frac{\Var(T)}{v}=1-\frac{\E\{\Var(\theta\mid X)\}}{v},
\end{equation}
provided $\Var(T)>0$. This is the matched-model EAP reliability of \textcite{adams_reliability_2005}. The expectation averages posterior variance over response patterns. The identity requires the scoring model and prior to match the law generating those patterns; changing the prior or estimating item parameters changes the calculation. It bounds the average posterior variance, not every response pattern's posterior variance.

Approximating score precision with Fisher information depends on the prior, test length and ability region. \textcite{green_technical_1984} average conditional error variance and compare it with the variance of observed ability estimates. \textcite{kim2012note} shows how estimator bias affects the interpretation of reliability coefficients. \textcite{cheng_comparison_2012} use an asymptotic maximum likelihood approximation to derive the inverse-information coefficient. \textcite{andersson2018large} study the pointwise form on a unit-variance scale.

Score reliability requires an estimator-specific calculation. Weighted likelihood estimation reduces asymptotic bias \parencite{warm_weighted_1989}. EAP posterior standard deviations use the likelihood and prior \parencite{bock_adaptive_1982}. The Rasch person-separation index divides the error-adjusted person standard deviation by the root-mean-square estimated error \parencite{wright_rating_1982}.

For raw scores $S=\sum_iX_i$, local independence gives conditional score variance $\sum_ip_i(1-p_i)$. With true-score function $T_S(\theta)=\sum_ip_i(\theta)$, the model-implied coefficient is $R_S=\Var_G\{T_S\}/[\Var_G\{T_S\}+\E_G\sum_ip_i(1-p_i)]$ \parencite{kim2010estimation}. Its error term is conditional score variance, so it too differs from an information index.

We compared the indices with fitted score reliabilities in 48 cells: Rasch or 2PL, normal or positively skewed ability, 20 or 40 items, 500 or 1,000 persons, and targets .60, .75 or .90. One calibrated bank per cell generated 200 response replications. We fitted the models in TAM \parencite{RobitzschKieferWu2025TAM4325} and computed EAP and WLE scores. 

EAP reliabilities were closer to $\rhopw$ than to $\rhotilde$ in 37 of the 48 cells (sign test $p=.0002$); WLE reliabilities were closer to $\wbar$ than to $\rhopw$ in 43 of 48 ($p<10^{-7}$). We found no clear difference between $\rhopw$ and $\wbar$ as approximations to EAP reliability (28 of 48, $p=.31$). \Cref{tab:s4-correspondence} reports the comparisons.

\input{floats/tab_s4_correspondence}

\subsection{Calibrating each form or a population of forms}
\label{app:sa05}

When forms are redrawn, the order of averaging changes the target. If forms $F$ are drawn from a distribution with finite mean information, concavity gives $\E_F[\phi\{vA_F(c)\}]\le\phi\{v\E_F[A_F(c)]\}$, with equality only when $A_F(c)$ is almost surely constant. Thus, transforming the mean information across forms gives a higher index than averaging the forms' individual indices. Calibrating every form separately answers a third question: each form receives its own multiplier. \Cref{app:algorithms} identifies which quantity each calibration procedure targets.

\subsection{Finite-grid illustrations}
\label{app:sa03-grids}

The 20-item worked example in \cref{app:example} shows why the choice of index matters for calibration. Under a standard normal population, 200-node Gauss--Hermite quadrature with $v=1$ gives one crossing of $\rhotilde(c)=.85$ on $[0.1,10]$, at $c\approx1.2825$. The package root, $1.278$, uses node sample variance $1.0056$. The inverse-information index peaks near $c\approx2.45$ at about .892 and crosses .85 twice, near $c\approx1.40$ and $c\approx3.38$. As discriminations grow, information troughs between items make reciprocal information large. Targeting $\wbar$ therefore requires a choice between branches even for this form.

\Cref{tab:turning} gives maxima and turning points on the 32 grids in \cref{fig:functionals}, over $[.02,50]$. Mean information increased on every grid. \Cref{tab:guttman} illustrates the large-multiplier limits for a 10-item 2PL form under a standard normal population. In \cref{tab:tail-trunc}, widening the standardized $t_5$ quadrature range reduces the inverse-information index from .76 to essentially zero; the mean-information index and normal-population results remain stable. These finite-grid calculations illustrate the population results. The companion file \texttt{theory\_checks.csv} records each check, criterion and result.

\input{tables/tab_turning_points}

\input{tables/tab_guttman_limit}

\input{tables/tab_tail_truncation}

%% file: floats/tab_s4_correspondence.tex
\begin{table}[htbp]\centering
\caption{Information indices as approximations to fitted EAP and WLE reliability}
\label{tab:s4-correspondence}
{\small\begin{tabular}{llrrrr}
\toprule
Score & Comparison & Favours second & Median advantage & Sign $p$ & Status \\
\midrule
EAP & $\wbar$ vs $\rhopsd$ & 28/48 & $6.18\times 10^{-4}$ & $0.3123$ & unresolved \\
EAP & $\rhotilde$ vs $\rhopsd$ & 37/48 & $0.001779$ & $2.22\times 10^{-4}$ & supported \\
WLE & $\rhopsd$ vs $\wbar$ & 43/48 & $0.001452$ & $<10^{-4}$ & supported \\
\bottomrule
\end{tabular}}
\floatnote{The 48 cells cross Rasch/2PL, two ability shapes, two lengths, two person counts and three targets, with 200 response replications each. Counts favoring the second index, median absolute-distance advantages and sign tests summarize those cells. The EAP comparison between pointwise and inverse information is inconclusive ($p=.3123$). These are comparisons with fitted TAM coefficients, not identities for population reliability.}
\end{table}

%% file: tables/tab_turning_points.tex
\begin{table}[!htbp]
\centering
\caption{Turning points of the pointwise and inverse-information curves on the 32 finite grids of \cref{fig:functionals}}
\label{tab:turning}
\small
\setlength{\tabcolsep}{4pt}
\begin{tabular}{llrrrrr}
\toprule
 & & \multicolumn{2}{c}{pointwise $\rhopw$} & \multicolumn{2}{c}{inverse-information $\wbar$} & \\
\cmidrule(lr){3-4}\cmidrule(lr){5-6}
Shape & $I$ & $c$ at maximum & maximum & $c$ at maximum & maximum & $\rhotilde$ monotone \\
\midrule
\multicolumn{7}{l}{\emph{2PL}} \\
normal & 10 & 10.46 & .930 & 2.84 & .825 & yes \\
normal & 20 & 16.15 & .966 & 3.68 & .912 & yes \\
normal & 40 & 17.62 & .985 & 4.78 & .961 & yes \\
normal & 60 & 14.80 & .988 & 3.68 & .972 & yes \\
skewed & 10 & 9.59 & .922 & 1.19 & .710 & yes \\
skewed & 20 & 16.15 & .958 & 1.68 & .844 & yes \\
skewed & 40 & 16.15 & .976 & 1.68 & .920 & yes \\
skewed & 60 & 12.44 & .979 & 1.68 & .942 & yes \\
bimodal & 10 & 14.80 & .951 & 6.77 & .917 & yes \\
bimodal & 20 & 16.15 & .979 & 6.21 & .958 & yes \\
bimodal & 40 & 27.21 & .993 & 10.46 & .985 & yes \\
bimodal & 60 & 24.94 & .996 & 8.06 & .990 & yes \\
heavy-tailed & 10 & 11.41 & .927 & --- & --- & yes \\
heavy-tailed & 20 & 13.57 & .956 & --- & --- & yes \\
heavy-tailed & 40 & 12.44 & .972 & --- & --- & yes \\
heavy-tailed & 60 & 13.57 & .978 & --- & --- & yes \\
\multicolumn{7}{l}{\emph{3PL, $g=.20$}} \\
normal & 10 & 9.59 & .882 & 2.00 & .659 & yes \\
normal & 20 & 14.80 & .947 & 2.00 & .799 & yes \\
normal & 40 & 16.15 & .974 & 2.19 & .898 & yes \\
normal & 60 & 20.96 & .984 & 2.84 & .936 & yes \\
skewed & 10 & 13.57 & .900 & 1.68 & .604 & yes \\
skewed & 20 & 16.15 & .944 & 1.84 & .785 & yes \\
skewed & 40 & 17.62 & .969 & 2.00 & .879 & yes \\
skewed & 60 & 14.80 & .976 & 1.68 & .910 & yes \\
bimodal & 10 & 12.44 & .908 & 6.21 & .844 & yes \\
bimodal & 20 & 16.15 & .966 & 5.22 & .923 & yes \\
bimodal & 40 & 20.96 & .987 & 5.69 & .964 & yes \\
bimodal & 60 & 27.21 & .993 & 6.77 & .980 & yes \\
heavy-tailed & 10 & 11.41 & .893 & --- & --- & yes \\
heavy-tailed & 20 & 17.62 & .945 & --- & --- & yes \\
heavy-tailed & 40 & 13.57 & .962 & --- & --- & yes \\
heavy-tailed & 60 & 17.62 & .973 & --- & --- & yes \\
\bottomrule
\end{tabular}

\smallskip
\begin{minipage}{0.95\textwidth}
\footnotesize
\emph{Note.} Each grid evaluates the three indices on 4,001 density-weighted nodes with represented variance one at 91 log-spaced multipliers from .02 to 50. The maximum and its location are grid values, not population quantities. For the heavy-tailed population the finite-grid inverse-information curve is a truncation of a divergent expectation (\cref{prop:tail}) and is not reported. The mean-information curve increased across every grid.
\end{minipage}
\end{table}

%% file: tables/tab_guttman_limit.tex
\begin{table}[!htbp]
\centering
\caption{One 10-item 2PL form under a standard normal population at increasing multipliers}
\label{tab:guttman}
\small
\begin{tabular}{rrrrrr}
\toprule
$c$ & mean information $A(c)$ & $\rhotilde$ & $\rhopw$ & $\wbar$ (finite rule) & $\E[1/\Jinfo]$ (finite rule) \\
\midrule
1 & 1.8 & .6469 & .6419 & .6371 & .570 \\
5 & 13.8 & .9323 & .9064 & .1538 & 5.502 \\
10 & 28.2 & .9657 & .9278 & $5.91\times10^{-12}$ & $1.69\times10^{11}$ \\
20 & 56.7 & .9827 & .9131 & $3.00\times10^{-38}$ & $3.33\times10^{37}$ \\
50 & 142.0 & .9930 & .7819 & $1.02\times10^{-119}$ & $9.79\times10^{118}$ \\
100 & 284.0 & .9965 & .5226 & $8.77\times10^{-257}$ & $1.14\times10^{256}$ \\
500 & 1420.3 & .9993 & .1412 & 0 & $\infty$ \\
2000 & 5681.1 & .9998 & .0432 & 0 & $\infty$ \\
20000 & 56811.0 & 1.0000 & .0056 & 0 & $\infty$ \\
\bottomrule
\end{tabular}

\smallskip
\begin{minipage}{0.95\textwidth}
\footnotesize
\emph{Note.} Evaluated on 60,011 density-weighted nodes with $v=1$. As the multiplier grows the mean-information index approaches one (\cref{prop:overlap}), the pointwise index falls after its maximum (\cref{prop:vanish}), and the finite-rule inverse-information index collapses because the mean reciprocal information explodes; at $c=500$ the finite rule returns an infinite mean reciprocal. The multipliers beyond 10 lie outside any interval used in this paper.
\end{minipage}
\end{table}

%% file: tables/tab_tail_truncation.tex
\begin{table}[!htbp]
\centering
\caption{The truncation artifact: indices of one form under a standardized $t_5$ and a normal population as the quadrature range widens}
\label{tab:tail-trunc}
\small
\setlength{\tabcolsep}{4pt}
\begin{tabular}{rrrrrrrr}
\toprule
 & & \multicolumn{3}{c}{standardized $t_5$} & \multicolumn{3}{c}{standard normal} \\
\cmidrule(lr){3-5}\cmidrule(lr){6-8}
Range $\pm$ & excluded mass & $\rhotilde$ & $\rhopw$ & $\wbar$ (finite rule) & $\rhotilde$ & $\rhopw$ & $\wbar$ \\
\midrule
6 & $1.85\times10^{-3}$ & .7716 & .7642 & .7576 & .7697 & --- & .7612 \\
8 & $4.93\times10^{-4}$ & .7721 & .7645 & .7552 & .7697 & --- & .7612 \\
12 & $7.09\times10^{-5}$ & .7724 & .7648 & .7497 & .7697 & --- & .7612 \\
20 & $5.78\times10^{-6}$ & .7725 & .7649 & .7216 & .7697 & --- & .7612 \\
40 & $1.84\times10^{-7}$ & .7726 & .7650 & .0488 & .7697 & --- & .7612 \\
80 & $5.78\times10^{-9}$ & .7726 & .7650 & $1.03\times10^{-8}$ & .7697 & --- & .7612 \\
\bottomrule
\end{tabular}

\smallskip
\begin{minipage}{0.95\textwidth}
\footnotesize
\emph{Note.} The same 20-item 2PL form is evaluated on density-weighted grids truncated at $\pm$ the stated range, with the represented distribution restandardized to variance one. Under the normal population the three indices are stable in every decimal shown. Under the $t_5$ population the mean-information and pointwise indices are stable while the finite-rule inverse-information index falls toward zero as more of the tail is represented, which is the finite form of \cref{prop:tail}: its population value is zero at every range.
\end{minipage}
\end{table}

%% file: appendices/B_algorithms.tex
\section{Algorithm Details}
\label{app:algorithms}

We describe how the two calibration procedures find a multiplier, what their convergence results require, and how we assess a computed solution. EQC solves a fixed numerical objective. SAC updates the multiplier using fresh draws. Both require a specified population, variance rule and root-selection policy.

\subsection{Fixed-node calibration}
\label{app:sb01}

In \texttt{IRTsimrel} 0.3.1, \texttt{eqc\_calibrate()} draws the nodes once and uses their sample variance $\hat v$ throughout the solve (\cref{alg:eqc}). The procedure scans $\log c$ at a minimum of 65 initial points, with up to eight local refinement passes around detected sign changes and extrema. It then refines the selected crossing with a bracketed solver. The tolerance applies to $\log c$; the residual is the achieved index minus the target.

The statuses \texttt{ok} and \texttt{ok\_multiple\_roots} report a selected admissible crossing; the latter records a choice among several. An \texttt{infeasible} return gives the best detected multiplier, an endpoint or an interior extremum, when the target lies outside the detected range. Endpoint crossings, contacts without a sign change and flat regions receive \texttt{boundary}, \texttt{tangent} and \texttt{plateau} statuses; these do not count as success by default.

A finite scan can miss crossings between grid points. Finding one crossing in each of the 743 successful validation cells therefore does not establish global monotonicity (\cref{app:validation}).

\begin{algorithm}[ht]
\caption{Empirical quadrature calibration (fixed-node EQC)}
\label{alg:eqc}
\begin{algorithmic}[1]
\Require target $\rhotilde^*\in(0,1)$; population generator for $G$; realized form $\{(\beta_i,\lambda_{i,0})\}_{i=1}^I$ (and $g_i$ if 3PL); node count $M$; interval $[c_L,c_U]$; branch policy; tolerance
\State draw nodes $\theta_1,\dots,\theta_M\sim G$ once; set $\hat v\leftarrow$ sample variance of the nodes
\State define $F(c)\leftarrow \hat v\hat A_M(c)/\{1+\hat v\hat A_M(c)\}-\rhotilde^*$ with $\hat A_M(c)=M^{-1}\sum_m\Jinfo(\theta_m;c)$
\State evaluate $F$ on a grid in $\log c$ over $[\log c_L,\log c_U]$; refine locally around sign changes and extrema
\State record every sign change (crossing) and its direction (increasing or decreasing); record extrema and the range of $F$
\If{the target lies outside the range detected by the scan}
   \State \Return the multiplier at which the scan detected its best value (an endpoint or an interior extremum), that value, status \texttt{infeasible}, and a warning
\EndIf
\State select a crossing by the branch policy (default: the lowest increasing crossing)
\State refine the selected crossing with a bracketed root finder on $\log c$ to the tolerance; set $c^*$
\State \Return $c^*$, $\hat{\rhotilde}_M(c^*)$, the residual $F(c^*)$, the number and directions of crossings, the interval, the status, and $\lambda_i(c^*)=c^*\lambda_{i,0}$
\end{algorithmic}
\end{algorithm}

\subsection{Stochastic calibration and its target}
\label{app:sb02}

The function \texttt{sac\_calibrate()} first draws a presample to fix $\hat v$ and scan the objective. The scan selects an interval $[L,U]\subseteq[c_L,c_U]$ around the chosen crossing, with a small margin. SAC then updates the multiplier within that interval using a fresh batch of $M_{\text{iter}}$ draws at each step. With $q=\rhotilde^*/(1-\rhotilde^*)$, its score is
\begin{equation}
\label{eq:app-score}
U_n(c)=(1-\rhotilde^*)^2\big\{\hat v\,\hat A_n(c)-q\big\},
\end{equation}
and the update is given in \cref{eq:sac}. Conditional on the frozen inputs, the score has mean $(1-\rhotilde^*)^2\{\hat vA(c)-q\}$, so its zero is the target root with variance $\hat v$. The factor $(1-\rhotilde^*)^2$ matches the local slope of the transformed objective.

The default estimate averages the post-update iterates after discarding the first half, with $a=1$, $A=50$ and $\gamma=.67$. With \texttt{resample\_items = TRUE}, each batch draws a new form. The default target then transforms mean information across forms. The mean of their individual indices is reported separately and is calibrated only if requested (\cref{app:sa05}).

Release 0.3.0 used transformed batch reliability minus the target as its score. At finite batch sizes, that nonlinear transformation gives a biased score for the transformed population mean. The fixed-form SAC comparisons in \cref{app:sb06} use 0.3.0; they do not evaluate the score in \cref{eq:app-score}.

\Needspace*{6\baselineskip}
\subsection{Conditions for stochastic convergence}
\label{app:sb03}

The convergence argument is conditional on the presample variance and the selected branch. It requires the expected update to point toward one equilibrium throughout that branch.

\begin{proposition}[Projected stochastic approximation; conditional statement]
\label{prop:sac}
Condition on the frozen presample variance $\hat v$, the realized form (or form distribution), and the interval $C=[L,U]$ with $0<L<U<\infty$. Define the conditional drift $u(c)=(1-\rhotilde^*)^2\{\hat vA(c)-q\}$ and noise $\xi_{n+1}=U_n(c_n)-u(c_n)$. Assume:
\begin{enumerate}[label=(\roman*)]
\item $u$ is Lipschitz on $C$.
\item There is an equilibrium $c^\dagger\in C$ such that $(c-c^\dagger)u(c)>0$ for every other $c\in C$. If $c^\dagger$ is interior, $u(c^\dagger)=0$; if $c^\dagger=L$, $u(L)\ge0$; and if $c^\dagger=U$, $u(U)\le0$.
\item $\E[\xi_{n+1}\mid\mathcal F_n]=0$ and $\E[\xi_{n+1}^2\mid\mathcal F_n]\le K$ for every $n$, for a deterministic $K<\infty$.
\item $a_n>0$, $\sum_na_n=\infty$ and $\sum_na_n^2<\infty$.
\end{enumerate}
Then the projected iterates of \cref{eq:sac} converge almost surely to $c^\dagger$. The step sequence $a_n=a/(n+A)^\gamma$ satisfies condition (iv) when $\gamma\in(1/2,1]$.
\end{proposition}

\begin{proof}
Set $V_n=(c_n-c^\dagger)^2$ and $d_n=(c_n-c^\dagger)u(c_n)\ge0$. Projection onto an interval fixes $c^\dagger$ and is nonexpansive with respect to it, so expanding the unprojected square and using the martingale property gives $\E[V_{n+1}\mid\mathcal F_n]\le V_n-2a_nd_n+Ka_n^2$ for a finite $K$ by (i) and (iii). 

By (iv), $S_n=V_n+K\sum_{k\ge n}a_k^2$ is a nonnegative supermartingale, so $V_n$ converges almost surely and $\sum_na_nd_n<\infty$ almost surely. If the limit of $V_n$ were positive, the iterates would eventually stay a positive distance from $c^\dagger$, where (ii) and continuity bound $d_n$ below by a positive constant, contradicting $\sum_na_n=\infty$. 

The endpoint cases follow from the projection itself: at $L$ an outward step is clipped exactly when $u(L)\ge0$, and at $U$ exactly when $u(U)\le0$. The proof uses the projected-recursion framework of \textcite[Chap.~5]{kushner_stochastic_2003}. The underlying stochastic approximation method is due to \textcite{robbins_monro_1951}.
\end{proof}

The scan selects a neighborhood of an increasing crossing, but cannot prove that the whole interval satisfies condition (ii). For a fixed finite set of items, bounded information supplies condition (iii). Randomly redrawn forms need corresponding moment bounds under their generating distribution. These are conditions on the calibration problem, beyond the choice of step sizes.

\begin{theorem}[Interior tail-average central limit]
\label{thm:sac-pr}
Assume the conditions of \cref{prop:sac}, strengthened to a uniform conditional $(2+\delta)$-moment bound $\E[|\xi_{n+1}|^{2+\delta}\mid\mathcal F_n]\le K$ for some $\delta>0$. Suppose the equilibrium is an interior root $c^*\in(L,U)$, that $u$ is twice continuously differentiable near $c^*$ with $H=u'(c^*)>0$, that $\E[\xi_{n+1}^2\mid\mathcal F_n]\to\tau^2\in[0,\infty)$ almost surely, and that $1/2<\gamma<1$. 

Let $B_n$ be a deterministic burn-in with $0\le B_n<n$ and $B_n/n\to b\in[0,1)$, and put $N_n=n-B_n$. Then the average of the stored post-update iterates satisfies
\begin{equation}
\label{eq:app-clt}
\bar c_n^{\mathrm{post}}=\frac1{N_n}\sum_{k=B_n+1}^{n}c_{k+1},\qquad
\sqrt{N_n}\,(\bar c_n^{\mathrm{post}}-c^*)\ \Rightarrow\ N\!\left(0,\frac{\tau^2}{H^2}\right),
\end{equation}
conditionally on the frozen inputs. The statement excludes an equilibrium at an endpoint, $\gamma=1$, a trajectory-dependent burn-in, and a noise variance that remains random after conditioning.
\end{theorem}

\begin{proof}
We first control the iterate error, then show that its average has a linear expansion in the noise. Write $e_n=c_n-c^*$. By the attraction condition, $u(c)/(c-c^*)$ extends continuously to $H>0$ at the root and is positive elsewhere on $C$. It therefore has a positive minimum $h$, giving $(c-c^*)u(c)\ge h(c-c^*)^2$. The projection inequality gives
\[
d_{n+1}\le(1-2ha_n)d_n+Ka_n^2,\qquad d_n=\E[e_n^2].
\]
Since $(a_n-a_{n+1})/a_n^2\to0$ for $\gamma<1$, induction gives $\E[e_n^2]=O(a_n)$ and $\E|e_n|=O(a_n^{1/2})$. Local twice-continuous differentiability and compactness also give $|u(c)-H(c-c^*)|\le K_1(c-c^*)^2$ throughout $C$, after enlarging $K_1$.

Projection is eventually inactive almost surely. Indeed, $c_n\to c^*$ by \cref{prop:sac}, and $a_nu(c_n)\to0$. The moment bound and $\sum_na_n^{2+\delta}<\infty$ imply $a_n\xi_{n+1}\to0$ by Borel--Cantelli. Thus the unprojected proposal eventually stays inside $C$. Write $p_k$ for the projection correction, which is nonzero only finitely often.

Set $m=B_n+1$ and $R_k=u(c_k)-He_k$. Rearranging and summing the recursion gives
\[
H\sum_{k=m}^n e_k=-\sum_{k=m}^n\xi_{k+1}-\sum_{k=m}^n R_k
+\sum_{k=m}^n\frac{e_k-e_{k+1}}{a_k}
+\sum_{k=m}^n\frac{p_k}{a_k}.
\]
After division by $\sqrt{N_n}$, the projection sum vanishes almost surely. The curvature sum has expected absolute value
$O(N_n^{-1/2}\sum_{k\le n}a_k)=O(n^{1/2-\gamma})\to0$.
Abel summation rewrites the third sum as
\[
\frac{e_m}{a_m}-\frac{e_{n+1}}{a_n}
+\sum_{k=m+1}^ne_k\left(\frac1{a_k}-\frac1{a_{k-1}}\right).
\]
Its expected absolute value is $O(n^{\gamma/2})$, since
$\E|e_k|=O(k^{-\gamma/2})$ and
$1/a_k-1/a_{k-1}=O(k^{\gamma-1})$.
Dividing by $\sqrt{N_n}$ gives $O(n^{(\gamma-1)/2})\to0$.
The pre-update mean therefore satisfies
\[
\sqrt{N_n}\left(\frac1{N_n}\sum_{k=m}^nc_k-c^*\right)
=-\frac1{H\sqrt{N_n}}\sum_{k=m}^n\xi_{k+1}+o_p(1).
\]

The conditional variances averaged over the deterministic window converge to $\tau^2$ by Ces\`aro convergence. The $(2+\delta)$-moment bound gives the conditional Lindeberg condition, so the martingale central limit theorem gives the normal limit. Finally, the post-update and pre-update means differ by $(c_{n+1}-c_m)/N_n$, at most $(U-L)/N_n$ in absolute value. This is negligible at scale $\sqrt{N_n}$. The argument applies the averaging principle of \textcite{polyak_acceleration_1992} to the stated burn-in window and stored iterates.
\end{proof}

For fixed logistic items, the information function supplies the smoothness and noise bounds used above. The following result makes this part of the argument explicit.

\begin{proposition}[Regularity supplied by a fixed logistic bank]
\label{prop:envelope}
Fix a finite 2PL or 3PL bank and a compact interval $C=[L,U]\subset(0,\infty)$. For $j=0,1,2$ there are finite constants $K_j$, depending on the bank and $C$, such that $\sup_{c\in C,\theta}|\partial_c^j\Jinfo_c(\theta)|\le K_j$. Consequently $A(c)=\E_G[\Jinfo_c(\theta)]$ is twice continuously differentiable on $C$ under any law $G$, with derivatives taken under the expectation. If each batch consists of $m$ independent fresh draws from $G$, the score \cref{eq:app-score} has bounded conditional moments of every order and conditional variance
\begin{equation}
\label{eq:app-noise}
\E[\xi_{n+1}^2\mid\mathcal F_n]=\frac{(1-\rhotilde^*)^4\,\hat v^2}{m}\,\Var_G\{\Jinfo_{c_n}(\theta)\},
\end{equation}
which is continuous in $c_n$ and therefore converges to its value at $c^*$ whenever $c_n\to c^*$.
\end{proposition}

\begin{proof}
With $t=ca_i(\theta-b_i)$ and $h_g(t)=(1-g)s(t)^2\{1-s(t)\}/\{g+(1-g)s(t)\}$, \cref{eq:app-tif} gives $\Jinfo_{i,c}(\theta)=c^2a_i^2h_{g_i}(t)$. Writing $z=e^{t}$ for $t\to-\infty$ and $y=e^{-t}$ for $t\to\infty$ shows that $h_g$, $th_g'$ and $t^2h_g''$ are bounded on the real line for each $g\in[0,1)$, because $h_g$ is $z^2$ (or $z$ when $g=0$) times a function analytic at zero on one side and $y$ times such a function on the other, and $d/dt$ acts as $z\,d/dz$ or $-y\,d/dy$. 

Direct differentiation gives $\partial_c\Jinfo_{i,c}=ca_i^2\{2h+th'\}$ and $\partial_c^2\Jinfo_{i,c}=a_i^2\{2h+4th'+t^2h''\}$ evaluated at $t$, and the fixed bank, finite sum and compact interval give the envelopes. Dominated differentiation under any $G$ follows. 

For independent batches, $\hat A_n(c)$ is bounded by $K_0$, so all conditional moments of the score are bounded; the variance of a mean of $m$ independent information values is $\Var_G(\Jinfo_c)/m$, and multiplying by the squared score coefficient $(1-\rhotilde^*)^2\hat v$ gives \cref{eq:app-noise}. Boundedness and pointwise continuity of $\Jinfo_c$ in $c$ make $\E_G[\Jinfo_c]$ and $\E_G[\Jinfo_c^2]$ continuous, hence the variance.
\end{proof}

\subsection{Checking a finite stochastic run}
\label{app:sb03-checks}

The asymptotic results do not certify a particular run. We check projections onto the bounds, compare averages from the two halves of the post-burn-in window, and require an interior final average with positive local slope. We then evaluate the returned multiplier independently. Projection records distinguish the requested bounds from the narrower branch bounds and record consecutive hits and the post-burn-in hit rate.

A run that fails these checks is classified as not converged. Aggregate summaries can hide an unsuccessful run: among 413 finite projection diagnostics for inverse-information calibration, one run lost its branch with projection proportion .284, although the median was zero. For redrawn forms, the theoretical bounds also require integrable envelopes over the form distribution.

\subsection{Sampling error in a fixed-node root}
\label{app:sb04}

EQC has no computational randomness after its nodes are fixed, but a new node sample generally gives a different root. The next result describes that sampling error for independent draws from $G$, conditional on the items. It also shows why estimating the population variance contributes to the error. Weighted or sample-standardized nodes require a different argument.

\begin{theorem}[Branch-local consistency and limit of the fixed-node root]
\label{thm:eqc}
Fix a finite 2PL or 3PL bank and let $\theta_1,\theta_2,\dots$ be independent draws from $G$ with mean $\mu$, variance $v\in(0,\infty)$ and finite fourth moment. Fix $\rhotilde^*\in(0,1)$, put $q=\rhotilde^*/(1-\rhotilde^*)$, and let $K=[\ell,u]\subset(0,\infty)$ be a compact interval on which $F(c)=vA(c)-q$ has a unique zero $c^*$, interior to $K$, with $A'(c^*)\ne0$. Let $\hat v_M$ be the sample variance of the first $M$ draws and $\hat A_M(c)=M^{-1}\sum_{m\le M}\Jinfo_c(\theta_m)$. 

Then almost surely, for all large $M$, the equation $\hat v_M\hat A_M(c)=q$ has a root in $K$; any measurable selection $\hat c_M$ of such roots converges almost surely to $c^*$; and
\begin{align}
\label{eq:app-eqc-limit}
\sqrt M(\hat c_M-c^*)&\ \Rightarrow\ N\!\left(0,\frac{\Omega}{\{vA'(c^*)\}^2}\right),\notag\\
\Omega&=\Var_G\!\Big[v\{\Jinfo_{c^*}(\theta)-A(c^*)\}+A(c^*)\{(\theta-\mu)^2-v\}\Big].
\end{align}
If the population variance is supplied instead of estimated, the second term inside the variance is absent.
\end{theorem}

\begin{proof}
The envelopes in \cref{prop:envelope} bound $\Jinfo_c(\theta)$ and its derivatives uniformly on $K$. A finite-net argument and the strong law give uniform almost-sure convergence of $\hat A_M$ to $A$ and $\hat A_M'$ to $A'$. Also, $\hat v_M\to v$ almost surely. Hence $\hat F_M(c)=\hat v_M\hat A_M(c)-q$ and its derivative converge uniformly to $F(c)$ and $vA'(c)$.

Because $A'(c^*)\ne0$, $F$ changes sign across the root. Uniform convergence eventually preserves its signs at fixed points on either side, and continuity gives an empirical root between them. The population root is unique on $K$, so every selected empirical root converges to $c^*$.

For the limit distribution, the sample variance has expansion
\[
\sqrt M(\hat v_M-v)
=M^{-1/2}\sum_m\{(\theta_m-\mu)^2-v\}+o_p(1).
\]
Expanding the product on the same draws gives
\[
\sqrt M\,\hat F_M(c^*)=M^{-1/2}\sum_{m=1}^M
\Big[v\{\Jinfo_{c^*}(\theta_m)-A(c^*)\}
+A(c^*)\{(\theta_m-\mu)^2-v\}\Big]+o_p(1).
\]
The central limit theorem applies by the finite fourth moment and bounded information. Its variance is $\Omega$, including the covariance between information and squared ability deviation.

The mean value theorem gives, for some $\tilde c_M$ between $\hat c_M$ and $c^*$,
\[
\sqrt M(\hat c_M-c^*)
=-\frac{\sqrt M\,\hat F_M(c^*)}{\hat F_M'(\tilde c_M)}.
\]
Since $\tilde c_M\to c^*$, uniform derivative convergence and Slutsky's lemma give the result.
\end{proof}

The theorem applies to an exact sample root. For a computed root with residual $e_M=\hat F_M(c_M^{\mathrm{num}})$, a nonzero local derivative gives $c_M^{\mathrm{num}}-\hat c_M\approx e_M/\{vA'(c^*)\}$ to first order. Consistency therefore requires $e_M\to0$ (almost surely for the strong form), and the limit distribution requires $e_M=o_p(M^{-1/2})$. A fixed tolerance in $\log c$ guarantees neither as $M\to\infty$.

In the tested range, numerical error was negligible compared with sampling error. Solver residuals were near $10^{-9}$, five to six orders of magnitude below fresh-draw errors. Mean absolute fresh-draw error fell from .0117 at $M=500$ to .0014 at $M=50{,}000$, a factor of 8.6 compared with the square-root prediction of 10. This pattern is consistent with the theorem; the evaluation also contains error from its finite fresh sample.

\subsection{Numerical comparisons and sensitivity to settings}
\label{app:sb06}
\label{app:theta-var}

We report the calibration comparisons in \cref{tab:s2-accuracy}, separating implementations and evaluation rules. Under release 0.3.0, fixed-form SAC targeting $\rhotilde$ converged in 733 of 748 cells, with mean absolute discrepancy .0029 on fresh evaluation blocks. This discrepancy includes finite-iteration error, transformed-score bias, presample-variance error and evaluation noise. Inverse-information calibration succeeded in 407 of 561 eligible cells.

A separate 0.3.1 run with redrawn forms completed 377 of 380 conditions, with three input-check errors and mean absolute discrepancy .0049. Neither discrepancy is directly comparable with the EQC residual, which measures error on the solver's own nodes.

\input{floats/tab_s2_accuracy}

\paragraph{Starting values.}
The main SAC comparison started at the EQC solution, so it measures drift from a known root. We also ran SAC without a supplied start (\cref{tab:s2-coldstart}). Across 729 cells with paired outputs, the median SAC-to-EQC multiplier ratio was 1.0007 with warm starts and 1.0063 with cold starts. The 10th and 90th percentiles were .868 and 1.142 (warm) and .814 and 1.200 (cold); median absolute multiplier differences were .048 and .070. Regressing the warm-start multiplier on the EQC multiplier gave slope .983 and correlation .989. Cold starts usually reached the same neighborhood, with greater variability.

\input{floats/tab_s2_coldstart}

\paragraph{Iterations and runtime.}
Across 60 cells, nine of ten step configurations converged in every cell with absolute multiplier bias below .02 (\cref{tab:s3-step}). The 100-iteration configuration had bias .059 and convergence fraction .717. The fractional design does not identify all interactions among step scale, offset and exponent or establish optimal settings; it supports the tested defaults and indicates that 100 iterations can be insufficient. In the runtime comparison, EQC was faster at 10 and 20 items and SAC with its fixed iteration budget was faster at 40 and 60 items (\cref{tab:runtime}).

\input{floats/tab_s3_step}

\input{tables/tab_runtime}

\paragraph{Variance rule.}
With the bank, multiplier and information values fixed, replacing the calibration sample variance by the population variance one shifted the mean-information index by a median of $6\times10^{-4}$ to $2\times10^{-3}$ across shapes, and by up to .033 in a heavy-tailed cell (\cref{tab:theta-var}). These shifts exceed the solver residual by five to six orders of magnitude. Thus, even when every population has variance one, the numerical variance rule matters.

\input{floats/tab_theta_var}

\paragraph{Complete settings.}
\Cref{tab:settings} records the study settings and the corresponding package defaults. Ten of eighteen settings differ: for example, we used 20,000 nodes instead of 10,000 and $[0.1,10]$ instead of $[0.3,3]$, with root tolerance $10^{-6}$. We did not optimize these settings. \Cref{tab:estimand-conventions} separates the variance rules used for calibration and evaluation.

\input{floats/tab_settings}

\input{floats/tab_estimand_conventions}

%% file: floats/tab_s2_accuracy.tex
\begin{table}[htbp]\centering
\caption{Calibration accuracy under each procedure and evaluation rule}
\label{tab:s2-accuracy}
{\footnotesize\resizebox{\textwidth}{!}{\begin{tabular}{>{\raggedright\arraybackslash}p{0.17\textwidth}rrrrrr}
\toprule
Algorithm & Successful/attempted & Mean dev. & MAE & Max abs. & Within $10^{-2}$ & Median s \\
\midrule
EQC & 743/748 & $1.29\times 10^{-9}$ & $1.54\times 10^{-9}$ & $3.81\times 10^{-8}$ & $1.000$ & $26.8$ \\
SAC (rho-tilde) & 733/748 & $-1.45\times 10^{-4}$ & $0.0029$ & $0.0394$ & $0.914$ & $24.0$ \\
SAC (superpopulation) & 377/380 & $-3.44\times 10^{-4}$ & $0.0049$ & $0.0292$ & $0.865$ & $17.6$ \\
SAC (w-bar) & 407/561 & $-0.0013$ & $0.0021$ & $0.0367$ & $0.936$ & $28.0$ \\
\bottomrule
\end{tabular}}}
\floatnote{Fixed-form EQC and SAC use version 0.3.0. EQC error uses its own nodes; SAC error uses fresh nodes and frozen presample variance. The inverse-information row includes 407 successes among 561 production cells, distinct from the 576-cell prescreen. The version 0.3.1 form-population study attempted 380 length-20/40 conditions after 368 exclusions and targets transformed mean information across forms. Error and runtime summaries condition on success, so the rows are not direct accuracy rankings.}
\end{table}

%% file: floats/tab_s2_coldstart.tex
\begin{table}[htbp]\centering
\caption{SAC multipliers with and without an EQC starting value}
\label{tab:s2-coldstart}
{\small\begin{tabular}{lrrrrr}
\toprule
Start & Cells & Median ratio & 10th & 90th & Median $|\Delta|$ \\
\midrule
Cold & 729 & $1.0063$ & $0.814$ & $1.200$ & $0.0703$ \\
Warm (as run) & 729 & $1.0007$ & $0.868$ & $1.142$ & $0.0479$ \\
\bottomrule
\end{tabular}}
\floatnote{Both rows use the same 729 cells under version 0.3.0. Warm starts use the EQC solution; cold starts supply no starting value. The ratio is SAC divided by EQC, and $|\Delta|$ is their absolute difference. Quantiles describe variation across cells; each arm follows its own stochastic path.}
\end{table}

%% file: floats/tab_s3_step.tex
\begin{table}[htbp]\centering
\caption{Sensitivity to ten SAC settings across 60 design cells}\label{tab:s3-step}
{\footnotesize\setlength{\tabcolsep}{4pt}\begin{tabular}{lrrrrrrrrrr}
\toprule
Setting & $a$ & $A$ & $\gamma$ & Iterations & Draws & Bias & $\mathrm{RMSE}/\sqrt{60}$ & RMSE & Converged & Median s \\
\midrule
cfg06 & $1.0$ & $50$ & $0.55$ & 1000 & 1000 & $-0.001$ & $0.011$ & $0.086$ & $1.000$ & $20.3$ \\
cfg10 & $1.0$ & $50$ & $0.67$ & 1000 & 500 & $0.002$ & $0.012$ & $0.089$ & $1.000$ & $10.3$ \\
cfg03 & $2.0$ & $50$ & $0.67$ & 1000 & 1000 & $-0.004$ & $0.010$ & $0.078$ & $1.000$ & $20.1$ \\
cfg02 & $0.5$ & $50$ & $0.67$ & 1000 & 1000 & $-0.007$ & $0.011$ & $0.082$ & $1.000$ & $20.1$ \\
cfg01 & $1.0$ & $50$ & $0.67$ & 1000 & 1000 & $-0.008$ & $0.011$ & $0.082$ & $1.000$ & $20.0$ \\
cfg05 & $1.0$ & $200$ & $0.67$ & 1000 & 1000 & $0.008$ & $0.010$ & $0.076$ & $1.000$ & $20.0$ \\
cfg04 & $1.0$ & $10$ & $0.67$ & 1000 & 1000 & $-0.010$ & $0.010$ & $0.076$ & $1.000$ & $19.7$ \\
cfg09 & $1.0$ & $50$ & $0.67$ & 300 & 1000 & $0.017$ & $0.010$ & $0.077$ & $1.000$ & $6.5$ \\
cfg07 & $1.0$ & $50$ & $0.90$ & 1000 & 1000 & $0.017$ & $0.013$ & $0.101$ & $1.000$ & $20.1$ \\
cfg08 & $1.0$ & $50$ & $0.67$ & 100 & 1000 & $0.059$ & $0.013$ & $0.102$ & $0.717$ & $3.5$ \\
\bottomrule
\end{tabular}}
\floatnote{Rows are ordered by absolute mean multiplier error relative to the reference solution. Converged is the recorded success fraction. The column $\mathrm{RMSE}/\sqrt{60}$ is a scaled RMSE, not a centered Monte Carlo standard error. The fractional design does not separate all tuning effects or identify optimal settings.}
\end{table}

%% file: tables/tab_runtime.tex
\begin{table}[!htbp]
\centering
\caption{Median elapsed seconds per calibration by algorithm and test length}
\label{tab:runtime}
\small
\begin{tabular}{lrrrr}
\toprule
Algorithm & 10 items & 20 items & 40 items & 60 items \\
\midrule
EQC (0.3.0, fixed form) & 11.3 & 18.1 & 32.7 & 46.9 \\
SAC, $\rhotilde$ (0.3.0, fixed form) & 13.5 & 19.4 & 30.1 & 40.2 \\
SAC, $\wbar$ (0.3.0, fixed form) & 13.3 & 19.2 & 30.2 & 40.9 \\
SAC, form population (0.3.1) & --- & 12.5 & 18.1 & --- \\
\bottomrule
\end{tabular}

\smallskip
\begin{minipage}{0.95\textwidth}
\footnotesize
\emph{Note.} Medians condition on success (63--191 runs per cell) on the recorded 30-worker setup. EQC uses 20,000 nodes. Fixed-form SAC uses a 20,000-node presample and 1,000 iterations of 1,000 draws each. The form-population study includes only 20 and 40 items. Runtime depends on this implementation and computing setup.
\end{minipage}
\end{table}

%% file: floats/tab_theta_var.tex
\begin{table}[htbp]\centering
\caption{Changes in the index when sample variance is replaced by population variance one}
\label{tab:theta-var}
{\small\begin{tabular}{lrrrrr}
\toprule
Latent shape & Cells & Median shift & Largest shift & Median solver error & Ratio \\
\midrule
Bimodal & 186 & $6.42\times 10^{-4}$ & $0.0058$ & $9.32\times 10^{-10}$ & $6.89\times 10^{+5}$ \\
Heavy-Tailed & 187 & $0.0018$ & $0.0326$ & $1.21\times 10^{-9}$ & $1.51\times 10^{+6}$ \\
Normal & 187 & $8.46\times 10^{-4}$ & $0.0077$ & $9.62\times 10^{-10}$ & $8.79\times 10^{+5}$ \\
Skewed & 188 & $0.0012$ & $0.0074$ & $1.03\times 10^{-9}$ & $1.12\times 10^{+6}$ \\
\bottomrule
\end{tabular}}
\floatnote{The form, multiplier and information values are held fixed. Rows summarize absolute shifts by ability shape. The ratio compares the median shift with the median solver residual; it describes these runs rather than a general error relationship.}
\end{table}

%% file: floats/tab_settings.tex
\begin{table}[htbp]\centering
\caption{Study settings and recorded package defaults}\label{tab:settings}
{\scriptsize\setlength{\tabcolsep}{3pt}\begin{adjustbox}{max width=\textwidth}\begin{tabular}{>{\raggedright\arraybackslash}p{0.13\textwidth}>{\raggedright\arraybackslash}p{0.19\textwidth}>{\raggedright\arraybackslash}p{0.16\textwidth}>{\raggedright\arraybackslash}p{0.36\textwidth}c}
\toprule
Algorithm & Setting & Package default & Used here & Differs \\
\midrule
EQC & quadrature size M & \texttt{10000L} & \texttt{20000} & yes \\
EQC & search interval for c & \texttt{c(0.3,} \texttt{3)} & \texttt{c(0.1,} \texttt{10)} & yes \\
EQC & root tolerance & \texttt{1e-04} & \texttt{1e-06} & yes \\
EQC & root policy & \code{"lowest\_increasing"} & \code{lowest\_increasing} & -- \\
SAC & presample size M\_pre & \texttt{10000L} & \texttt{20000} & yes \\
SAC & draws per iteration M\_per\_iter & \texttt{500L} & \texttt{1000} & yes \\
SAC & iterations n\_iter & \texttt{300L} & \texttt{1000} & yes \\
SAC & burn-in & \texttt{floor(n\_iter} \texttt{/} \texttt{2)} & \texttt{floor(n\_iter} \texttt{/} \texttt{2)} & -- \\
SAC & step size (a, A, gamma) & \texttt{(1,} \texttt{50,} \texttt{0.67)} & \texttt{(1,} \texttt{50,} \texttt{0.67)} & -- \\
SAC & projection interval for c & \texttt{c(0.01,} \texttt{20)} & \texttt{c(0.05,} \texttt{20)} & yes \\
SAC & item resampling & \texttt{TRUE} & \texttt{FALSE} \texttt{for} \texttt{the} \texttt{fixed-form} \texttt{arms;} \texttt{TRUE} \texttt{for} \texttt{the} \code{superpopulation} \texttt{arm} & yes \\
SAC (item superpopulation) & item aggregation & \texttt{NULL} & \code{functional\_of\_mean\_information} & yes \\
SAC (item superpopulation) & reported reliability outputs & \texttt{--} & \code{rho\_at\_expected\_information} \texttt{(primary);} \code{mean\_random\_form\_reliability} \code{(descriptive)} & -- \\
SAC & starting value & \texttt{NULL} & \texttt{the} \texttt{EQC} \texttt{solution} \texttt{(warm} \texttt{start)} & yes \\
both & target ladder & \texttt{--} & \texttt{0.5,} \texttt{0.6,} \texttt{0.7,} \texttt{0.8,} \texttt{0.9,} \texttt{0.95} & -- \\
both & replications per calibrated cell & \texttt{--} & \texttt{1000} & -- \\
both & sample sizes in the replication stage & \texttt{--} & \texttt{100,} \texttt{250,} \texttt{500,} \texttt{1000,} \texttt{2000} & -- \\
both & global seed & \texttt{--} & \texttt{20260826} & -- \\
\bottomrule
\end{tabular}\end{adjustbox}}
\floatnote{Fixed-form runs use version 0.3.0; the form-population extension uses 0.3.1. Defaults are version-specific. SAC may narrow the requested interval during branch selection. Its main fixed-form comparison starts at the EQC solution. The 1,000 replications and sample-size ladder refer to person redraws after calibration, not independently calibrated forms.}
\end{table}

%% file: floats/tab_estimand_conventions.tex
\begin{table}[htbp]\centering
\caption{Variance and integration rules for calibration and evaluation}\label{tab:estimand-conventions}
{\footnotesize\begin{tabular}{>{\raggedright\arraybackslash}p{0.2\textwidth}>{\raggedright\arraybackslash}p{0.31\textwidth}>{\raggedright\arraybackslash}p{0.36\textwidth}}
\toprule
Context & Variance used & Interpretation \\
\midrule
Population functional & $v=\operatorname{Var}_G(\theta)$ & A property of the stated latent distribution \\
Fixed-node EQC & $n-1$ sample variance of its fixed node draw & Part of the retained numerical objective; the node rule approximates the population integral \\
Fixed-form SAC & Presample variance, frozen through updates & Fresh stochastic information batches retain the same variance factor \\
Weighted quadrature & Normalized weighted population variance & A discrete weighted distribution; truncation and discretization can affect its moments \\
S2 fresh-node evaluation & The calibration presample variance & Changes the information evaluation rule while retaining the variance convention \\
S3 holdout evaluation & $n-1$ sample variance of the holdout draw & Changes both the information rule and its empirical variance factor \\
Worked software example & Calibration uses its node variance; holdout sets $v=1$ & The known standard-normal population variance is supplied explicitly for the holdout \\
Realized response sample & $n-1$ variance where a sample variance is reported & An empirical statistic; fitted-score reliability has its own definition \\
\bottomrule\end{tabular}}
\floatnote{A population variance of one does not force finite draws or truncated quadrature to have variance one. \Cref{app:theta-var} compares the variance rules; \cref{app:sa04} defines score reliability.}
\end{table}

%% file: appendices/C_validation.tex
\section{Validation Design and Full Results}
\label{app:validation}

We distinguish three sources of discrepancy: numerical error on the calibration nodes, error on an independent ability sample, and variation among person samples after a form has been calibrated. The studies below address these in turn. \Cref{app:algorithms} reports the additional SAC comparisons, and \cref{app:illustration} gives the item-level treatment-effect design.

\subsection{Study designs and software}
\label{app:sc01}

\Cref{tab:app-blocks} summarizes the fixed-form validation and fresh-draw studies. The first asks how accurately EQC solves its numerical objective. The second asks how well the resulting multiplier meets the target on new ability draws. Both use the package through the study wrapper.

\begin{table}[!htb]
\centering
\caption{Designs and evaluation rules for the two calibration studies}
\label{tab:app-blocks}
\footnotesize
\setlength{\tabcolsep}{5pt}
\begin{tabularx}{\textwidth}{@{}>{\raggedright\arraybackslash}p{0.17\textwidth}XX@{}}
\toprule
 & Fixed-form validation & Fresh-draw evaluation \\
\midrule
Software & \texttt{IRTsimrel} 0.3.0, \texttt{eqc\_calibrate()} & Same package and run series \\
Item models & Rasch, 2PL, 3PL with $g=.20$, 3PL with beta-distributed $g$ & Rasch, 2PL, 3PL with $g=.20$ \\
Ability distributions & Normal, bimodal, positively skewed, heavy-tailed & Normal and positively skewed \\
Difficulties & Parametric and Item Response Warehouse & Parametric \\
Test lengths & 10, 20, 40, 60 & 10, 40 \\
Targets & .50, .60, .70, .80, .90, .95 & .70 \\
Calibration nodes & 20,000 per cell & 500, 1,000, 2,500, 5,000, 10,000, 20,000, 50,000 \\
Search & $[0.1,10]$; lowest increasing crossing; tolerance $10^{-6}$ on $\log c$ & Same interval, policy and tolerance \\
Variance & Calibration-node sample variance & Calibration-node sample variance \\
Runs & 768 nominal cells; 748 after screening; 743 successful & 12 cells $\times$ 20 seeds $\times$ 7 node counts $=$ 1,680 calibrations \\
Evaluation & Residual on the calibration nodes & Index on 200,000 fresh draws, using their sample variance \\
\bottomrule
\end{tabularx}
\end{table}

\subsection{Generators}
\label{app:sc02}

The four ability distributions have population mean zero and variance one: standard normal; an equal mixture of $N(-.8,.6^2)$ and $N(.8,.6^2)$; $(Z-4)/2$ with $Z\sim\operatorname{Gamma}(4,1)$; and $t_5/\sqrt{5/3}$. \Cref{app:shapes} gives their properties. Finite draws need not have sample variance one, so calibration and evaluation each specify a variance rule.

The parametric generator draws $N(0,1)$ difficulties and log-normal baseline discriminations with log-scale standard deviation .3. A rank-based Gaussian copula imposes dependence parameter $-.3$ between the two \parencite{nelsen_introduction_2006}. Rasch forms instead set every baseline discrimination to one. The empirical option samples difficulties from 6,143 estimates in 145 Item Response Warehouse data sets \parencite{domingue_introduction_2025,zhang_realistic_2025}. It supplies difficulties only; the generator supplies the other parameters and their dependence (\cref{app:sf04}).

The two 3PL variants set guessing to .20 or draw it from $\operatorname{Beta}(5,17)$. Guessing parameters remain fixed during calibration.

\subsection{Screening targets before calibration}
\label{app:sc03b}

Before production, we screened the 768 nominal cells using one form and 20,000 nodes per structural cell. A target passed if it lay strictly between the mean-information values at the endpoints of $[0.1,10]$. Twenty targets failed: one at .90 and nineteen at .95. This left 748 production cells.

The same screen excluded 167 of 576 inverse-information targets, including 79 of 96 at .95 (\cref{tab:s2-reachability-unpruned}). We omitted the heavy-tailed population from this comparison because its inverse-information moment is infinite. An endpoint screen can also miss an interior crossing of a nonmonotone curve. These exclusions therefore describe the screening rule, rather than proving global infeasibility. \Cref{tab:s2-reachability,fig:accuracy-reachability} distinguish this screen from the subsequent production results.

\input{floats/tab_s2_reachability_unpruned}

\input{floats/tab_s2_reachability}

\begin{figure}[ht]
\centering
\includegraphics[width=\textwidth]{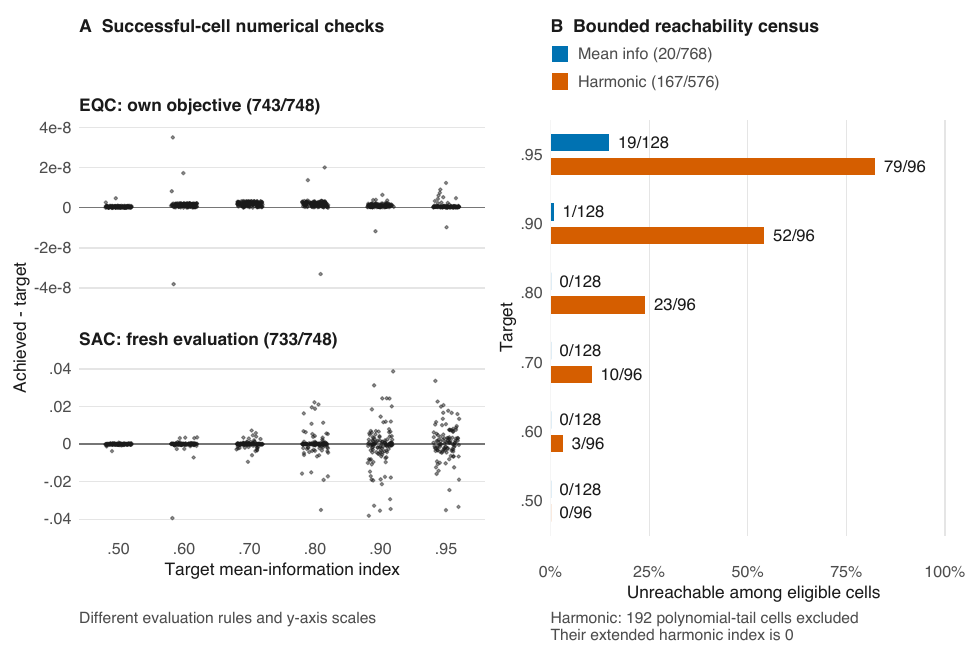}
\caption{Calibration discrepancies under each procedure's evaluation rule, and targets excluded by the preliminary screen}
\label{fig:accuracy-reachability}

\smallskip
\begin{minipage}{0.95\textwidth}
\footnotesize
\emph{Note.} A: achieved minus target for 743 successful EQC calibrations on their own nodes, and 733 successful SAC calibrations from version 0.3.0, evaluated on 20 fresh blocks of 1,000 nodes with frozen presample variance. The procedures use different evaluation rules and vertical scales; SAC uses the earlier transformed-batch score. B: requests failing the strict endpoint screen on $[0.1,10]$ with 20,000 nodes (128 cells per target for mean information, 96 for inverse information after excluding heavy tails). These preliminary exclusions are separate from the five infeasible production cells.
\end{minipage}
\end{figure}

\subsection{Multipliers and infeasible cells}
\label{app:sc03}

Among production runs, 743 met the target and five did not. All five infeasible cells used 3PL forms at the two highest targets (\cref{tab:app-fail}). Production generated new forms and nodes with target-specific seeds, so a target that passed the preliminary screen could fail for its production form.

The required multiplier varied across designs (\cref{fig:cstar}). At target .70 with 20 items, it differed by more than a factor of two across shapes within some model-and-pool combinations; parametric Rasch forms were more similar (\cref{tab:app-cstar}). Because each cell used a separate form, these differences combine population shape with realized item differences.

\begin{table}[ht]
\centering
\caption{Multiplier $c^*$ returned by EQC at target .70 with 20 items, by item model, difficulty pool and latent shape}
\label{tab:app-cstar}
\small
\input{tables/tab_sc03_cstar_by_shape}
\end{table}

\begin{table}[ht]
\centering
\caption{The five production cells in which the target was not reached within $[0.1,10]$}
\label{tab:app-fail}
\small
\input{tables/tab_sc03_failures}
\end{table}

\subsection{Accuracy on fresh ability draws}
\label{app:sc04}

For each of the 12 model-by-shape-by-length cells, we used 20 seeds at each of seven node counts, giving 240 calibrations per node count. Every run generated a new form and calibration sample, targeted .70, and evaluated the calibrated form on 200,000 fresh draws. Each calculation used its own sample variance.

Fresh-draw errors and their spread decreased as the number of calibration nodes increased. Solver residuals stayed near $10^{-9}$, and median runtime grew roughly in proportion to node count (\cref{tab:app-holdout}). The first tested node count with mean fresh-draw error below .005 was 2,500. This is an average for the tested grid, not an error guarantee for each form. Multiplier variation across seeds remained stable because the seeds changed the forms as well as the nodes.

Skewed populations had larger errors at every node count (\cref{fig:m-sensitivity,tab:holdout-cell}). Length showed no consistent ordering: 40-item forms had larger pooled errors at five of seven node counts, with the direction also varying within model-by-shape groups. At this target, 3PL forms were not systematically harder to calibrate. These comparisons show why node count should be checked for the intended design rather than chosen from test length alone.

\begin{table}[ht]
\centering
\caption{Fresh-node study: error against 200,000 fresh draws by node count (240 calibrations per $M$)}
\label{tab:app-holdout}
\small
\input{tables/tab_sc04_holdout_by_M}
\end{table}

\begin{figure}[ht]
\centering
\includegraphics[width=\textwidth]{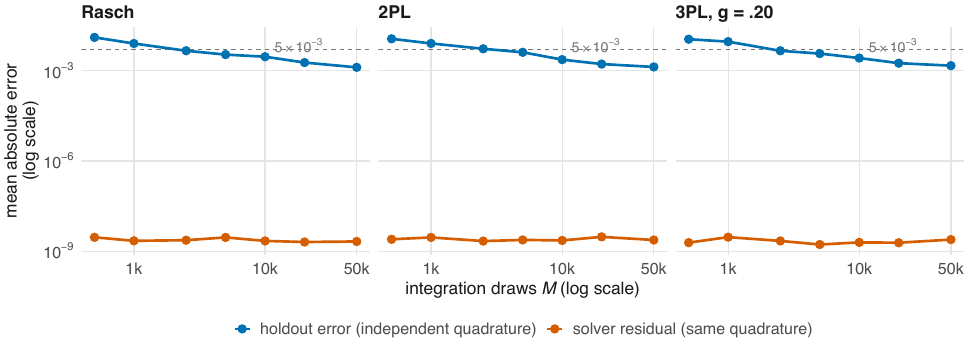}
\caption{Solver residual and fresh-draw error by node count, separately for the three item models}
\label{fig:m-sensitivity}

\smallskip
\begin{minipage}{0.95\textwidth}
\footnotesize
\emph{Note.} Target .70; each point aggregates 80 runs (two latent shapes, two test lengths, 20 seeds). The fresh-draw error uses 200,000 fresh draws and their sample variance; the residual uses the calibration draw and its variance. The dashed line is the descriptive .005 reference. Both axes are logarithmic.
\end{minipage}
\end{figure}

\input{tables/tab_holdout_by_cell}

\subsection{Variation after calibration when persons are redrawn}
\label{app:sc05}

Calibration fixes the form, but information evaluated in a new person sample still varies. For each of the 743 calibrated forms, we drew 1,000 independent samples at sizes 100, 250, 500, 1,000 and 2,000, and recomputed the index from each sample's information and variance.

The median within-form standard deviation fell from .0256 at 100 persons to .0057 at 2,000, within a few percent of the $N^{-1/2}$ reference at every size (\cref{tab:s2-replication,fig:person-replication}). Median bias relative to the achieved calibration index fell from $-.0026$ to $-.0001$; it was about three times larger under heavy tails. These results condition on the calibrated forms. They exclude variation from generating and calibrating new forms.

\input{floats/tab_s2_replication}

\begin{figure}[ht]
\centering
\includegraphics[width=\textwidth]{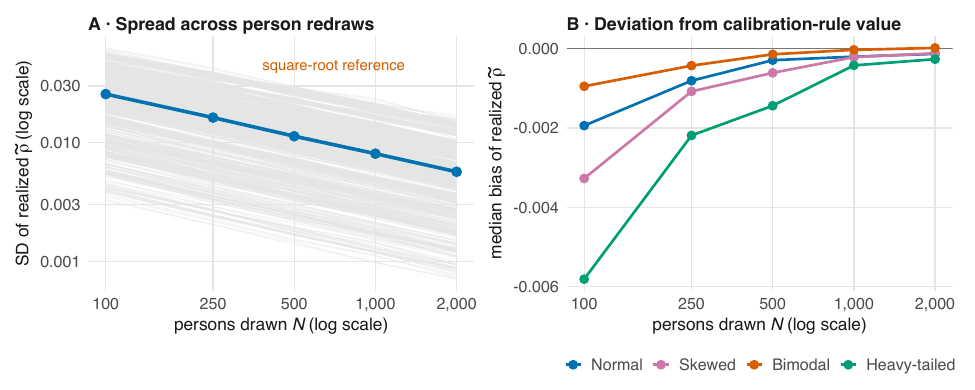}
\caption{Person-sampling variation of the realized mean-information index after calibration}
\label{fig:person-replication}

\smallskip
\begin{minipage}{0.95\textwidth}
\footnotesize
\emph{Note.} Each of the 743 calibrated banks has 1,000 independent person redraws at each of five sample sizes. A: grey curves are within-bank standard deviations of the realized index and the blue curve is their median across banks; the dashed reference scales the median at $N=100$ by $\sqrt{100/N}$. B: within-shape medians of the redraw mean minus the bank's calibration-rule value. No confidence intervals are plotted.
\end{minipage}
\end{figure}

\subsection{Guessing and the choice of information index}
\label{app:sc06}

With the other item parameters fixed, guessing reduces information at every ability (\cref{eq:app-tif}). On an increasing branch, a larger multiplier is therefore needed to reach the same target. In cells matched on target, shape and length, median 3PL-to-2PL multiplier ratios were 1.31 for fixed guessing and 1.41 for beta-distributed guessing; the Rasch-to-2PL ratio was 1.07 (\cref{tab:c-inflation}). Individual ratios can fall below one because different models used different forms.

Index ordering and root ordering are separate questions. Among 407 cells with both mean-information and inverse-information solutions, the inverse-information root was larger in 76\% of cells in the widest quartile of $\rhotilde-\wbar$, compared with 48\% in the narrowest (\cref{tab:jensen-gap}). Stochastic error and branch differences were comparable to the gaps in the latter group. The index ordering held in all 743 evaluated cells, but need not order solutions from different procedures or branches.

\input{floats/tab_c_inflation}

\input{floats/tab_jensen_gap}

\subsection{Interpreting the validation results}
\label{app:sc07}

A fresh-draw discrepancy includes uncertainty in mean information and in the fresh sample's variance. Supplying population variance one would change that discrepancy (\cref{tab:theta-var}). The fixed-form SAC comparisons also use the earlier 0.3.0 score, as explained in \cref{app:algorithms}. These studies assess calibration of information: they generate no responses and fit no scoring model. The separate score-reliability comparison is in \cref{app:sa04}.

%% file: floats/tab_s2_reachability_unpruned.tex
\begin{table}[htbp]\centering
\caption{Targets excluded by the preliminary endpoint screen}
\label{tab:s2-reachability-unpruned}
{\small\begin{tabular}{lrrrrrr}
\toprule
 & \multicolumn{3}{c}{$\rhotilde$} & \multicolumn{3}{c}{$\wbar$} \\
\cmidrule(lr){2-4}\cmidrule(lr){5-7}
Target & Cells & Unreachable & Share & Cells & Unreachable & Share \\
\midrule
$0.50$ & 128 & 0 & $0.000$ & 96 & 0 & $0.000$ \\
$0.60$ & 128 & 0 & $0.000$ & 96 & 3 & $0.031$ \\
$0.70$ & 128 & 0 & $0.000$ & 96 & 10 & $0.104$ \\
$0.80$ & 128 & 0 & $0.000$ & 96 & 23 & $0.240$ \\
$0.90$ & 128 & 1 & $0.008$ & 96 & 52 & $0.542$ \\
$0.95$ & 128 & 19 & $0.148$ & 96 & 79 & $0.823$ \\
\bottomrule
\end{tabular}}
\floatnote{Each target must lie strictly between the endpoint values on $[.1,10]$, using 20,000 nodes. The screen excludes 20 of 768 mean-information requests and 167 of 576 inverse-information requests. The latter omit 192 heavy-tailed cases with infinite reciprocal-information moment. Respectively 3 and 6 target--endpoint pairs lie within .001 of a boundary. Interior roots can be missed by this screen.}
\end{table}

%% file: floats/tab_s2_reachability.tex
\begin{table}[htbp]\centering
\caption{Production-run statuses after the preliminary screen}
\label{tab:s2-reachability}
{\small\begin{tabular}{llrrrrrrr}
\toprule
 & & & & & \multicolumn{2}{c}{Not reached} & & \\
\cmidrule(lr){6-7}
Arm & Functional & Target & Eligible & Solved & Unreach. & Infeas. & Other & Share \\
\midrule
EQC & $\rhotilde$ & 0.50 & 128 & 128 & 0 & 0 & 0 & $0.000$ \\
EQC & $\rhotilde$ & 0.60 & 128 & 128 & 0 & 0 & 0 & $0.000$ \\
EQC & $\rhotilde$ & 0.70 & 128 & 128 & 0 & 0 & 0 & $0.000$ \\
EQC & $\rhotilde$ & 0.80 & 128 & 128 & 0 & 0 & 0 & $0.000$ \\
EQC & $\rhotilde$ & 0.90 & 127 & 126 & 0 & 1 & 0 & $0.008$ \\
EQC & $\rhotilde$ & 0.95 & 109 & 105 & 0 & 4 & 0 & $0.037$ \\
SAC\_INFO & $\rhotilde$ & 0.50 & 128 & 128 & 0 & 0 & 0 & $0.000$ \\
SAC\_INFO & $\rhotilde$ & 0.60 & 128 & 128 & 0 & 0 & 0 & $0.000$ \\
SAC\_INFO & $\rhotilde$ & 0.70 & 128 & 126 & 0 & 0 & 2 & $0.000$ \\
SAC\_INFO & $\rhotilde$ & 0.80 & 128 & 125 & 0 & 0 & 3 & $0.000$ \\
SAC\_INFO & $\rhotilde$ & 0.90 & 127 & 119 & 0 & 0 & 8 & $0.000$ \\
SAC\_INFO & $\rhotilde$ & 0.95 & 109 & 107 & 0 & 0 & 2 & $0.000$ \\
SAC\_MSEM & $\wbar$ & 0.50 & 96 & 95 & 0 & 0 & 1 & $0.000$ \\
SAC\_MSEM & $\wbar$ & 0.60 & 96 & 90 & 4 & 0 & 2 & $0.042$ \\
SAC\_MSEM & $\wbar$ & 0.70 & 96 & 85 & 9 & 0 & 2 & $0.094$ \\
SAC\_MSEM & $\wbar$ & 0.80 & 96 & 67 & 26 & 0 & 3 & $0.271$ \\
SAC\_MSEM & $\wbar$ & 0.90 & 95 & 49 & 44 & 0 & 2 & $0.463$ \\
SAC\_MSEM & $\wbar$ & 0.95 & 82 & 21 & 59 & 0 & 2 & $0.720$ \\
\bottomrule
\end{tabular}}
\floatnote{The mean-information procedures use 748 cells. The inverse-information procedure excludes heavy-tailed cells and uses 561. The reported share combines unreachable and infeasible statuses over all eligible cells. These statuses apply to the specified search. Production draws new forms and nodes, so five EQC targets became infeasible after passing the screen in \cref{tab:s2-reachability-unpruned}.}
\end{table}

%% file: tables/tab_sc03_cstar_by_shape.tex
\begin{tabular}{llrrrr}
\hline
Item model & Difficulty pool & normal & bimodal & skewed & heavy-tailed \\
\hline
Rasch & parametric & 0.750 & 0.758 & 0.750 & 0.747 \\
Rasch & IRW & 1.480 & 0.743 & 0.731 & 0.735 \\
2PL & parametric & 0.617 & 0.674 & 0.683 & 0.698 \\
2PL & IRW & 0.640 & 1.150 & 1.592 & 0.600 \\
3PL, $g=.20$ & parametric & 0.896 & 0.988 & 0.845 & 0.999 \\
3PL, $g=.20$ & IRW & 1.121 & 0.922 & 0.921 & 1.119 \\
3PL, beta $g$ & parametric & 1.029 & 0.882 & 1.035 & 0.944 \\
3PL, beta $g$ & IRW & 1.100 & 1.379 & 1.220 & 1.400 \\
\hline
\end{tabular}

%% file: tables/tab_sc03_failures.tex
\setlength{\tabcolsep}{4pt}
\begin{tabular}{llllrrl}
\hline
Item model & Latent shape & Pool & $I$ & Target & Largest $\tilde\rho$ found & Status \\
\hline
3PL, beta $g$ & heavy-tailed & IRW & 10 & 0.90 & 0.870 & above $c_U$ \\
3PL, beta $g$ & heavy-tailed & parametric & 10 & 0.95 & 0.948 & above $c_U$ \\
3PL, beta $g$ & normal & IRW & 10 & 0.95 & 0.917 & above $c_U$ \\
3PL, $g=.20$ & heavy-tailed & IRW & 20 & 0.95 & 0.935 & above $c_U$ \\
3PL, $g=.20$ & skewed & parametric & 10 & 0.95 & 0.946 & above $c_U$ \\
\hline
\end{tabular}

%% file: tables/tab_sc04_holdout_by_M.tex
\setlength{\tabcolsep}{4pt}
\begin{tabular}{rrrrrr}
\hline
$M$ & runs & residual MAE & fresh-draw MAE & SD of signed fresh error & SD of $c^*$ \\
\hline
500 & 240 & $2.4\times10^{-9}$ & 0.0117 & 0.0147 & 0.424 \\
1,000 & 240 & $2.7\times10^{-9}$ & 0.0083 & 0.0106 & 0.419 \\
2,500 & 240 & $2.2\times10^{-9}$ & 0.0048 & 0.0060 & 0.434 \\
5,000 & 240 & $2.3\times10^{-9}$ & 0.0037 & 0.0047 & 0.403 \\
10,000 & 240 & $2.1\times10^{-9}$ & 0.0026 & 0.0034 & 0.417 \\
20,000 & 240 & $2.3\times10^{-9}$ & 0.0018 & 0.0023 & 0.419 \\
50,000 & 240 & $2.3\times10^{-9}$ & 0.0014 & 0.0017 & 0.419 \\
\hline
\end{tabular}

%% file: tables/tab_holdout_by_cell.tex
\begin{table}[!htbp]
\centering
\caption{Fresh-draw mean absolute error of the 12 holdout cells at each node count (20 seeds per cell)}
\label{tab:holdout-cell}
\footnotesize
\setlength{\tabcolsep}{2.2pt}
\begin{tabular}{lllrrrrrrr}
\toprule
Model & Shape & $I$ & $M$ = 500 & 1k & 2.5k & 5k & 10k & 20k & 50k \\
\midrule
Rasch & normal & 10 & .0089 & .0061 & .0037 & .0029 & .0023 & .0016 & .0010 \\
Rasch & normal & 40 & .0109 & .0058 & .0039 & .0030 & .0020 & .0021 & .0009 \\
Rasch & skewed & 10 & .0139 & .0103 & .0051 & .0032 & .0032 & .0017 & .0019 \\
Rasch & skewed & 40 & .0169 & .0094 & .0056 & .0045 & .0041 & .0020 & .0013 \\
2PL & normal & 10 & .0080 & .0065 & .0040 & .0035 & .0016 & .0012 & .0011 \\
2PL & normal & 40 & .0110 & .0076 & .0048 & .0032 & .0021 & .0015 & .0011 \\
2PL & skewed & 10 & .0151 & .0093 & .0078 & .0045 & .0021 & .0021 & .0017 \\
2PL & skewed & 40 & .0112 & .0084 & .0045 & .0050 & .0035 & .0018 & .0014 \\
3PL, $g=.20$ & normal & 10 & .0082 & .0054 & .0042 & .0031 & .0020 & .0011 & .0010 \\
3PL, $g=.20$ & normal & 40 & .0092 & .0089 & .0035 & .0030 & .0022 & .0017 & .0012 \\
3PL, $g=.20$ & skewed & 10 & .0111 & .0091 & .0055 & .0046 & .0030 & .0019 & .0018 \\
3PL, $g=.20$ & skewed & 40 & .0156 & .0132 & .0049 & .0040 & .0033 & .0023 & .0018 \\
\bottomrule
\end{tabular}

\smallskip
\begin{minipage}{0.95\textwidth}
\footnotesize
\emph{Note.} Target .70; each entry is the mean absolute error, over 20 seeds, of the calibrated form's mean-information index on 200,000 fresh draws (their sample variance) against the target. The pooled means over the 12 cells are the fresh-draw line of \cref{fig:holdout}. Within a node count the cells differ because information is averaged over different populations and forms; the skewed population is the harder case at every node count, whereas test length has no consistent ordering across node counts.
\end{minipage}
\end{table}

%% file: floats/tab_s2_replication.tex
\begin{table}[htbp]\centering
\caption{Variation in information when persons are redrawn after calibration}
\label{tab:s2-replication}
{\small\begin{tabular}{rrrrrr}
\toprule
 & & \multicolumn{2}{c}{Realized $\rhotilde$} & \multicolumn{2}{c}{SD relative to $N = 100$} \\
\cmidrule(lr){3-4}\cmidrule(lr){5-6}
Persons & Cells & Median SD & Median bias & Observed & $\sqrt{N}$ law \\
\midrule
100 & 743 & $0.0256$ & $-0.0026$ & $1.000$ & $1.000$ \\
250 & 743 & $0.0163$ & $-9.86\times 10^{-4}$ & $0.634$ & $0.632$ \\
500 & 743 & $0.0113$ & $-4.32\times 10^{-4}$ & $0.442$ & $0.447$ \\
1000 & 743 & $0.0081$ & $-1.39\times 10^{-4}$ & $0.315$ & $0.316$ \\
2000 & 743 & $0.0057$ & $-8.41\times 10^{-5}$ & $0.222$ & $0.224$ \\
\bottomrule
\end{tabular}}
\floatnote{Each of 743 fixed forms has 1,000 independent redraws at each sample size. Rows give medians across forms of within-form SD and bias. Bias is relative to the achieved calibration index. The final columns compare SD ratios with $\sqrt{100/N}$. These summaries exclude variation from drawing or recalibrating forms.}
\end{table}

%% file: floats/tab_c_inflation.tex
\begin{table}[htbp]\centering
\caption{Calibration multipliers relative to 2PL in matched design cells}\label{tab:c-inflation}
{\footnotesize\begin{tabular}{lrrrrr}
\toprule
Model & Matched cells & Median ratio & Q1 & Q3 & Share above 2PL \\
\midrule
3PL, $g \sim$ Beta & 92 & $1.41$ & $1.30$ & $1.56$ & $0.967$ \\
3PL, $g = .20$ & 92 & $1.31$ & $1.25$ & $1.52$ & $0.978$ \\
Rasch & 96 & $1.07$ & $1.02$ & $1.12$ & $0.781$ \\
\bottomrule
\end{tabular}}
\floatnote{Cells match on target, latent shape and length: 92 pairs for each 3PL variant and 96 for Rasch. The forms differ across models, so the ratios also reflect item variation. Quartiles describe variation across cells, not confidence intervals.}
\end{table}

%% file: floats/tab_jensen_gap.tex
\begin{table}[htbp]\centering
\caption{Differences between inverse-information and mean-information multipliers}
\label{tab:jensen-gap}
{\small\begin{tabular}{lrrrr}
\toprule
Sandwich width & Cells & Median width & Median gap in $c$ & Share with positive gap \\
\midrule
Q1 narrowest & 102 & $7.75\times 10^{-4}$ & $-4.66\times 10^{-4}$ & $0.480$ \\
Q2 & 102 & $0.0027$ & $0.0042$ & $0.569$ \\
Q3 & 101 & $0.0065$ & $0.0245$ & $0.673$ \\
Q4 widest & 102 & $0.0188$ & $0.1066$ & $0.755$ \\
all & 407 & $0.0040$ & $0.0114$ & $0.619$ \\
\bottomrule
\end{tabular}}
\floatnote{The 407 cells with both solutions are grouped by the gap $\rhotilde-\wbar$ at a common multiplier. The reported multiplier difference is inverse-information minus mean-information. It combines stochastic error and branch differences, so its sign need not follow the index ordering. Separately, the index ordering held in all 743 evaluated cells, with smallest gap $2.71\times10^{-5}$.}
\end{table}

%% file: appendices/D_example.tex
\section{Worked Example and Software}
\label{app:example}

We calibrate a 20-item form to .85, check its information on new ability draws, and then show how the package handles a target outside the detected range. The code illustrates what to inspect before using a calibrated form in a simulation.

\subsection{Calibrate a form}
\label{app:sd01}

We ran the example with \texttt{IRTsimrel} 0.3.1 \parencite{Lee2026IRTsimrel031} under R 4.6.0 \parencite{RCoreTeam2026R460}. The \href{https://github.com/joonho112/IRTsimrel-paper-replication}{replication repository} contains the frozen archive \path{data/software/IRTsimrel_0.3.1.tar.gz}; \path{manifest/files.csv} records its SHA-256 checksum (prefix \texttt{96025549a8fcb6c9}). From the repository root, run \texttt{Rscript code/install\_dependencies.R} to install the archived package versions. Then run the code below in a fresh R session from that directory. The retained form in \path{data/precomputed/software/normal_item_form.csv} supplies \cref{fig:example}.

\begin{small}\begin{verbatim}
library(IRTsimrel, lib.loc = ".library/0.3.1")
stopifnot(as.character(packageVersion("IRTsimrel")) == "0.3.1")
RNGkind("Mersenne-Twister", "Inversion", "Rejection")
fit <- eqc_calibrate(
  target_rho = .85, n_items = 20L, model = "2pl",
  latent_shape = "normal", item_source = "parametric",
  item_params = list(difficulty_params = list(mu = 0, sigma = 1)),
  reliability_metric = "info", M = 20000L, c_bounds = c(.1, 10),
  tol = 1e-6, root_policy = "lowest_increasing", seed = 20260826)
c(multiplier = fit$c_star, numerical_target = fit$achieved_rho)
fit$misc[c("calibration_status", "root_status", "selected_root_direction")]
quantile(fit$lambda_scaled, c(.1, .5, .9))
\end{verbatim}\end{small}

The solver returns $c^*=1.27785739$ and achieved index $.850000002$, a residual of $2.0\times10^{-9}$. Its status fields are \texttt{ok}, \texttt{uniroot\_success} and \texttt{increasing}; it detected one crossing in 88 objective evaluations. Calibration-node variance was $1.00556$. The 10th, 50th and 90th percentiles of calibrated discriminations were $.9597$, $1.2509$ and $1.9158$.

\Needspace{20\baselineskip}
\subsection{Check the returned form on fresh draws}
\label{app:sd02}

Next we evaluate all three indices on 200,000 fresh normal draws, standardized to sample variance one. The calculation uses variance one and computes 2PL information in the log domain:

\begin{small}\begin{verbatim}
set.seed(20260827)
theta <- rnorm(200000L)
theta <- (theta - mean(theta)) / sd(theta)
info <- numeric(length(theta))
for (i in seq_along(fit$beta_vec)) {
  eta  <- fit$lambda_scaled[i] * (theta - fit$beta_vec[i])
  info <- info + fit$lambda_scaled[i]^2 *
    exp(plogis(eta, log.p = TRUE) + plogis(-eta, log.p = TRUE))
}
c(mean_information = mean(info) / (1 + mean(info)),
  pointwise        = mean(info / (1 + info)),
  harmonic         = 1 / (1 + mean(1 / info)))
\end{verbatim}\end{small}

The fresh draws give mean information $5.6367$ and indices $.849322$, $.840007$ and $.835443$ for mean, pointwise and inverse information. The package function \texttt{compute\_rho\_tilde()} with \texttt{theta\_var = 1} gives the same mean-information index.

A 200-node Gauss--Hermite calculation gives $.849326$, $.840011$ and $.835406$, within about $4\times10^{-5}$ of the fresh-draw values. With variance $1.00556$, the same quadrature gives mean-information index $.850035$; with variance one, the root for target .85 is $1.2825$. This difference illustrates why the variance rule must accompany a reported index. No responses or scores were generated, so these are not EAP or WLE reliability estimates.

\subsection{Read an infeasible return}
\label{app:sd03}

We next request .95 from two Rasch items with bounds $[.1,.2]$ and seed 20260929. The scan finds no crossing and returns $c=.2$, achieved index $.0209$, and statuses \texttt{infeasible\_above\_range} and \texttt{infeasible}. The warning below wraps one output line for display. The returned form has not met the target; the detected maximum applies only within the supplied bounds.

\noindent\begin{minipage}{\linewidth}
\begin{small}\begin{verbatim}
Target rho* = 0.950 exceeds the maximum achievable reliability
  for this configuration.
  - Items: 2
  - Best detected c in bounds: 0.2000
  - Best detected rho: 0.0209
  - At c = 0.20 (upper bound): rho = 0.0209
  - Gap: 0.9291

Suggestions:
  1. Revise the explicit c_bounds if scientifically justified
  2. Increase n_items for higher achievable reliability
  3. Accept achieved rho = 0.0209 as the maximum detected for this design

Returning c_star = 0.2000.
\end{verbatim}\end{small}
\end{minipage}

We previously reran both examples from a relocated directory with a fresh installation of the frozen package. The calibration summary, form, root scan and infeasible-case files agreed byte for byte, as did a repeat against the same isolated library. These checks cover the examples only.

\subsection{Correspondence between the paper and the package}
\label{app:sd04}

\Cref{tab:function-mapping} maps the notation to package functions. The package generates items, calibrates forms and evaluates information indices. The study code supplies the illustrations' response generation, fitting and testing. Reproducing those studies therefore requires both. For each calibrated form, retain the returned parameters, multiplier, achieved index and diagnostics, including the search bounds and crossing policy.

\input{tables/tab_function_mapping}

\subsection{Scope of the accompanying materials}
\label{app:sd05}

The GitHub replication repository provides study code, inputs, archived package versions and retained results. Its guide distinguishes executable studies from analyses supplied as retained evidence.

From the repository root, \texttt{Rscript run\_all.R example} runs the worked example, \texttt{Rscript run\_all.R quick} rebuilds displays from retained results, and \texttt{Rscript run\_all.R verify} checks file hashes and selected numerical results. These commands do not rerun the full simulation studies. The file \path{verification/release-checks.md} records the completed checks. The historical compute budget totals about 409 worker-hours across 20 stages, with more than two thirds spent on the item-level treatment-effect and DIF studies.

%% file: tables/tab_function_mapping.tex
\begin{table}[ht]
\centering
\caption{Objects of the paper and the corresponding functions of \texttt{IRTsimrel} 0.3.1}
\label{tab:function-mapping}
\footnotesize
\setlength{\tabcolsep}{5pt}
\begin{tabular}{@{}>{\raggedright\arraybackslash}p{3.9cm}>{\raggedright\arraybackslash}p{4.6cm}>{\raggedright\arraybackslash}p{4.7cm}@{}}
\toprule
Object in the paper & Function & What it does \\
\midrule
Latent distribution $G$ (\cref{subsec:model}) & \texttt{sim\_latentG()} & draws abilities from the built-in shapes, standardized to mean zero and variance one before any location-scale change \\
Item form $\{(\beta_i,\lambda_{i,0})\}$ (\cref{subsec:model}) & \texttt{sim\_item\_params()} & generates difficulties (parametric or IRW pool) and baseline discriminations with the copula dependence \\
Fixed-node calibration (\cref{subsec:eqc}) & \texttt{eqc\_calibrate()} & scans the declared interval, selects a crossing by the branch policy, refines it, and returns $c^*$, the achieved value, the residual, the calibrated form and the status \\
Stochastic calibration (\cref{subsec:sac}) & \texttt{sac\_calibrate()} & runs the projected recursion with fresh batches, optional form redraws, tail averaging and finite-run diagnostics \\
The three indices (\cref{subsec:target}) & \texttt{compute\_rho\_tilde()}\newline \texttt{compute\_rho\_bar()}\newline (pointwise: no export) & evaluate the mean-information and the inverse-information (\texttt{"msem"}) index of a form on supplied nodes under a stated variance; the pointwise index has no package function and is computed by the code in \cref{app:sd02} \\
Fresh-draw evaluation (\cref{subsec:fresh-eval}) & the index functions on new draws from \texttt{sim\_latentG()} & the independent evaluation of a calibrated form \\
Response generation (step five of \cref{subsec:workflow}) & \texttt{simulate\_response\_data()} & generates response matrices from a calibrated form for a chosen number of persons \\
Fitted-score reliability (\cref{app:targets}) & \texttt{compute\_reliability\_tam()} & fits a model with TAM and returns its EAP and WLE reliability coefficients \\
\bottomrule
\end{tabular}

\smallskip
\begin{minipage}{0.95\textwidth}
\footnotesize
\emph{Note.} In \texttt{reliability\_metric}, \texttt{"info"} selects mean information and \texttt{"msem"} selects inverse information. Release 0.3.1 permits inverse-information calibration only through SAC. That API restriction reflects its potentially nonmonotone objective; it does not make inversion impossible. All function names are exports of the frozen release.
\end{minipage}
\end{table}

%% file: appendices/E_illustration.tex
\section{Details of the Item-Level Treatment-Effect Simulation}
\label{app:illustration}
\label{app:replication}

We vary the information target while holding sample size, test length and the ability-scale effect distribution fixed within each effect condition. We then change the effect convention and, in a second extension, vary test length as well. These comparisons ask which performance patterns persist across specified measurement conditions.

\subsection{Design and calibration}
\label{app:se01}

The item-level treatment-effect simulation in \cref{subsec:illustration} uses
\begin{equation}
\label{eq:app-ilhte}
\logit\Pr(Y_{pi}=1\mid\theta_p,b_i,\tau_i,T_p)=d\,(\theta_p-b_i)+\tau_i\,T_p ,
\end{equation}
with 500 persons, 20 items and exactly 250 treated persons ($T_p=1$). Abilities and difficulties are independent standard-normal draws, and item effects are independent of difficulties. The benchmark of \textcite{gilbert_modeling_2023} fixes $d=1.7$ and draws $\tau_i\sim N(\mu_0,\sigma_0^2)$, with $\mu_0\in\{0,.4\}$ and $\sigma_0\in\{0,.2,.4\}$.

We generate new difficulties in each replication and calibrate the untreated model to one of six targets: .50, .60, .70, .80, .887 and .90. The study code solves for $d$ on $[.05,20]$, using tolerance $10^{-8}$, 20,000 shared standard-normal reference draws and variance one. The recorded index met its target within $10^{-11}$ on those draws. This is a direct study-code solve; it does not use the package solver or include a fresh-draw check for each replication.

In the primary analysis, we set $\tau_i=(d/1.7)\tau_{i,0}$ with $\tau_{i,0}\sim N(\mu_0,\sigma_0^2)$. The ability-scale effect $\tau_i/d=\tau_{i,0}/1.7$ then has the same distribution across targets within each condition. A reference mean of .4 corresponds to $.4/1.7\approx.235$ ability standard deviations. \Cref{tab:app-slope} shows how the calibrated slope and logit coefficient change together.

\begin{table}[ht]
\centering
\caption{Calibrated common slope $d$ by target in the ability-scale convention (reference effect mean .4, item-effect SD 0; 1,000 replications per target)}
\label{tab:app-slope}
\small
\begin{tabular}{rrrrr}
\toprule
Target & mean $d$ & min $d$ & max $d$ & implied logit coefficient $.4\,d/1.7$ \\
\midrule
.50 & 0.469 & 0.461 & 0.483 & 0.110 \\
.60 & 0.588 & 0.574 & 0.619 & 0.138 \\
.70 & 0.761 & 0.732 & 0.820 & 0.179 \\
.80 & 1.068 & 1.003 & 1.178 & 0.251 \\
.887 & 1.717 & 1.563 & 2.079 & 0.404 \\
.90 & 1.909 & 1.710 & 2.426 & 0.449 \\
\bottomrule
\end{tabular}
\end{table}

\subsection{Fitted models and the treatment-effect null}
\label{app:se01-models}

We use \texttt{lme4::glmer} with \texttt{bobyqa} to fit a constant-effect model, \texttt{correct \textasciitilde{} treat + (1 | person) + (1 | item)}, and a varying-effect model, \texttt{correct \textasciitilde{} treat + (1 | person) + (treat | item)}. Both use a two-sided Wald test of the treatment coefficient at .05.

The constant-effect model estimates the effect for the administered items. The varying-effect model includes item-effect variation and the uncertainty needed for an item-population mean \parencite{gilbert_modeling_2023}. When $\sigma_0>0$, a sampled form usually has a nonzero average effect even when $\mu_0=0$. Thus, the null condition permits positive and negative item effects and concerns their population mean. \Cref{app:se04} documents the benchmark comparison and the difference from its one-sided test.

\subsection{Results with effects fixed on the ability scale}
\label{app:se02}

The six targets and six effect distributions give 36 conditions. We ran 1,000 replications per condition for each model, with no recorded convergence flag (\cref{tab:arm-a-all}). Targets use independent replications; the six effect conditions within a target share random-number streams. Their comparisons are therefore paired, and marginal intervals do not give intervals for their differences. \Cref{fig:arm-a} reports exact two-sided 95\% binomial intervals for the two conditions highlighted in the main text.

The standard errors help explain the rejection rates. In the null condition with item-effect SD .4, the ratio of mean fitted standard error to the empirical standard deviation of estimates ranged from .85 to .96 for the constant-effect model, compared with .96 to 1.06 for the varying-effect model. Ratios below one indicate that model standard errors understate variation across replications. The constant-effect model continued to omit that item-sampling variation as information increased. \Cref{tab:s6-arms}, panel A, reports endpoint comparisons for all six effect conditions.

\input{tables/tab_arm_a_all}

\subsection{Changing the scale on which effects are fixed}
\label{app:se03}

We repeated the simulation with $\tau_i=\tau_{i,0}$ fixed on the logit scale. As the calibrated slope grows, the corresponding ability-scale effect $\tau_i/d$ becomes smaller. In the main text's power condition, rejection fell from 1.00 at target .50 to .56 at .90, compared with an increase from .466 to .723 when the ability-scale effect was fixed (\cref{tab:app-arm-logit}). \Cref{tab:s6-te-scale} compares the conventions when item-effect SD is .4.

These results answer different questions because different effects are held fixed. A comparison across information targets therefore needs an explicit effect convention. The sensitivity analysis also uses only 100 replications per condition, so its rates are less precise than those from the primary analysis.

\begin{table}[ht]
\centering
\caption{Constant-effect model under the fixed-logit convention: rejection rate by target and effect condition (100 replications per condition)}
\label{tab:app-arm-logit}
\footnotesize
\input{tables/tab_se03_arm_a_logit}
\end{table}

\input{floats/tab_s6_te_scale}

\subsection{Benchmark comparison}
\label{app:se04}

Before the extensions, we reconstructed the generating model of \textcite{gilbert_modeling_2023} and compared it with archived replication-package summaries. The design used discrimination 1.7, 500 persons, 20 items, standard-normal abilities and difficulties, and normal item effects with means 0 or .4 and SDs 0, .2 or .4. We ran 2,000 replications per condition and summarized both fitted models, giving twelve strata.

All 72 comparisons met the specified criteria (\cref{tab:s6-fidelity}). These covered mean estimates, empirical standard deviations, mean model standard errors, their ratios, rejection rates and heterogeneity-detection rates. Estimates and rates used uncertainty bands; the 36 standard-deviation, standard-error and ratio checks used a 5\% relative tolerance. The decision rules differ, however: the source uses a one-sided positive-effect test and our analyses use two-sided .05 tests.

\input{floats/tab_s6_fidelity}

\subsection{Varying test length and information target together}
\label{app:se05}

The second extension crosses lengths 10, 20 and 40, targets .80, .887 and .94, and reference item-effect SDs .1, .2 and .4. The item-population effect mean is zero. We planned 200 replications in each of the 27 cells, generating and calibrating a new form each time with the same study-code solve and ability-scale effect convention. Target .887 corresponds to the benchmark's 20-item design.

Here the outcome is detection of variation in treatment effects across items. We compare the varying- and constant-effect models by a likelihood-ratio test. Because the item-slope variance and covariance are on a boundary under the null, we use $\tfrac12\chi^2_1+\tfrac12\chi^2_2$ at .05 as an asymptotic working reference. Its accuracy is uncertain with few items. The benchmark comparison retains its original $\chi^2_2$ reference.

\paragraph{Handling flagged fits.}
The original fitting flag combined singular-boundary messages with convergence failures, flagging 461 varying-model fits in ten-item conditions. We reconstructed these data sets from their seeds and tried the original optimizer, a controlled warm start, two alternative optimizers and a bounded refinement. Of 445 fits passing numerical checks, 432 also matched the original estimate, standard error and likelihood-ratio statistic within .001. We retain those 432, leaving 29 unresolved: 5,371 observed outcomes among 5,400 planned. A sensitivity analysis counts unresolved outcomes as non-detections.

\paragraph{Comparing changes in length and target.}
We summarize the four combinations of the lowest and highest length and target. Write $q_{LL}$, $q_{HL}$, $q_{LH}$ and $q_{HH}$ for detection rates at $(I,\rhotilde^*)=(10,.80)$, $(40,.80)$, $(10,.94)$ and $(40,.94)$. We average the length contrast at the two targets and the target contrast at the two lengths:
\begin{equation}
\label{eq:app-paths}
D_I=\tfrac12\{(q_{HL}-q_{LL})+(q_{HH}-q_{LH})\},\qquad
D_r=\tfrac12\{(q_{LH}-q_{LL})+(q_{HH}-q_{HL})\},
\end{equation}
so that $D_I+D_r=q_{HH}-q_{LL}$. This identity describes the total change from the shortest, lowest-target design to the longest, highest-target design.

For item-effect SDs .1, .2 and .4, the total changes were .074, .454 and .271. The averaged target contrasts were .092, .468 and .286; the averaged length contrasts were $-.018$, $-.014$ and $-.014$, with approximate 95\% intervals all containing zero (\cref{tab:s6-arms}, panel B). Counting unresolved fits as non-detections gives total changes .075, .455 and .275. The contrasts share endpoints and are not independent. \Cref{fig:replication} shows the full grid.

\input{floats/tab_s6_arms}

\begin{figure}[ht]
\centering
\includegraphics[width=\textwidth]{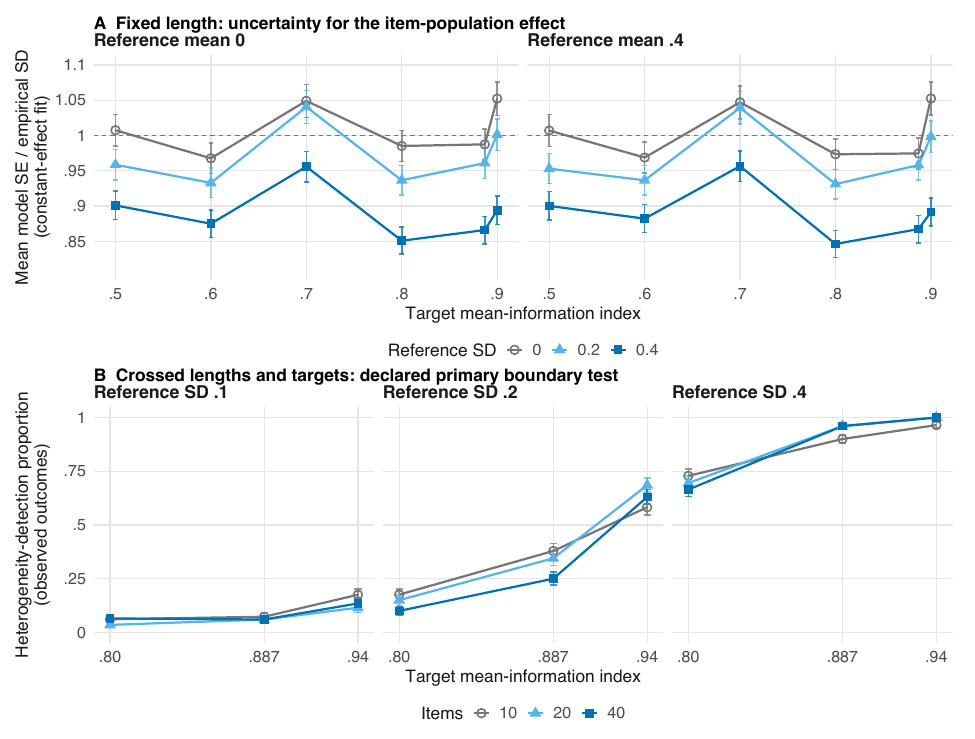}
\caption{Standard-error ratios of the fixed-length illustration and heterogeneity detection in the crossed extension}
\label{fig:replication}

\smallskip
\begin{minipage}{0.95\textwidth}
\footnotesize
\emph{Note.} Both designs use 500 persons, equal treatment allocation, standard-normal abilities and reference population, and $\tau_i=(d/1.7)\tau_{i,0}$. Each replication generates and calibrates a new form. A: mean fitted SE divided by empirical estimator SD for the constant-effect model, with 1,000 replications per condition. The dashed line marks one; bars are $\pm1$ approximate Monte Carlo SE. B: heterogeneity-detection rates among resolved fits (192--200 per cell); bars are $\pm1$ plug-in binomial SE, zero when all fits detect heterogeneity. Targets use independent replications; some effect conditions share streams.
\end{minipage}
\end{figure}

The crossed design lets us compare changes in target at each length and changes in length at each target. Both changes regenerate and recalibrate the form, so these are comparisons of complete designs. The decomposition does not isolate causal effects of reliability and item count.

%% file: tables/tab_arm_a_all.tex
\begingroup\footnotesize\setstretch{1.0}\setlength{\tabcolsep}{3.5pt}
\begin{longtable}{llrrrr}
\caption{All 36 conditions of the fixed-length illustration under the ability-scale convention: rejection rates of the constant- and varying-effect fits (1,000 replications each)}\label{tab:arm-a-all}\\
\toprule
Reference mean & item-effect SD & target & constant effect & varying effect & 95\% interval (constant) \\
\midrule
\endfirsthead
\toprule
Reference mean & item-effect SD & target & constant effect & varying effect & 95\% interval (constant) \\
\midrule
\endhead
\midrule\multicolumn{6}{r}{\footnotesize continued on the next page}\\
\endfoot
\bottomrule
\endlastfoot
0 & 0 & .5 & .060 & .047 &  \\
0 & 0 & .6 & .060 & .052 &  \\
0 & 0 & .7 & .034 & .030 &  \\
0 & 0 & .8 & .059 & .055 &  \\
0 & 0 & .887 & .056 & .058 &  \\
0 & 0 & .9 & .041 & .039 &  \\
0 & .2 & .5 & .063 & .055 &  \\
0 & .2 & .6 & .067 & .058 &  \\
0 & .2 & .7 & .045 & .036 &  \\
0 & .2 & .8 & .066 & .055 &  \\
0 & .2 & .887 & .056 & .046 &  \\
0 & .2 & .9 & .050 & .043 &  \\
0 & .4 & .5 & .076 & .059 & [.060, .094] \\
0 & .4 & .6 & .089 & .061 & [.072, .108] \\
0 & .4 & .7 & .063 & .044 & [.049, .080] \\
0 & .4 & .8 & .090 & .059 & [.073, .109] \\
0 & .4 & .887 & .091 & .055 & [.074, .111] \\
0 & .4 & .9 & .075 & .041 & [.059, .093] \\
.4 & 0 & .5 & .466 & .441 & [.435, .497] \\
.4 & 0 & .6 & .545 & .517 & [.514, .576] \\
.4 & 0 & .7 & .607 & .594 & [.576, .637] \\
.4 & 0 & .8 & .684 & .664 & [.654, .713] \\
.4 & 0 & .887 & .718 & .706 & [.689, .746] \\
.4 & 0 & .9 & .723 & .717 & [.694, .751] \\
.4 & .2 & .5 & .446 & .418 &  \\
.4 & .2 & .6 & .550 & .510 &  \\
.4 & .2 & .7 & .609 & .576 &  \\
.4 & .2 & .8 & .663 & .635 &  \\
.4 & .2 & .887 & .693 & .670 &  \\
.4 & .2 & .9 & .703 & .680 &  \\
.4 & .4 & .5 & .442 & .384 &  \\
.4 & .4 & .6 & .544 & .469 &  \\
.4 & .4 & .7 & .593 & .510 &  \\
.4 & .4 & .8 & .652 & .575 &  \\
.4 & .4 & .887 & .674 & .608 &  \\
.4 & .4 & .9 & .672 & .586 &  \\
\end{longtable}
\endgroup

%% file: tables/tab_se03_arm_a_logit.tex
\setlength{\tabcolsep}{3.5pt}
\begin{tabular}{lrrrrrr}
\hline
& \multicolumn{3}{c}{reference effect mean 0} & \multicolumn{3}{c}{reference effect mean .4} \\
Target & SD 0 & SD .2 & SD .4 & SD 0 & SD .2 & SD .4 \\
\hline
.5 & .060 & .200 & .290 & 1.000 & 1.000 & .990 \\
.6 & .100 & .100 & .280 & 1.000 & 1.000 & .990 \\
.7 & .020 & .060 & .150 & 1.000 & 1.000 & 1.000 \\
.8 & .090 & .090 & .130 & .930 & .900 & .860 \\
.887 & .070 & .080 & .130 & .680 & .630 & .630 \\
.9 & .050 & .080 & .090 & .560 & .570 & .550 \\
\hline
\end{tabular}

%% file: floats/tab_s6_te_scale.tex
\begin{table}[htbp]\centering
\caption{Effect-scale sensitivity in the item-level treatment-effect simulation}
\label{tab:s6-te-scale}
{\small\begin{tabular}{lrrrr}
\toprule
Convention & Effect & $\rhotilde$ & SE ratio & Rejection rate \\
\midrule
Log-odds scale & $0.0$ & $0.500$ & $0.530$ & $0.290$ \\
Log-odds scale & $0.0$ & $0.600$ & $0.576$ & $0.280$ \\
Log-odds scale & $0.0$ & $0.700$ & $0.758$ & $0.150$ \\
Log-odds scale & $0.0$ & $0.800$ & $0.778$ & $0.130$ \\
Log-odds scale & $0.0$ & $0.887$ & $0.839$ & $0.130$ \\
Log-odds scale & $0.0$ & $0.900$ & $0.917$ & $0.090$ \\
Log-odds scale & $0.4$ & $0.500$ & $0.542$ & $0.990$ \\
Log-odds scale & $0.4$ & $0.600$ & $0.587$ & $0.990$ \\
Log-odds scale & $0.4$ & $0.700$ & $0.752$ & $1.000$ \\
Log-odds scale & $0.4$ & $0.800$ & $0.758$ & $0.860$ \\
Log-odds scale & $0.4$ & $0.887$ & $0.848$ & $0.630$ \\
Log-odds scale & $0.4$ & $0.900$ & $0.923$ & $0.550$ \\
Latent scale & $0.0$ & $0.500$ & $0.901$ & $0.076$ \\
Latent scale & $0.0$ & $0.600$ & $0.875$ & $0.089$ \\
Latent scale & $0.0$ & $0.700$ & $0.956$ & $0.063$ \\
Latent scale & $0.0$ & $0.800$ & $0.851$ & $0.090$ \\
Latent scale & $0.0$ & $0.887$ & $0.866$ & $0.091$ \\
Latent scale & $0.0$ & $0.900$ & $0.894$ & $0.075$ \\
Latent scale & $0.4$ & $0.500$ & $0.900$ & $0.442$ \\
Latent scale & $0.4$ & $0.600$ & $0.882$ & $0.544$ \\
Latent scale & $0.4$ & $0.700$ & $0.957$ & $0.593$ \\
Latent scale & $0.4$ & $0.800$ & $0.846$ & $0.652$ \\
Latent scale & $0.4$ & $0.887$ & $0.867$ & $0.674$ \\
Latent scale & $0.4$ & $0.900$ & $0.892$ & $0.672$ \\
\bottomrule
\end{tabular}}
\floatnote{Reference item-effect SD is .4. Ability-fixed rows multiply reference coefficients by $d/1.7$; logit-fixed rows retain them. Effect denotes the reference mean, 0 or .4. All rows use the constant-effect model and a new 20-item form per replication: 1,000 replications for ability-fixed and 100 for logit-fixed. SE ratio is mean fitted SE divided by empirical estimator SD; rejection uses a two-sided .05 test.}
\end{table}

%% file: floats/tab_s6_fidelity.tex
\begin{table}[htbp]\centering
\caption{Agreement with archived benchmark summaries}
\label{tab:s6-fidelity}
{\scriptsize\setlength{\tabcolsep}{3pt}\begin{tabular}{llrrrr}
\toprule
Quantity & Criterion & Cells & Within & Largest $|\Delta|$ & Median $|\Delta|$ \\
\midrule
Mean estimate & simultaneous difference band & 12 & 12 & $0.0066$ & $0.0015$ \\
Empirical SD & 5 percent equivalence margin & 12 & 12 & $0.0039$ & $0.0018$ \\
Mean standard error & 5 percent equivalence margin & 12 & 12 & $0.0011$ & $5.30\times 10^{-4}$ \\
Ratio of the two & 5 percent equivalence margin & 12 & 12 & $0.0199$ & $0.0083$ \\
Rejection rate & simultaneous difference band & 12 & 12 & $0.0205$ & $0.0085$ \\
Detection rate & simultaneous difference band & 12 & 12 & $0.0070$ & $0.0045$ \\
\bottomrule
\end{tabular}}
\floatnote{Six conditions, two fitted models and six quantities give 72 comparisons, all meeting the stated criteria. Each condition has 2,000 replications. Mean estimates and rates use simultaneous difference bands; SD, SE and their ratio use 5\% relative margins. The archived positive-effect test is one-sided; the calibrated extensions use two-sided .05 tests.}
\end{table}

%% file: floats/tab_s6_arms.tex
\begin{table}[htbp]\centering
\caption{Endpoint comparisons in the two item-level treatment-effect simulations}
\label{tab:s6-arms}
{\small\begin{tabular}{llrrrr}
\toprule
\multicolumn{6}{l}{\emph{Panel A. Fixed length: targets .50 and .90}} \\
 & & \multicolumn{2}{c}{SE ratio} & \multicolumn{2}{c}{Rejection rate} \\
\cmidrule(lr){3-4}\cmidrule(lr){5-6}
Mean & SD & low $\rhotilde$ & high $\rhotilde$ & low $\rhotilde$ & high $\rhotilde$ \\
\midrule
$0.0$ & $0.0$ & $1.008$ & $1.052$ & $0.060$ & $0.041$ \\
$0.0$ & $0.2$ & $0.959$ & $1.001$ & $0.063$ & $0.050$ \\
$0.0$ & $0.4$ & $0.901$ & $0.894$ & $0.076$ & $0.075$ \\
$0.4$ & $0.0$ & $1.007$ & $1.053$ & $0.466$ & $0.723$ \\
$0.4$ & $0.2$ & $0.953$ & $0.998$ & $0.446$ & $0.703$ \\
$0.4$ & $0.4$ & $0.900$ & $0.892$ & $0.442$ & $0.672$ \\
\bottomrule
\end{tabular}}
\vspace{0.6em}
{\scriptsize\begin{adjustbox}{max width=\textwidth}\begin{tabular}{llrrrrl}
\toprule
\multicolumn{7}{l}{\emph{Panel B. Varying length and target}} \\
SD & Analysis & Diagonal & Length & Target & Contrast SE & Endpoint $n$ \\
\midrule
$0.1$ & Resolved fits & $0.074$ & $-0.018$ & $0.092$ & $0.022$ & 196/200/194/200 \\
$0.1$ & Unresolved as non-detections & $0.075$ & $-0.015$ & $0.090$ & $0.022$ & 200/200/200/200 \\
$0.2$ & Resolved fits & $0.454$ & $-0.014$ & $0.468$ & $0.030$ & 199/200/194/200 \\
$0.2$ & Unresolved as non-detections & $0.455$ & $-0.005$ & $0.460$ & $0.030$ & 200/200/200/200 \\
$0.4$ & Resolved fits & $0.271$ & $-0.014$ & $0.286$ & $0.024$ & 199/200/199/200 \\
$0.4$ & Unresolved as non-detections & $0.275$ & $-0.010$ & $0.285$ & $0.024$ & 200/200/200/200 \\
\bottomrule
\end{tabular}\end{adjustbox}}
\floatnote{Panel A uses 20 items and targets .50 and .90, with 1,000 replications per condition. Divide the listed reference effect means and SDs by 1.7 for the ability scale. SE ratio is mean fitted SE divided by empirical estimator SD; rejection uses the two-sided .05 test. Panel B applies \cref{eq:app-paths}. Endpoint counts follow $(10,.80)$, $(40,.80)$, $(10,.94)$, $(40,.94)$. Heterogeneity tests use the equal $\chi^2_1/\chi^2_2$ mixture at .05. The full grid has 5,371 resolved of 5,400 planned fits. The two averaged contrasts share endpoints and have the same independent-endpoint SE; sensitivity rows count unresolved fits as non-detections.}
\end{table}

%% file: appendices/F_populations.tex
\section{Latent Distributions and the Empirical Difficulty Pool}
\label{app:shapes}
\label{app:pool}

The simulation inputs serve two purposes: the four ability distributions test calibration under different shapes, and the Item Response Warehouse supplies an empirical spread of item difficulties. We describe each source and the assumptions needed to combine those difficulties with the other item parameters.

\subsection{The four generating distributions}
\label{app:sf01}

We use four distributions, each with population mean zero and variance one:
\begin{itemize}
\item Standard normal.
\item Positive skew: $(Z-4)/2$ with $Z\sim\operatorname{Gamma}(4,1)$.
\item Bimodal: equal mixture of $N(-.8,.6^2)$ and $N(.8,.6^2)$.
\item Heavy tails: $t_5/\sqrt{5/3}$.
\end{itemize}
\Cref{tab:latent-shapes} gives moments from 200,000 draws per generator. The $t_5$ sample had kurtosis 7.49, compared with population kurtosis 9. Its expected reciprocal information is infinite by \cref{prop:tail}; that conclusion follows from the population tails, not the sample kurtosis.

Population standardization does not make a finite sample's variance exactly one. The theory illustrations also use their own quadrature rules: the grids in \cref{fig:functionals} are density-weighted, truncated and restandardized. Reproduction therefore requires the numerical rule as well as the distribution.

\input{floats/tab_latent_shapes}

These distributions cover symmetry, positive skew, separated modes and polynomial tails. We did not include negative skew or fitted empirical populations. In the DIF study, a skewed calibration reference is paired with normal response-generating abilities (\cref{app:dif}); the reference and response populations must therefore be distinguished.

\subsection{Comparison with fitted Item Response Warehouse distributions}
\label{app:shapes-realism}

The Item Response Warehouse contains diverse response data \parencite{domingue_introduction_2025}. A companion analysis fitted normal and flexible latent distributions to 504 units from 273 studies \parencite{lee_irw_nonnormality_2026}. For 282 of 504 units, the maximum cumulative-distribution difference from the standardized normal, evaluated at density-grid boundaries, exceeded .10. This describes differences between fitted distributions; it is not a normality-test rejection rate or a direct observation of latent ability.

Median fitted excess kurtosis was 3.58 (\cref{tab:shape-realism}). The heavy-tailed generator's sample excess kurtosis, 4.49, was exceeded by 212 of 504 fitted distributions (42.1\%). This comparison helps place the chosen tail weight in context. Kurtosis alone cannot establish how often populations are multimodal.

With item estimates held fixed, replacing a normal density with a flexible one usually changed the pointwise information index little. Refitting the full model produced larger changes in some data sets. The fixed-item result applies to the pointwise index under that calculation; it does not establish how the mean-information index or fitted-score reliability behaves when items and scoring also change.

\input{floats/tab_shape_realism}

A second Item Response Warehouse study assessed marginal reliability in 889 units \parencite{lee_irw_reliability_2026}. The frequency of low reliability depended on the coefficient. EAP and WLE results had different available-case denominators and, for polytomous items, different fitting procedures. These findings motivate stating the coefficient and population explicitly; we do not use the study's deconvolution estimates to validate our information indices.

\subsection{The empirical difficulty pool}
\label{app:sf03}

The frozen \texttt{irw::diff\_long} snapshot from package version 1.0.0 contains 6,143 difficulty estimates from 145 data sets, with data-set, item, difficulty and standard-error fields. Its SHA-256 is recorded in the archive. \textcite{zhang_realistic_2025} describes the empirical-difficulty approach. The snapshot contains no joint distribution of difficulties, discriminations, guessing and ability.

We checked data-set identifiers against a separate IRW item-key alignment record, matching 88 of the 145 data sets. Two matched data sets had accepted item flips in that analysis and contributed 41 difficulty rows (.67\%). The response matrices and fits differed, so this does not establish that all 41 estimates were incorrectly keyed.

Excluding those rows changes difficulty SD from 1.6544 to 1.6467. The simulations use the full pool; the archived exclusion snapshot has 6,102 rows and has not been used to rerun them. This check addresses marginal difficulty spread and leaves 57 unmatched data sets unaudited.

\subsection{Joint item parameters}
\label{app:sf04}

We generate discrimination conditionally on difficulty rank. The Gaussian-copula construction maps difficulty ranks to normal scores, combines them with independent normal noise at dependence parameter $-.3$, and transforms the result to log-normal discriminations \parencite{nelsen_introduction_2006}. The resulting Pearson correlation depends on the marginals and finite sample; it need not equal~$-.3$.

The IRW pool contains no discriminations with which to assess this assumption. Fitted difficulty--discrimination associations can also differ between 2PL and 3PL models \parencite{sweeney_investigation_2022}. The dependence parameter is therefore a design choice. Calibration preserves Pearson and rank associations within a form because it applies one positive multiplier to all discriminations; item plausibility still depends on the initial distributions and multiplier bounds.

%% file: floats/tab_latent_shapes.tex
\begin{table}[htbp]\centering
\caption{Moments and percentiles from 200,000 draws of each ability distribution}
\label{tab:latent-shapes}
{\small\begin{tabular}{lrrrrr}
\toprule
Shape & SD & Skewness & Kurtosis & 1st pct. & 99th pct. \\
\midrule
Normal & $1.001$ & $-0.007$ & $2.980$ & $-2.33$ & $2.31$ \\
Skewed & $0.998$ & $0.980$ & $4.374$ & $-1.59$ & $3.01$ \\
Bimodal & $1.001$ & $-0.003$ & $2.178$ & $-2.03$ & $2.03$ \\
Heavy-Tailed & $0.999$ & $0.067$ & $7.493$ & $-2.62$ & $2.58$ \\
\bottomrule
\end{tabular}}
\floatnote{All generators have population variance one. The heavy-tailed sample kurtosis is 7.493; population kurtosis is 9. The package bimodal distribution and the separate theory quadrature use different component locations (\cref{app:shapes}).}
\end{table}

%% file: floats/tab_shape_realism.tex
\begin{table}[htbp]\centering
\caption{Design shapes compared with 504 fitted IRW ability distributions}\label{tab:shape-realism}
{\small\begin{tabular}{lrrlll}
\toprule
 & Skewness & Excess kurtosis & & & \\
\midrule
\multicolumn{6}{l}{\itshape This design} \\
\quad Normal & $-0.007$ & $-0.020$ & & & \\
\quad Skewed & $0.980$ & $1.374$ & & & \\
\quad Bimodal & $-0.003$ & $-0.822$ & & & \\
\quad Heavy-Tailed & $0.067$ & $4.493$ & & & \\
\addlinespace
\multicolumn{6}{l}{\itshape 504 IRW instruments, fitted latent densities} \\
\quad 10th percentile & $-1.545$ & $-0.122$ & & & \\
\quad Median & $0.178$ & $3.584$ & & & \\
\quad 90th percentile & $1.828$ & $14.976$ & & & \\
\bottomrule
\end{tabular}}
\floatnote{Design rows use 200,000 generator draws; IRW rows give quantiles across fitted distributions. Both report excess kurtosis. The heavy-tailed sample value is 4.493, compared with population value 6; 212 of 504 fitted distributions (42.1\%) exceed the sample value. IRW densities depend on the fitted response model and grid.}
\end{table}

%% file: appendices/G_atlas.tex
\section{A Sweep of 25 Item-Generation Schemes}
\label{app:atlas}

We compare the information produced by 25 item-generation schemes, then examine what changes when we calibrate those schemes to common targets. The comparison shows how familiar item-generation choices can imply different measurement conditions, even at the same test length.

\subsection{Information before calibration}
\label{app:sg01}

We crossed 25 item-generation schemes with four latent shapes and two difficulty sources, giving 200 cells with 500 forms each. We evaluated the mean-information index at $c=1$ from the generating parameters, without fitting response data. Three schemes follow specific simulation designs: the benchmark item-level treatment-effect simulation, our initial validation baseline, and an archived Bayesian IRT design. Six use baselines or package defaults. Sixteen vary length, fixed discrimination, difficulty spread or guessing one at a time around a 20-item 2PL form. This is a designed sensitivity comparison, not a survey of published simulations.

Information varied across schemes (\cref{fig:atlas}). Averaging each scheme's eight cell medians gave values from .515 for a five-item form to .930 for a 60-item form. Among eight 20-item 2PL schemes that differed only in discrimination or difficulty specification, the range was .781 to .835. Thus, test length and model name alone did not specify the information level.

We also regressed the index log-odds on the design factors. Length and discrimination specification accounted for .964 of the explained sequential sum of squares in the parametric arm. Their share was .894 when we pooled both difficulty sources and used realized difficulty spread (\cref{fig:variance-decomposition}). These shares describe this grid and depend on factor order and model specification. They are not causal contributions or general rankings of the factors.

\begin{figure}[H]
\centering
\includegraphics[width=\textwidth]{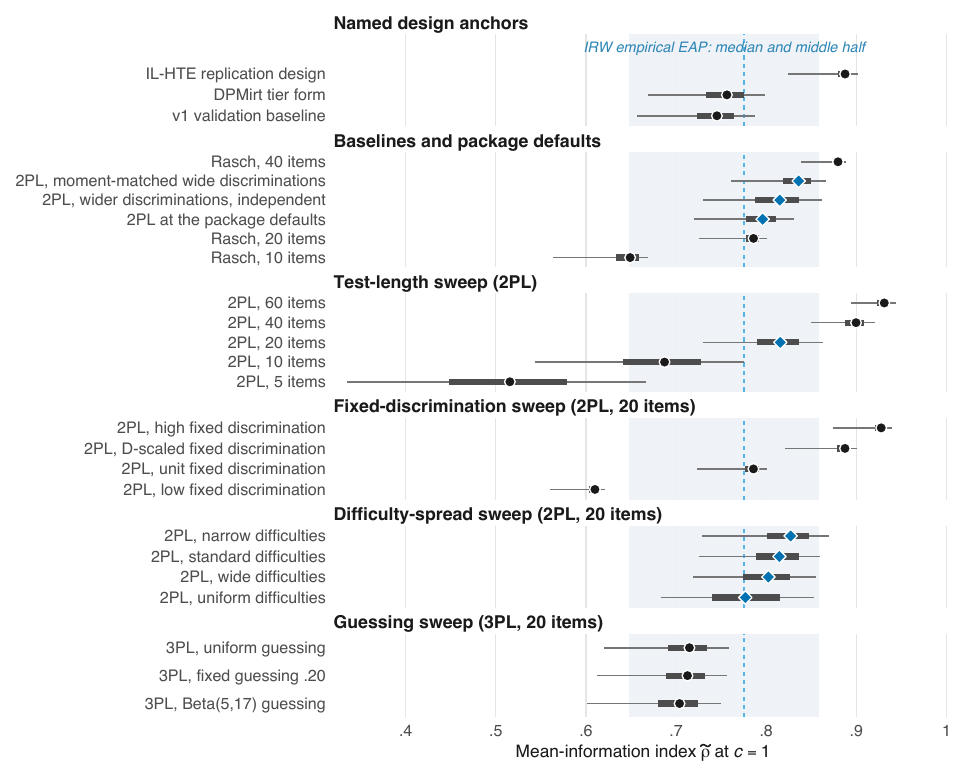}
\caption{The mean-information index of 25 generating schemes at their generated discriminations}
\label{fig:atlas}

\smallskip
\begin{minipage}{0.95\textwidth}
\footnotesize
\emph{Note.} Each row pools 4,000 forms across four ability shapes and two difficulty sources at $c=1$. Points are pooled medians; thick and thin bars show interquartile and 5th--95th percentile ranges. They differ from the text summaries, which average eight cell medians. Blue diamonds identify the eight comparable 20-item 2PL schemes. The shaded band and dashed line show the interquartile range and median of fitted EAP reliabilities from 223 dichotomous IRW units. These scored-data coefficients provide context and are distinct from information indices.
\end{minipage}
\end{figure}

\begin{figure}[H]
\centering
\includegraphics[width=\textwidth]{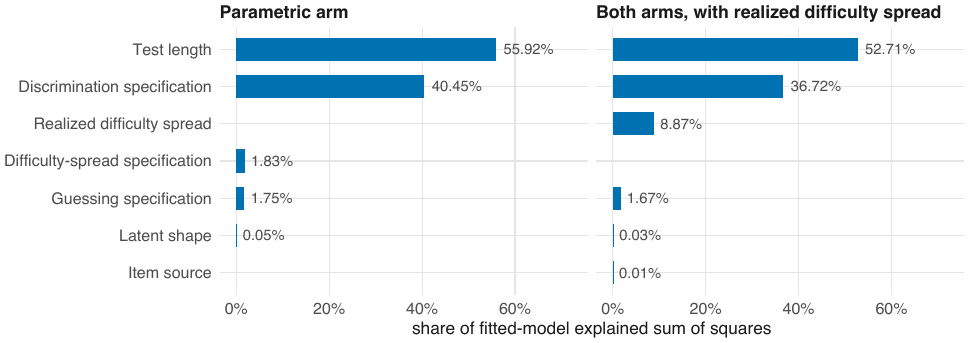}
\caption{Where the variation in the scheme sweep comes from}
\label{fig:variance-decomposition}

\smallskip
\begin{minipage}{0.95\textwidth}
\footnotesize
\emph{Note.} Bars give sequential sums of squares as shares of explained variation in the index log-odds. The left model uses 50,000 parametric forms; the right uses all 100,000 forms and realized difficulty SD. Blank positions indicate factors absent from the model. Shares depend on model specification and factor order.
\end{minipage}
\end{figure}

\subsection{A comparison with scored instruments}
\label{app:sg02}

For context, we summarized EAP reliability in the publicly released analysis data of \textcite{lee_irw_reliability_2026}: 223 dichotomous IRW units with median .775 and a broader set of 879 units with median .859. Using the sweep's pointwise index, 8.5\% of the dichotomous corpus and 3.2\% of the full corpus fell below the sweep minimum. Both were below our prespecified expectation of at least 20\%; three of five stored threshold checks were unmet.

The empirical and simulated quantities overlap, but they describe different instruments and coefficients. This comparison does not establish that every scheme is realistic or that an information index equals EAP reliability.

\subsection{Equalizing information across test lengths}
\label{app:sg03}

We next calibrated the 200 cells to .70 and .85. All 400 requests were reachable within the chosen interval. At .70, the median of the cell-median discriminations was 1.745 for five-item forms and .339 for 60-item forms; at .85 it was 3.761 and .568 (\cref{fig:equalization}). Shorter forms needed much more discriminating items to reach the same information target.

Of the 400 cell medians, 25 exceeded the study's diagnostic upper bound of 2.50, with a maximum near 8.90; none fell below .30. The exceedances occurred among shorter forms at demanding targets. The band is a study-specific check, and the comparison combines different forms and schemes. After calibration, researchers should check whether the item parameters are plausible for their application. If the items are unsuitable, the length, generating distribution or target needs to change.

\begin{figure}[H]
\centering
\includegraphics[width=0.9\textwidth]{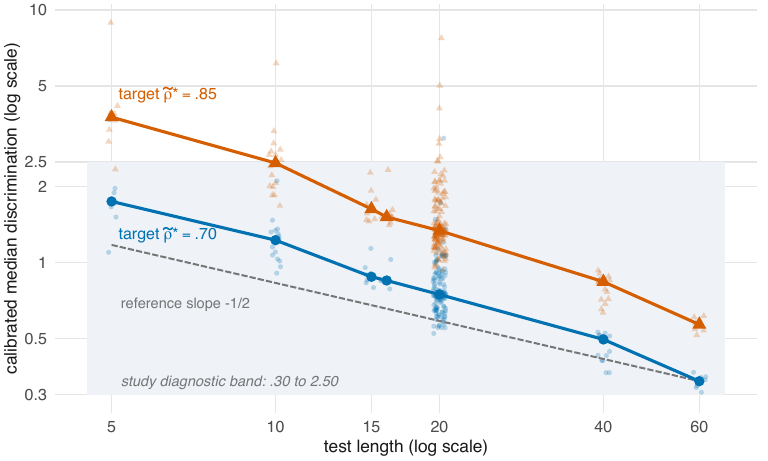}
\caption{Calibrated discrimination across the scheme sweep at two targets}
\label{fig:equalization}

\smallskip
\begin{minipage}{0.95\textwidth}
\footnotesize
\emph{Note.} Small points are the median calibrated discrimination within each scheme-by-shape-by-source-by-target cell (400 cells, all reachable); large points join the medians of these cell summaries at each length. Both axes are logarithmic. The shaded band marks the study's diagnostic range of .30 to 2.50, and the dashed line is an inverse-square-root reference anchored at the .70 summary for 60 items, drawn for comparison (not fitted).
\end{minipage}
\end{figure}

\subsection{Separating information from other design factors}
\label{app:sg04}

Metamodels relate method performance to simulation design factors \parencite{gilbert_multilevel_2026}. If those factors already determine the information index, its separate association with performance is difficult to estimate. In the 32 uncalibrated DIF cells, regressing the realized reference index on the design factors gave $R^2=.987$ and variance inflation factor 77.5. Adding the targeted arms reduced these to .019 and 1.02 (\cref{fig:confounding}).

Targeting supplies information levels that are less closely tied to the original factors. It does so by changing item parameters, so the resulting contrasts compare complete calibrated designs. They do not identify a causal effect of the index alone.

\begin{figure}[H]
\centering
\includegraphics[width=\textwidth]{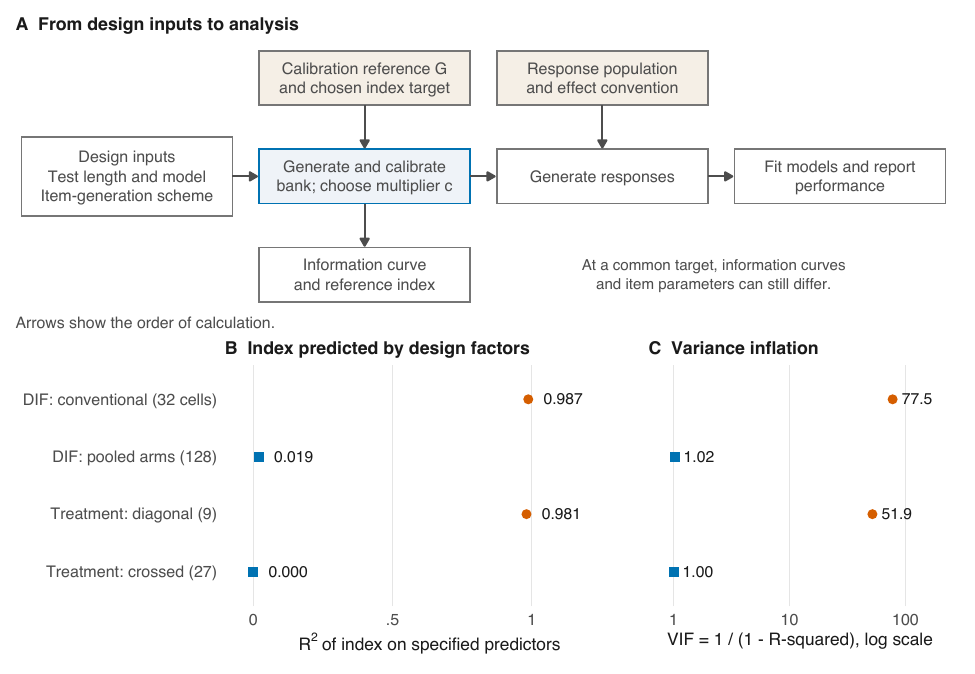}
\caption{The design workflow and the collinearity of the index with the design factors}
\label{fig:confounding}

\smallskip
\begin{minipage}{0.95\textwidth}
\footnotesize
\emph{Note.} A shows the steps from generating items to calibrating information and fitting the analysis model. B and C show $R^2$ and variance inflation factors from regressing the information index on design factors. The DIF comparison uses 32 uncalibrated cells and 128 rows after adding the three targets. The item-level treatment-effect extension compares nine diagonal cells with the 27-cell crossed grid. These are design diagnostics; the arrows do not represent causal identification.
\end{minipage}
\end{figure}

%% file: appendices/H_dif.tex
\section{Differential Item Functioning and Score Comparisons Under Targeting}
\label{app:dif}
\label{app:scores}

We use differential item functioning (DIF) detection to examine how calibration changes comparisons among simulation conditions. We compare Mantel--Haenszel detection under uncalibrated forms and forms calibrated to three targets. We also vary the effect scale, matching procedure and group ability means. A final experiment examines score-mean comparisons when groups differ in precision.

\subsection{Design}
\label{app:sh01}

We cross 2PL and 3PL items (guessing .20), parametric and Item Response Warehouse difficulties, normal and skewed calibration references, 20 and 40 items, and 500 and 1,000 persons per group. Each of the 32 cells contains 100 forms. We evaluate each form at $c=1$ and at targets .60, .75 and .90, sharing its baseline parameters, latent deviates and response uniforms across the four calibration arms. Parametric difficulties are standard-normal draws centered within form; empirical difficulties come from the frozen pool. The rank-based copula supplies discriminations (\cref{app:sf04}).

Calibration uses either a standard normal reference or $(Y-4)/2$ with $Y\sim\operatorname{Gamma}(4,1)$. For each structural cell, we draw 4,000 calibration nodes and 20,000 evaluation nodes and fix reference variance at one. We solve $\rhotilde(c)-\rhotilde^*=0$ on $[.05,10]$ with tolerance $10^{-6}$. An arm is unavailable if the endpoints do not bracket a root. Responses use normal abilities in every cell, so a skewed-reference target applies to the calibration population rather than the response population.

We plant DIF in 20\% of items at evenly spaced interior positions among the ordered difficulties: four items in a 20-item form and eight in a 40-item form. The baseline adds $\Delta b_i=.5$ to focal-group difficulty. Its linear-predictor contrast is $-ca_i\Delta b_i$ and therefore changes with calibration. An alternative sets $\Delta b_i=.5/(ca_i)$, fixing that contrast at $-.5$. This also fixes the conditional log-odds contrast in the 2PL, but not in the 3PL. Each form additionally generates a fully DIF-free condition.

Two other variants retain the baseline effect: two-stage purification of the matching set, and a focal ability mean of $+.5$ (group impact). We examine them separately. The continuity-corrected Mantel--Haenszel test uses rest-score bands generally three points wide, omits uninformative strata, and applies an itemwise .05 threshold without multiplicity adjustment. A successful test returns a finite $p$ value without error.

We distinguish three outcomes: detection of planted DIF items, flagging of unshifted items on forms containing DIF, and rejection on wholly DIF-free forms. The second can be affected by other items' DIF through the matching score; it is not the same null condition as the third.

\subsection{Rates, contrasts and uncertainty}
\label{app:sh02}

We treat the form as the replication unit, giving 3,200 clusters shared across calibration arms and variants. For form $f$ and arm $a$, let $R_{fa}$, $S_{fa}$ and $A_{fa}$ be rejections, successful tests and attempted tests. Arm-level rates pool counts, $\hat p_a=\sum_fR_{fa}/\sum_fS_{fa}$. A sensitivity analysis divides by attempted tests instead, counting failures as non-rejections. Unavailable arms contribute neither rejections nor tests.

Factor contrasts instead give equal weight to forms. We retain forms with finite rates $r_{fa}=R_{fa}/S_{fa}$ in both the uncalibrated and target-.75 arms and average within each factor level. The change in a contrast is
\begin{equation}
\label{eq:app-did}
\big(\bar r_{T,H}-\bar r_{T,L}\big)-\big(\bar r_{U,H}-\bar r_{U,L}\big),
\end{equation}
where $T$ denotes targeted and $U$ uncalibrated. The factor contrasts $H-L$ are 40 minus 20 items, 3PL minus 2PL, 1,000 minus 500 persons, empirical minus parametric difficulties, and skewed minus normal reference.

Each principal comparison uses 3,200 paired forms, 100 per cell. We obtain percentile intervals from 2,000 bootstrap samples of whole forms within cells, carrying all arms and variants together. This preserves the pairing; intervals remain conditional on the selected grid and shared calibration nodes.

Across variants, arms and item kinds, 3,057,047 of 3,069,600 attempted tests succeeded and 12,553 failed. All results and intervals below use the 100-form run. Earlier runs with 500 or 1,000 forms per cell remain in the accompanying records.

\subsection{Detection and flagging under the baseline effect}
\label{app:sh03}

Under the fixed difficulty shift, pooled detection rose from .429 at target .60 to .620 at .75 and .869 at .90 (\cref{fig:dif-rates}). The uncalibrated rate was .762 at a median realized index of .813. Because calibration increases discriminations, this convention also increases the planted linear-predictor effect.

Fully DIF-free forms with equal group ability means had rejection rates of .042 to .043 across arms, with intervals below .05. Unshifted items on forms containing DIF were flagged more often: .049, .065 and .138 at the three targets, compared with .087 without calibration. At target .90, almost 14\% were flagged (interval .136 to .140). This difference points to the importance of the matching set, not just the information target.

Two-stage purification reduced unshifted-item flagging to .043, .045 and .046 at the three targets, compared with .047 without calibration. At target .90, planted-item detection was .911 with purification and .869 at baseline. The comparison suggests that improving the matching set can address a problem that remains when information alone is changed.

\begin{figure}[!htbp]
\centering
\includegraphics[width=\textwidth]{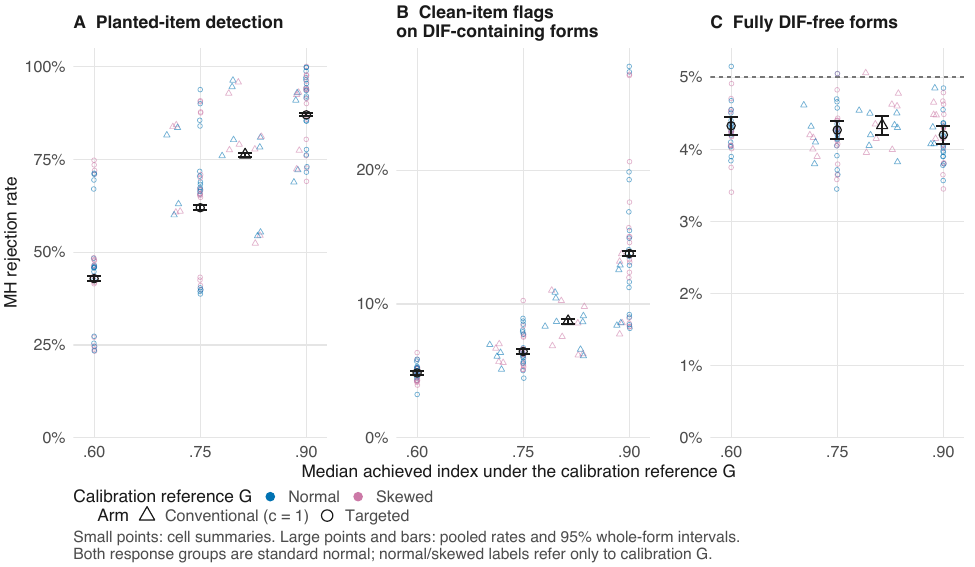}
\caption{Mantel--Haenszel rates by arm for the three item kinds}
\label{fig:dif-rates}

\smallskip
\begin{minipage}{0.95\textwidth}
\footnotesize
\emph{Note.} The baseline fixes focal difficulty shifts at $+.5$ and uses unpurified matching, itemwise .05 tests and standard-normal abilities in both groups. Color identifies the calibration reference. Small points summarize 100 forms per cell and arm; black symbols pool rejections over successful tests (triangles: $c=1$; circles: targets). Bars are 95\% intervals from 2,000 within-cell whole-form bootstrap samples. The horizontal coordinate is median achieved reference information. Fifteen of 3,200 forms could not reach .90 under the endpoint rule. Vertical scales differ; the dashed line marks .05.
\end{minipage}
\end{figure}

\subsection{How targeting changes comparisons among design factors}
\label{app:sh03-contrasts}

The direction of a comparison can depend on the scale on which DIF is held fixed (\cref{fig:main-dif,tab:s5-did}). With a fixed difficulty shift, the length contrast (40 minus 20 items) was $-.047$ without calibration and $-.250$ at target .75. The paired change was $-.203$, with interval $[-.213,-.193]$. Longer forms require smaller discriminations at a common target, which can reduce detection of a fixed difficulty shift.

With a fixed linear-predictor contrast, the same length comparison was $-.042$ without calibration and $+.043$ at target .75, a change of $+.086$ with interval $[.077,.094]$. The model contrast (3PL minus 2PL) changed by $+.120$ under the fixed difficulty shift and $-.032$ under the fixed predictor. Thus, information targets need to be interpreted together with the specified effect scale.

Under the baseline convention, three of five factor contrasts changed rank and none changed sign at target .75. \Cref{tab:s5c-contrast} gives the length and model contrasts for all four variants; \cref{tab:s5c-magnitude} shows the corresponding effect magnitudes. The contrasts compare complete response designs and cannot separate the effect of information from the item-parameter changes used to obtain it.

\begin{figure}[!htbp]
\centering
\includegraphics[width=\textwidth]{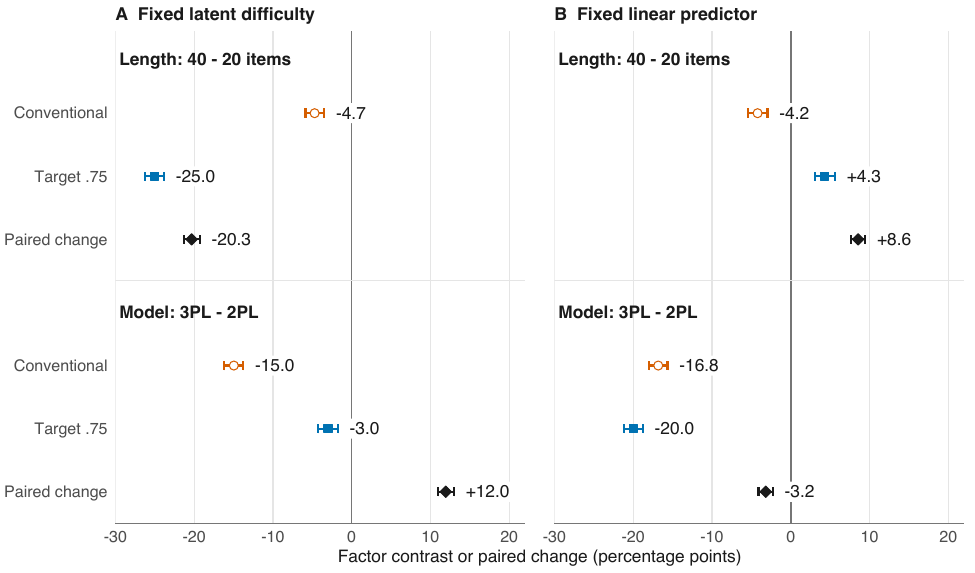}
\caption{The test-length and model contrasts under two effect conventions, conventional and targeted}
\label{fig:main-dif}

\smallskip
\begin{minipage}{0.95\textwidth}
\footnotesize
\emph{Note.} A fixes the focal difficulty shift at .5. B fixes the focal-minus-reference linear-predictor contrast at $-.5$ (magnitude .5); this is a constant log-odds contrast only in the 2PL. Contrasts compare Mantel--Haenszel detection rates with equal form weights. Calibration targets .75. Points and 95\% intervals use 3,200 paired forms and 2,000 whole-form bootstrap samples within the 32 cells, shared across arms and conventions. Paired change is targeted minus uncalibrated; units are percentage points.
\end{minipage}
\end{figure}

\input{floats/tab_s5_did}

\input{floats/tab_s5c_contrast}

\input{floats/tab_s5c_magnitude}

\subsection{Effect scale, group impact and failed tests}
\label{app:sh04}

Holding the linear-predictor effect fixed reversed the baseline trend in detection: rates decreased from .877 to .814 to .634 across targets, compared with .772 without calibration. Unshifted-item flagging was .067, .080 and .082, compared with .086 without calibration. These rates remained above rejection on fully DIF-free forms.

When the focal ability mean was $+.5$, rejection on fully DIF-free forms was .148, .134 and .098 across targets and .119 without calibration. All interval lower limits exceeded .05, including $[.146,.150]$ at target .60 and $[.096,.099]$ at .90. Unshifted-item flagging reached .325 at .90. Matching reference information therefore did not restore nominal rejection under group impact. This variant used unpurified matching; it does not assess purification in combination with impact.

\Cref{fig:conventions} compares all four variants. Counting failed tests as non-rejections lowered pooled rates by at most .0041 and preserved the sign of every change in factor contrast. The largest change was unshifted-item flagging under impact at .90: .325 to .321, based on 24,546 rejections, 75,601 successful tests and 76,560 attempts. Target .90 was unavailable for 15 of 3,200 forms under the endpoint-bracket rule; all other arms were available.

\begin{figure}[!htbp]
\centering
\includegraphics[width=\textwidth]{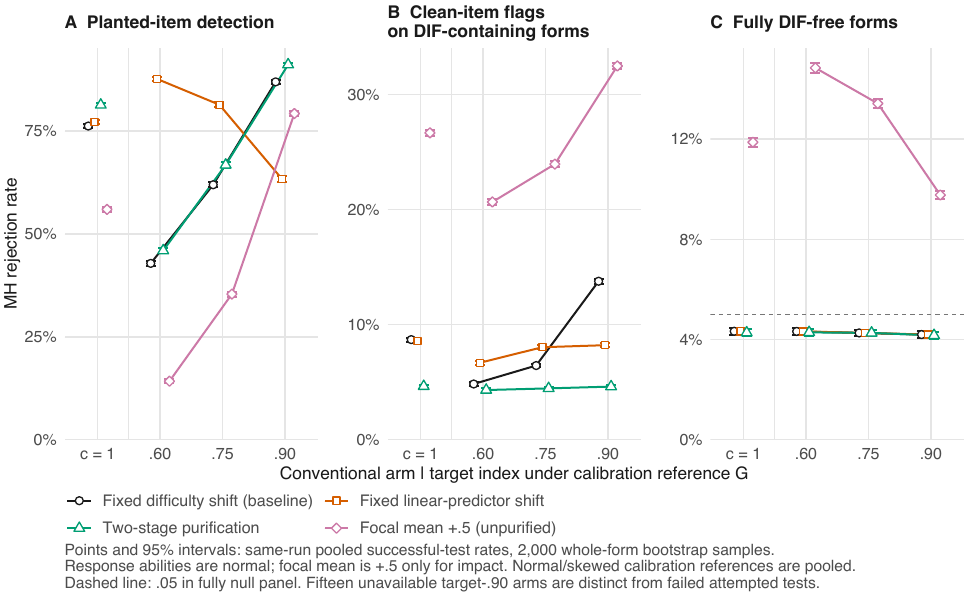}
\caption{The four DIF variants: detection, clean-item flagging and fully null rejection by arm}
\label{fig:conventions}

\smallskip
\begin{minipage}{0.95\textwidth}
\footnotesize
\emph{Note.} Points pool rejections over successful tests; bars are 95\% intervals from 2,000 within-cell whole-form bootstrap samples. Lines connect targets .60, .75 and .90. Baseline fixes difficulty shifts at $+.5$; fixed-predictor fixes the linear-predictor contrast at $-.5$; purification changes matching only; impact shifts the focal ability mean to $+.5$ with baseline DIF and unpurified matching. Variants are separate. Vertical scales differ; the dashed line marks .05.
\end{minipage}
\end{figure}

\subsection{A score comparison under unequal precision}
\label{app:sh05}

We also examined score-mean comparisons when groups differ in precision. Reference abilities are $N(0,1)$; focal abilities have mean 0 or .2 and SD 1 or .6, with 500 persons per group. For each of the four conditions, we generate one 20-item 2PL form, calibrate the reference index to .80 and keep the form fixed for 200 response replications.

We compare sum scores and EAP and WLE estimates from a pooled TAM fit without a group indicator, using two-sided Welch tests at .05. With equal latent means and unequal dispersion, the reference-minus-focal index gap was .187. Across scores and both dispersion conditions, rejection ranged from .035 to .045, with Monte Carlo standard errors about .013 to .015 (\cref{tab:main-null}). With latent mean difference .2, the largest absolute standardized bias was .025, for WLE under unequal dispersion (\cref{tab:s5b-scores}).

These low rejection rates apply to four fixed forms and have limited Monte Carlo precision. Equal latent means with unequal dispersion need not imply equal score means, so they cannot establish test size for a score-mean null or isolate calibration's effect. Score comparisons also depend on measurement invariance and calibration choices \parencite{soland_how_2024}; finite-sample measurement error need not yield simple attenuation \parencite{loken_measurement_2017}.

\input{floats/tab_main_null}

\input{floats/tab_s5b_scores}

%% file: floats/tab_s5_did.tex
\begin{table}[htbp]\centering
\caption{Changes in DIF detection contrasts after calibration to .75}
\label{tab:s5-did}
{\scriptsize\setlength{\tabcolsep}{2pt}\begin{tabular}{>{\raggedright\arraybackslash}p{0.20\textwidth}>{\centering\arraybackslash}p{0.245\textwidth}>{\centering\arraybackslash}p{0.245\textwidth}>{\centering\arraybackslash}p{0.245\textwidth}}
\toprule
Factor & Uncontrolled [95\% CI] & Targeted [95\% CI] & Paired difference [95\% CI] \\
\midrule
Test length, 40 vs 20 & $-0.047\;[-0.058,\,-0.035]$ & $-0.250\;[-0.263,\,-0.238]$ & $-0.203\;[-0.213,\,-0.193]$ \\
Model, 3PL vs 2PL & $-0.150\;[-0.162,\,-0.138]$ & $-0.030\;[-0.043,\,-0.018]$ & $0.120\;[0.109,\,0.130]$ \\
Group size, 1000 vs 500 & $0.209\;[0.197,\,0.220]$ & $0.230\;[0.218,\,0.242]$ & $0.021\;[0.011,\,0.030]$ \\
Item pool, IRW vs parametric & $-0.022\;[-0.034,\,-0.011]$ & $-0.003\;[-0.014,\,0.010]$ & $0.020\;[0.010,\,0.030]$ \\
Calibration $G$, skewed vs normal & $0.004\;[-0.007,\,0.016]$ & $0.001\;[-0.011,\,0.013]$ & $-0.003\;[-0.013,\,0.006]$ \\
\bottomrule
\end{tabular}}
\floatnote{The baseline fixes the focal difficulty shift at .5. Contrasts weight forms equally across 100 paired forms in each of 32 cells. Intervals are 95\% percentile intervals from 2,000 within-cell whole-form bootstrap samples shared across arms. The last column is targeted minus uncalibrated contrast in probability units. The reference-distribution factor changes calibration only; generated abilities remain normal.}
\end{table}

%% file: floats/tab_s5c_contrast.tex
\begin{table}[htbp]\centering
\caption{Length and model contrasts under four DIF variants at target .75}
\label{tab:s5c-contrast}
{\scriptsize\setlength{\tabcolsep}{2pt}\begin{tabular}{>{\raggedright\arraybackslash}p{0.17\textwidth}>{\raggedright\arraybackslash}p{0.16\textwidth}>{\centering\arraybackslash}p{0.195\textwidth}>{\centering\arraybackslash}p{0.195\textwidth}>{\centering\arraybackslash}p{0.195\textwidth}}
\toprule
Factor & Convention & Uncontr. [95\% CI] & Targeted [95\% CI] & Paired difference [95\% CI] \\
\midrule
Test length, 40 vs 20 & Baseline (latent scale) & \shortstack{$-0.047$\\[-0.2ex]{\tiny $[-0.058,\,-0.035]$}} & \shortstack{$-0.250$\\[-0.2ex]{\tiny $[-0.263,\,-0.238]$}} & \shortstack{$-0.203$\\[-0.2ex]{\tiny $[-0.213,\,-0.193]$}} \\
Test length, 40 vs 20 & Linear predictor fixed & \shortstack{$-0.042$\\[-0.2ex]{\tiny $[-0.054,\,-0.029]$}} & \shortstack{$0.043$\\[-0.2ex]{\tiny $[0.031,\,0.056]$}} & \shortstack{$0.086$\\[-0.2ex]{\tiny $[0.077,\,0.094]$}} \\
Test length, 40 vs 20 & Purified matching & \shortstack{$-0.029$\\[-0.2ex]{\tiny $[-0.041,\,-0.017]$}} & \shortstack{$-0.243$\\[-0.2ex]{\tiny $[-0.256,\,-0.231]$}} & \shortstack{$-0.214$\\[-0.2ex]{\tiny $[-0.225,\,-0.204]$}} \\
Test length, 40 vs 20 & With impact & \shortstack{$0.099$\\[-0.2ex]{\tiny $[0.085,\,0.114]$}} & \shortstack{$-0.175$\\[-0.2ex]{\tiny $[-0.188,\,-0.160]$}} & \shortstack{$-0.274$\\[-0.2ex]{\tiny $[-0.287,\,-0.261]$}} \\
Model, 3PL vs 2PL & Baseline (latent scale) & \shortstack{$-0.150$\\[-0.2ex]{\tiny $[-0.162,\,-0.138]$}} & \shortstack{$-0.030$\\[-0.2ex]{\tiny $[-0.043,\,-0.018]$}} & \shortstack{$0.120$\\[-0.2ex]{\tiny $[0.109,\,0.130]$}} \\
Model, 3PL vs 2PL & Linear predictor fixed & \shortstack{$-0.168$\\[-0.2ex]{\tiny $[-0.180,\,-0.156]$}} & \shortstack{$-0.200$\\[-0.2ex]{\tiny $[-0.212,\,-0.188]$}} & \shortstack{$-0.032$\\[-0.2ex]{\tiny $[-0.041,\,-0.023]$}} \\
Model, 3PL vs 2PL & Purified matching & \shortstack{$-0.146$\\[-0.2ex]{\tiny $[-0.158,\,-0.134]$}} & \shortstack{$-0.029$\\[-0.2ex]{\tiny $[-0.042,\,-0.016]$}} & \shortstack{$0.117$\\[-0.2ex]{\tiny $[0.107,\,0.128]$}} \\
Model, 3PL vs 2PL & With impact & \shortstack{$-0.251$\\[-0.2ex]{\tiny $[-0.266,\,-0.236]$}} & \shortstack{$-0.036$\\[-0.2ex]{\tiny $[-0.049,\,-0.022]$}} & \shortstack{$0.215$\\[-0.2ex]{\tiny $[0.203,\,0.228]$}} \\
\bottomrule
\end{tabular}}
\floatnote{Contrasts weight forms equally; paired differences are targeted minus uncalibrated, in probability units. Each variant uses 100 paired forms per cell across 32 cells. The 95\% intervals use 2,000 within-cell whole-form bootstrap samples shared across arms and variants. The four variants change the effect scale, matching or group means separately.}
\end{table}

%% file: floats/tab_s5c_magnitude.tex
\begin{table}[htbp]\centering
\caption{Effect magnitudes under the two DIF conventions}
\label{tab:s5c-magnitude}
{\small\begin{tabular}{lllrrr}
\toprule
Convention & Model & Arm & Magnitude & Log odds & ETS delta \\
\midrule
Latent scale fixed & 2pl & target 060 & $0.240$ & $0.240$ & $0.56$ \\
Latent scale fixed & 2pl & target 075 & $0.361$ & $0.361$ & $0.85$ \\
Latent scale fixed & 2pl & target 090 & $0.754$ & $0.754$ & $1.77$ \\
Latent scale fixed & 2pl & uncontrolled & $0.516$ & $0.516$ & $1.21$ \\
Latent scale fixed & 3pl\_g20 & target 060 & $0.306$ & --- & --- \\
Latent scale fixed & 3pl\_g20 & target 075 & $0.479$ & --- & --- \\
Latent scale fixed & 3pl\_g20 & target 090 & $1.106$ & --- & --- \\
Latent scale fixed & 3pl\_g20 & uncontrolled & $0.517$ & --- & --- \\
Linear predictor fixed & 2pl & target 060 & $0.500$ & $0.500$ & $1.18$ \\
Linear predictor fixed & 2pl & target 075 & $0.500$ & $0.500$ & $1.18$ \\
Linear predictor fixed & 2pl & target 090 & $0.500$ & $0.500$ & $1.18$ \\
Linear predictor fixed & 2pl & uncontrolled & $0.500$ & $0.500$ & $1.18$ \\
Linear predictor fixed & 3pl\_g20 & target 060 & $0.500$ & --- & --- \\
Linear predictor fixed & 3pl\_g20 & target 075 & $0.500$ & --- & --- \\
Linear predictor fixed & 3pl\_g20 & target 090 & $0.500$ & --- & --- \\
Linear predictor fixed & 3pl\_g20 & uncontrolled & $0.500$ & --- & --- \\
\bottomrule
\end{tabular}}
\floatnote{Entries are median absolute magnitudes. Fixing the difficulty shift at .5 lets the predictor shift $ca_i|\Delta b_i|$ vary; setting $|\Delta b_i|=.5/(ca_i)$ fixes it at .5. A predictor shift equals a conditional log-odds shift only for the 2PL, so ETS delta is reported only for that model.}
\end{table}

%% file: floats/tab_main_null.tex
\begin{table}[tbp]\centering
\caption{Score-test rejection when latent means are equal}
\label{tab:main-null}
{\small\begin{tabular}{llrrr}
\toprule
Focal SD & Score & Index gap & Rejection & MCSE \\
\midrule
$1.0$ & Sum & $-0.0004$ & $0.035$ & $0.013$ \\
$1.0$ & EAP & $-0.0004$ & $0.040$ & $0.014$ \\
$1.0$ & WLE & $-0.0004$ & $0.045$ & $0.015$ \\
\addlinespace
$0.6$ & Sum & $0.187$ & $0.040$ & $0.014$ \\
$0.6$ & EAP & $0.187$ & $0.035$ & $0.013$ \\
$0.6$ & WLE & $0.187$ & $0.045$ & $0.015$ \\
\addlinespace
\bottomrule
\end{tabular}}
\floatnote{Each condition uses one 20-item 2PL form, held fixed for 200 replications with 500 persons per group. Reference ability is $N(0,1)$; focal ability has SD 1 or .6. The index gap is reference minus focal. Tests compare score means using two-sided Welch tests at .05. Equal latent means need not imply equal score means when dispersions differ. Monte Carlo standard errors exclude variation across forms.}
\end{table}

%% file: floats/tab_s5b_scores.tex
\begin{table}[htbp]\centering
\caption{Score-mean comparisons under equal and unequal group precision}\label{tab:s5b-scores}
{\footnotesize\setlength{\tabcolsep}{4.5pt}\begin{tabular}{llrrrrr}
\toprule
Groups & Score & True $\delta$ & \shortstack{Information\\index gap} & Rejection & MCSE & Bias of $\hat d$ \\
\midrule
Balanced & Posterior mean & $0.0$ & $-4.43\times 10^{-4}$ & $0.040$ & $0.014$ & $-0.0022$ \\
Balanced & Sum score & $0.0$ & $-4.43\times 10^{-4}$ & $0.035$ & $0.013$ & $-0.0020$ \\
Balanced & Weighted likelihood & $0.0$ & $-4.43\times 10^{-4}$ & $0.045$ & $0.015$ & $-0.0017$ \\
Imbalanced & Posterior mean & $0.0$ & $0.187$ & $0.035$ & $0.013$ & $-0.0025$ \\
Imbalanced & Sum score & $0.0$ & $0.187$ & $0.040$ & $0.014$ & $-0.0041$ \\
Imbalanced & Weighted likelihood & $0.0$ & $0.187$ & $0.045$ & $0.015$ & $7.76\times 10^{-4}$ \\
Balanced & Posterior mean & $0.2$ & $6.30\times 10^{-4}$ & $0.840$ & $0.026$ & $-0.0152$ \\
Balanced & Sum score & $0.2$ & $6.30\times 10^{-4}$ & $0.830$ & $0.027$ & $-0.0167$ \\
Balanced & Weighted likelihood & $0.2$ & $6.30\times 10^{-4}$ & $0.835$ & $0.026$ & $-0.0155$ \\
Imbalanced & Posterior mean & $0.2$ & $0.193$ & $0.905$ & $0.021$ & $-0.0112$ \\
Imbalanced & Sum score & $0.2$ & $0.193$ & $0.915$ & $0.020$ & $-0.0021$ \\
Imbalanced & Weighted likelihood & $0.2$ & $0.193$ & $0.870$ & $0.024$ & $-0.0251$ \\
\bottomrule
\end{tabular}}
\floatnote{Each of four 20-item 2PL forms is fixed for 200 replications, with 500 persons per group. Reference ability is $N(0,1)$; focal ability has mean $\delta$ and SD 1 or .6. A pooled TAM fit supplies EAP and WLE scores. Rejection uses two-sided Welch tests at .05. The index gap is reference minus focal. Bias subtracts latent $\delta$ from the score difference standardized by reference-group score SD. MCSE describes binomial uncertainty in rejection, conditional on the form.}
\end{table}